\documentclass[12pt]{article}
\usepackage[margin=1in]{geometry} % Clean 1-inch margins
\usepackage{parskip}

\title{\large \bf AOC-CBS: Anytime-Optimal Continuous-time Conflict-Based Search for Generalised Multi-Agent Path Finding}
\author{Alvin~Combrink$^{1}$, Sabino~Francesco~Roselli$^{1}$, Martin~Fabian$^{1}$}
\date{}

\usepackage[backend=biber, style=numeric, sorting=none, maxbibnames=99]{biblatex}

\usepackage{comment}
\usepackage{amssymb,amsfonts}
\usepackage{algorithm}
\usepackage[noend]{algpseudocode}
\usepackage{graphicx}
\usepackage{textcomp}
\usepackage{booktabs}
\usepackage{hyperref}
\usepackage{xspace}
\usepackage{balance}
\usepackage{array}
\usepackage{caption}
\usepackage{xcolor}
\usepackage{mathtools}
\usepackage{subcaption}

\usepackage{amsthm}
\usepackage{enumitem} 
\newtheorem{theorem}{Theorem}[section]

\newtheorem{corollary}[theorem]{Corollary}

\begin{document}
\maketitle
\renewcommand{\thefootnote}{\arabic{footnote}}
\footnotetext[1]{Department of Electrical Engineering, Chalmers University of Technology, Gothenburg, Sweden. {\tt\small \{combrink, rsabino, fabian\}@chalmers.se}}
\thispagestyle{empty}
\pagestyle{empty}

% Personal comments
\newcommand{\AC}[1]{{#1}\xspace}

% mathematical commands
\newcommand{\Real}{\mathbb{R}}
\newcommand{\Natural}{\mathbb{N}}
\newcommand{\tuple}[1]{\langle #1 \rangle}
\newcommand{\set}[1]{\left\{ #1 \right\}}
\newcommand{\True}{\mathit{true}}
\newcommand{\False}{\mathit{false}}
\newcommand{\None}{\mathit{none}}
\newcommand{\argmin}[1]{\underset{#1}{\arg\min}}
\newcommand{\complexity}[1]{\mathcal{O}\left(#1\right)}

% Background
\newcommand{\MAPFR}{MAPF\ensuremath{_{R}}\xspace}
\newcommand{\OCCBS}{OC-CBS\xspace}
\newcommand{\deltaBR}{$\delta$-BR\xspace}
\newcommand{\MetricSpace}{\mathcal{M}}

\newcommand{\Astar}{A$^*$\xspace}
\newcommand{\MinCost}{\mathit{MinimumTime}}

\newcommand{\SIPParrivalTime}{t^\mathit{arrive}}
\newcommand{\SIPPstarttime}{t^\mathit{start}}
\newcommand{\SIPPendtime}{t^\mathit{end}}
\newcommand{\SIPPsafeI}{I^\mathit{safe}}

% This method's lables
\newcommand{\ours}{AOC-CBS\xspace}
\newcommand{\oursPP}{\text{MG-CSIPP}\xspace}
\newcommand{\oursTPSIPP}{\text{TP-SIPP}\xspace}

% OC-CBS constraint tree 
\newcommand{\CTroot}{R}
\newcommand{\CTnode}{N}
\newcommand{\CTconst}{C}

% Problem definition commands
\newcommand{\GraphSet}{\mathbb{G}}
\newcommand{\Graph}{\mathcal{G}}
\newcommand{\Vertices}{\mathcal{V}}
\newcommand{\Edges}{\mathcal{E}}
\newcommand{\StateSpace}{\mathcal{S}}

\newcommand{\vertex}{v}
\newcommand{\state}{s}
\newcommand{\waitable}[1]{\mathit{waitable}(#1)}
\newcommand{\edge}{e}
\newcommand{\traj}{\tau}
\newcommand{\dur}{D}
\newcommand{\source}[1]{{\substack{{\vertex} \\ {\cdot\rightarrow}}}(#1)}
\newcommand{\target}[1]{{\substack{{\vertex} \\ {\rightarrow\cdot}}}(#1)}
\newcommand{\timestart}[1]{{\substack{{t} \\ {\cdot\rightarrow}}}(#1)}
\newcommand{\timeend}[1]{{\substack{{t} \\ {\rightarrow\cdot}}}(#1)}

\newcommand{\Agents}{\mathcal{A}}
\newcommand{\agent}{a}
\newcommand{\Tasks}{\Lambda}
\newcommand{\task}{\lambda}
\newcommand{\AgentStart}{\vertex_\mathit{start}}
\newcommand{\model}{\rho}

\newcommand{\action}{\alpha}
\newcommand{\move}{m}
\newcommand{\wait}{w}
\newcommand{\inconflictfunc}{\mathit{InConflict}}
\newcommand{\inconflict}[1]{\inconflictfunc\left(#1\right)}
\newcommand{\isconflict}[1]{\mathit{IsConflict}\left(#1\right)}

\newcommand{\plan}{\pi}
\newcommand{\JointPlan}{\Pi}
\newcommand{\Solutions}{\mathbb{S}}

\newcommand{\obj}{\sigma}
\newcommand{\cost}{\sigma^\plan}

% Models
\newcommand{\position}{\mathbf{p}}
\newcommand{\velocity}{\mathbf{v}}
\newcommand{\orientation}{\varphi}

% Roadmap sampling
\newcommand{\freeCells}{X_\text{free}}
\newcommand{\density}{\rho}
\newcommand{\clearance}{r}
\newcommand{\stretchfactor}{\kappa}
\newcommand{\maxlength}{l_\text{max}}
\newcommand{\separation}{\delta}

% AOC-CBS commands
\newcommand{\NodeQueue}{\mathcal{Q}}
\newcommand{\NodeSelection}{\mathcal{Q}^\mathit{sel}}
\newcommand{\NewNodes}{\mathcal{Q}^\mathit{new}}
\newcommand{\NewSolutions}{\Solutions^\mathit{new}}
\newcommand{\Portfolio}{\Psi}
\newcommand{\policy}{\psi}
\newcommand{\policySelection}{{\mathit{sel}}}
\newcommand{\policyRepairSet}{\mathbf{R}}
\newcommand{\repair}{r}
\newcommand{\incumbentSolution}{\JointPlan^\mathit{UB}}
\newcommand{\objectiveUB}{\mathit{UB}}
\newcommand{\objectiveLB}{\mathit{LB}}

\newcommand{\CTnodeLB}[1]{LB\left( #1 \right)}
\newcommand{\CTnodeUB}{UB}
\newcommand{\CTreachableSolutions}[1]{\Solutions\left(#1\right)}

\newcommand{\ConflictSelectionArbitrary}{\emph{Arbitrary}\xspace}
\newcommand{\ConflictSelectionSemiCard}{\emph{Semi-/Cardinal}\xspace}
\newcommand{\ConflictSelectionBestof}[1]{\emph{Best-of-$#1$}\xspace}

% Misc
\newcommand{\epsmach}{\epsilon_\mathit{mach}}
\newcommand{\ourRepo}{\url{https://github.com/Adcombrink/AOC-CBS}\xspace}
\newcommand{\gapUB}{\ensuremath{\Delta^\mathit{UB}}\xspace}
\newcommand{\GithubAOCCBS}{\url{https://github.com/Adcombrink/AOC-CBS}\xspace}
\newcommand{\GithubOCCBS}{\url{https://github.com/Adcombrink/Optimal-Continuous-CBS}\xspace}

% MAPS
\newcommand{\mapRoom}{\emph{Room}\xspace}
\newcommand{\mapRandom}{\emph{Random}\xspace}
\newcommand{\mapDenSparse}{\emph{Den520d Sparse}\xspace}

\newcommand{\MovingAIMap}[1]{\emph{#1}\xspace}

\begin{abstract}
\noindent 
Many research fields share a common structure: a set of agents, each pursuing its own goal, whose actions must be coordinated so that no two of them conflict.
Multi-Agent Path Finding (MAPF) is a concrete instance of this structure, with applications from warehouses to road traffic and airports.
Much of MAPF research assumes discrete time, circular agents sharing one spatial graph, a single goal per agent, and that an agent must remain at its goal once reached, precluding heterogeneous fleets, non-geometric conflicts, task sequences, and agents that move on after completing them.
We generalise the MAPF formulation to lift these assumptions, and present Anytime-Optimal Continuous-time Conflict-Based Search (\ours), an exact and solution-complete solver for it.
\ours guarantees the eventual return of an optimal solution, while reporting an incumbent with a known optimality gap upper bound throughout its runtime; it is configurable with a portfolio of repair functions, one of which we introduce (Tier-Prioritized Safe Interval Path Planning), and can exploit multiple processor cores.
We demonstrate \ours on a mixed fleet of non-convex agents moving along smooth, kinodynamically feasible trajectories.
Preliminary experiments against the exact solver \OCCBS, on well-known benchmarks and roadmaps we sample from them, show \ours is comparable at finding optimal solutions while extending scalability from the tens to the hundreds of agents when a bounded optimality gap is accepted.

\end{abstract}

\section{Introduction}
\label{sec:Introduction}

Many research fields face a version of the same underlying problem: a set of agents, each trying to accomplish its own goal, whose actions must be coordinated so that no two of them conflict. 
The agents may be mobile robots navigating a shared space~\cite{SurveyStern2019},
jobs contending for machines in a production schedule~\cite{Pinedo}, 
patients moving through a hospital ward~\cite{lin2023genetic}, 
or robotic arms moving through a shared configuration space~\cite{Shaoul_Mishani_Likhachev_Li_2024}. 
What unites these problems is not the presence of robots or physical space, but a common structure: 
each agent moves through its own space of possible states, 
and two agents are said to conflict when they simultaneously occupy states that are mutually disallowed.

This work is grounded in the setting of mobile robots, \emph{Multi-Agent Path Finding} (MAPF).
We use it as a concrete starting point, since its formal structure is easy to state precisely: given a graph and a start and goal vertex for each agent, MAPF asks for a set of plans, one per agent, such that every agent gets to its goal and no two agents conflict~\cite{SurveyStern2019}. 
MAPF finds its application in
warehouses and factories~\cite{wurman2008coordinating, MAPF_buzzword, brown2020optimal, varambally2022mapf, brorsson_TUVE},
video games~\cite{snape2012reciprocal, Ma_Yang_Cohen_Kumar_Koenig_2017, harabor2022benchmarks},
road traffic and airports~\cite{LiAAAI23, goenawan2025astmautonomoussmarttraffic, morris2016planning, li2019departure},
lunar operations~\cite{krawciwLunarRovers},
and many more.
In the overwhelming majority of MAPF research, a conflict is understood to be two agents occupying the same vertex or traversing the same edge at the same time, typically under a discretization of time that makes this condition well defined~\cite{SurveyStern2019}. 
This is a simplification of physical collision, and narrows the range of real-world problems the model can represent to those where two agents in different states never collide. 
Here, we take a broader view still: not only do we drop this simplification and reason about collision directly, we treat collision itself as just one instance of a more general \emph{pairwise conflict}.

Furthermore, there are several aspects that are motivated by real-world problems which, to our knowledge, are seldom considered together in the MAPF literature.

First, real-world fleets can be heterogeneous not only in how the agents are shaped, but also in how they move.
A fixed-wing aircraft, a ground vehicle, and a quad-rotor do not share a motion model, so requiring every agent to follow trajectories on one common graph is itself a limiting assumption, independent of any assumption about waiting or goal structure.
We instead follow the same modelling approach as in~\cite{ECBS-CT} and let each agent move on its own graph over its own state space, so that an agent's kinodynamic constraints are captured entirely by the states and trajectories that are available to it, and heterogeneous fleets can be coordinated without forcing a shared representation. This formulation allows for completely different types of agents, such as mobile robots and robotic arms, in the same problem. 

Second, not every real-world conflict is a geometric collision. Two airborne drones may be in an unacceptable configuration---one caught in the downwash of the other---without their physical volumes ever intersecting. Treating collision as the only source of conflict, and volume as the quantity that determines it, therefore excludes by construction conflicts of this kind. We instead treat volume-based collision as one instance of a conflict definition that is left abstract: for any two agents, each in some state, whether or not the pair conflicts is specified by an external function. Conflicts stemming from causes other than physical overlap can therefore be represented directly.

Third, several assumptions about how agents use their goals and vertices limit real-world applicability. Some tasks involve traversing an edge rather than being located at a specific position: aerial survey missions require aircraft to fly along specific trajectories, and fire-fighting missions may require water-drops along specific boundaries. 
We therefore allow tasks to be either edge traversals or vertex occupations. Real-world tasks also typically involve a service time at a location, modelling an agent doing something at its goal (e.g., loading or unloading) before departing, rather than simply arriving. 
Additionally, the standard requirement that every agent must remain at its goal indefinitely once reached forces a solution to transition the entire system from its initial collective starting state to a final state where all agents occupy their goals simultaneously. 
Consequently, an agent cannot reach its goal early and then step out of the way to let others pass without returning to it again.
% Additionally, requiring every agent to remain at its goal indefinitely once reached, as is standard, forces a solution to bring all agents simultaneously to rest at their respective goals; agents cannot complete their goal and then move out of the way of others without returning to it. 
Relaxing this requirement has recently been shown to substantially improve solver performance~\cite{GabayReachability}, and we correspondingly allow agents to move after completing their tasks. Finally, it is commonly assumed that every vertex is \emph{waitable}, i.e., that an agent may remain there indefinitely. This precludes vertices representing states with non-zero higher-order derivatives: a fixed-wing aircraft, for instance, cannot remain stationary while airborne, even though specific airborne states must still be reachable.

With these practical motivations, we make a number of extensions to the common MAPF setting:
\begin{enumerate}
    \item Each agent moves within its own state space according to its own graph within that state space, allowing heterogeneous agents to coexist within a single problem.
    \item Time is continuous: edges can be traversed at any real-valued time, and edge traversals and waits at vertices can take any real-valued duration.
    \item Conflicts are defined abstractly rather than geometrically, with volume-based collision as one possibility, allowing conflicts such as drone downwash that do not correspond to any physical overlap.
    \item Tasks form a sequence of edge traversals or vertex occupations, each vertex with an associated service time; agents without any tasks (idle agents) are also considered.
    \item Agents may continue to move after completing all of their tasks, so as to avoid other agents or return to a rest-position, rather than remaining fixed at a single goal.
    \item Vertices do not need waitable, allowing for states that cannot be held indefinitely.
\end{enumerate}

Formally, we consider the problem that follows from these extensions. 
There is a set of agents, each moving on its own graph over its own state space.
Agents may traverse edges starting at any real-valued time, and edges may require any real-valued duration to traverse;
agents can also wait at vertices for any real-valued duration.
Each agent is assigned an arbitrarily long, possibly empty, sequence of tasks, where each task is either to occupy a specific vertex for at least some service duration (possibly forever) or to traverse a specific edge. 
Conflicts are defined abstractly: 
for any two agents, each in some state, an external conflict definition specifies whether the pair conflicts.
We make two assumptions on conflicts.
First, they are pairwise. That is, we consider only conflicts that are defined for exactly two agents. 
We do not model conflicts with more than two agents, such as, for instance, a conflict where the combined downwash of two drones disturbs a third below. 
Second, we only consider conflicts that last for at least a few orders of magnitude longer than the smallest time precision of the machine on which a solution algorithm runs;
conflicts lasting a shorter time than that could be impossible to detect.
A solution is a plan for each agent that takes it from its starting vertex through all of its tasks, in order, and finally to a vertex where it can remain indefinitely without conflict---only this terminal state is required to be waitable, so that a solution never leaves an agent obligated to depart again. 
Each agent's plan incurs a cost, and the global cost is a function of these individual plan costs. 
The common \emph{sum-of-costs} (SOC) and \emph{makespan} are two examples. 
An optimal solution minimises the global cost.

We do not aim to provide a full accounting of all of the problems that the above formulation may encompass. 
However, we postulate that this formulation generalises \emph{classical MAPF}~\cite{SurveyStern2019}, continuous-time MAPF~\cite{CCBS}, \emph{MAPF with large agents}~\cite{LargeAgents}, \emph{Multi-Agent Motion Planning}~\cite{ECBS-CT}, \emph{Multi-Goal MAPF}~\cite{MultiGoalMAPF}, 
and can be adapted to various problems from other fields, 
such as job-shop scheduling~\cite{RoselliJobshop}, 
quantum qubit routing~\cite{bansal2025paralleltokenswappingqubit}, 
railway traffic management~\cite{damatorailway}, 
and computer chip design~\cite{antoine2008global}. 
The generality of the problem formulation contributes to the value of a method for obtaining practical, or even optimal, solutions.

This work introduces \emph{Anytime-Optimal Continuous-time Conflict-Based Search} (\ours), a solver for the above problem that converges on an optimal solution in finite time, while providing an improving incumbent solution with a known upper bound on its optimality gap throughout runtime.
\ours builds on the exact solver \OCCBS~\cite{OCCBS} and addresses its most significant limitation:
no solution is returned until the optimal has been found, which may take an unpractical amount of time despite guarantees of a solution in finite time.
\ours retains \OCCBS's exactness and solution completeness on this more general problem formulation, 
while at the same time managing a configurable portfolio of repair functions that constantly searches for suboptimal solutions.
As one such repair method, we introduce \emph{Tier-Prioritized Safe Interval Path Planning} (\oursTPSIPP), a prioritized planner that resolves a limitation of prior continuous-time prioritized methods~\cite{PSIPP} by allowing agents to vacate a location after completing a task there, rather than assuming each agent's final position is fixed in advance.
The \ours architecture opens the possibility for multi-core processing, which our experiments show to greatly improve performance on harder problems. 

The value of \ours is not necessarily in finding optimal solutions, although our experiments show equivalent performance with \OCCBS. 
Instead, \ours's value lies in the anytime property: an incumbent solution with a known optimality gap upper bound is available at any time once a first solution has been found. 
A purely heuristic method could offer something similar, but without any guarantee of ever finding a solution, let alone an optimal one, and without an updating optimality gap upper bound;
\ours retains both the guarantee and the anytime availability.

The generality of the problem formulation is handled by \ours not itself encoding any domain-specific detail, these are supplied separately as a pre-processing step. 
Specifically, domain knowledge about agents, their state spaces and graphs, and the conflict definitions lies outside of \ours.
We provide implementations of these preprocessing steps for MAPF with 
2D circular and arbitrary (including non-convex) polygonal agents, with straight-line or arbitrary connected and continuous trajectories.
\AC{Work is on-going to extend these 2D models to 3D space.}
Figure~\ref{fig:Demo} showcases an example of a problem solved by \ours,
including a mixed fleet of non-convex agent shapes, moving on a graph with smooth trajectories under acceleration limits.
The pre-processing step is a one-time cost that is reusable for any instance on the same graph using any number of the same types of agents. 
This complex system of non-convex agents on smooth trajectories took around $11$~minutes to pre-process;
circular agents, and straight-line trajectories, are significantly shorter (see Section~\ref{sec:Experiments}).
\ours took $5.50$~milliseconds to find the optimal solution. 
More details can be found in Section~\ref{sec:Experiments:ComplexDemonstration}.
We compare \ours to \OCCBS on the well-known MovingAI~\cite{SturtevantBenchmarks} benchmark problems using 2D circular agents and straight-line traversals, 
showing that \ours finds optimal solutions equivalently while extending scalability from the tens to the hundreds of agents.
\begin{figure}
    \centering
    \includegraphics[width=1\linewidth]{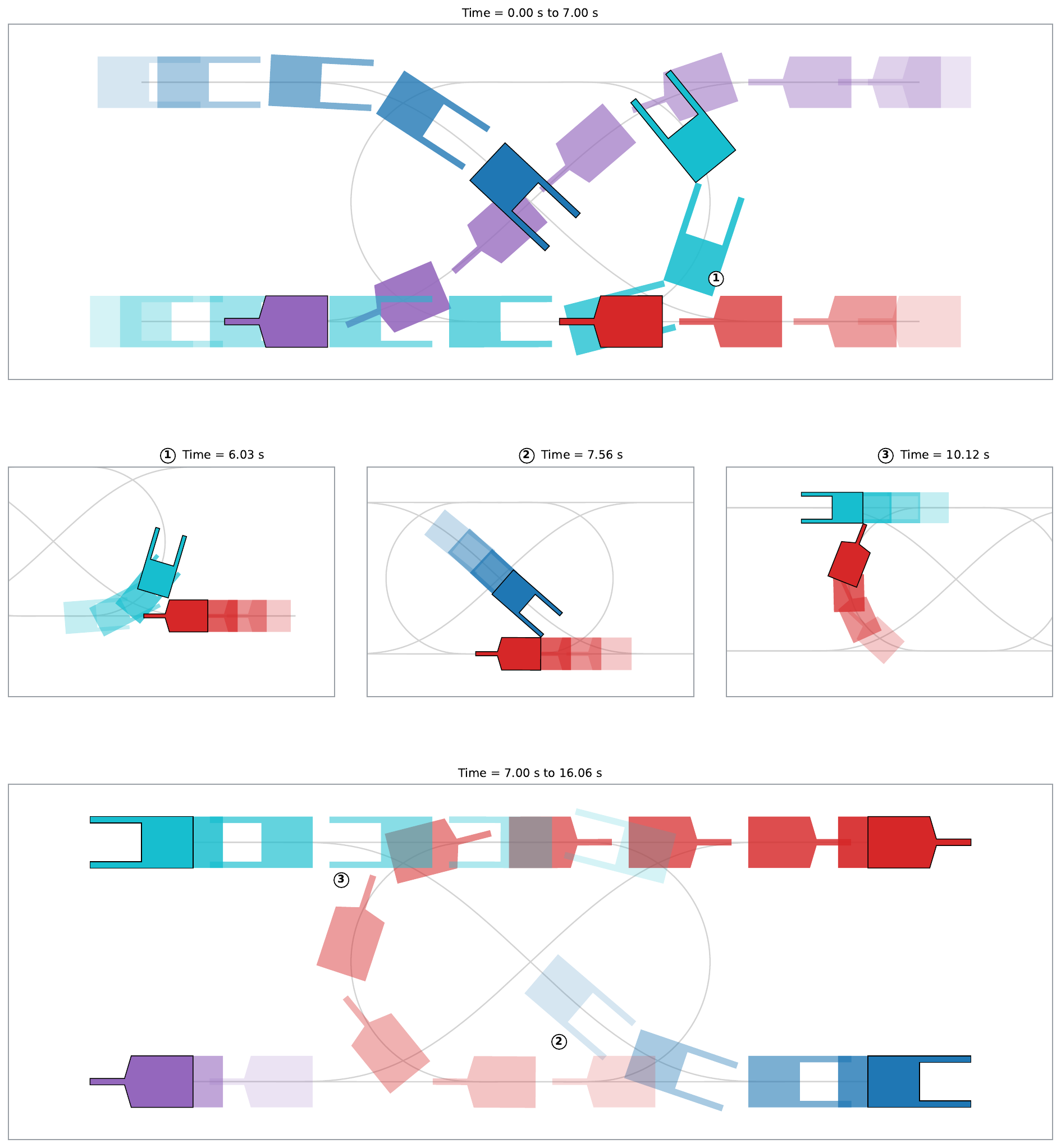}
    \caption{An example problem with two forklifts (starting on the left) and two mobile cranes (starting on the right), solved to optimality by \ours. The top figure shows the first half of the solution, with each agent's trace showing where it started; 
    the centre figures show each of the three closest encounters between agents (enumerated) in the top and bottom figures; 
    and the bottom shows the second half of the solution and each agent at its goal.}
    \label{fig:Demo}
\end{figure}

The remainder of this article is structured as follows:
Section~\ref{sec:Background} presents existing problem formulations and solution methods, putting this work in the context of the state of the art.
The problem formulation addressed here is formally defined in Section~\ref{sec:ProblemFormulation}, 
followed by Section~\ref{sec:AOCCBS}'s introduction of \ours and \oursTPSIPP.
\ours's theoretical guarantees are derived in Section~\ref{sec:TheoreticalGuarantees},
and Section~\ref{sec:Experiments} presents an experimental ablation study, comparisons with \OCCBS, and more details on the specific demonstration problem in Figure~\ref{fig:Demo}.
Finally, discussions surrounding the experimental results and future work are found in Section~\ref{sec:Discussion}, followed by conclusions is Section~\ref{sec:Conclusion}.
For additional details regarding the specific MAPF implementation and benchmark map sampling, see Appendices~\ref{sec:Implementation} and~\ref{sec:BenchmarkGeneration}, respectively. 

\section{Background and Related Work}
\label{sec:Background}

\subsection{Existing Formulations}
\label{sec:Background:ExistingProblems}
This section details a few relevant formulations of the many that exist for MAPF, 
together with some discussion regarding the implications of their assumptions.

% Central limiting assumptions to point out
% - Idle agents?
% - still assumes all vertices can be waited at
% - still assumes volumes / geometric-based conflicts, with over conservatism
% - Discretises space
% - allows for arbitrary volumes,
% - allows for hetergeneous motion models since agents move on their own graph.
% - still assumes only single goal vertices, and that agents end at goal

\subsubsection{Classical MAPF}
\label{sec:Background:ExistingProblems:MAPF}

The classical MAPF problem~\cite{SurveyStern2019} is posed on an undirected and unweighted graph $\Graph = \tuple{\Vertices, \Edges}$ with vertices $\Vertices$ and edges $\Edges$.
Agents in $\Agents$ exist on the graph, starting at time $t=0$ at a unique vertex and tasked with reaching a unique goal vertex. 
Time is discrete, such that at every timestep an agent can either wait at its current vertex or traverse an edge to a neighbouring vertex.
Two agents conflict if they are located at the same vertex or traverse the same edge in opposite directions at the same timestep.
A plan $\plan^\agent$ is a sequence of actions, one for each timestep, or agent $\agent$ to execute.
The plan $\plan^\agent$ is \emph{valid} if it respectively starts and ends at $\agent$'s start and goal vertex.
A joint plan $\JointPlan = \left\{ \plan^\agent \mid \agent\in\Agents\right\}$ contains one plan for each agent, and is valid if every one of its plans is valid.
A solution is a valid joint plan where no two agents conflict while following their plans.
Each plan $\plan$ has a \emph{plan cost} $\cost(\plan)\in\Real$, 
which is commonly the length of the plan.
The \emph{global cost} $\obj(\JointPlan)$ is a function determining the quality of a joint plan $\JointPlan$;
% over the costs of the individual plans $\plan\in\JointPlan$; 
an optimal solution minimises $\obj$.
Two common cost functions are
\begin{align*}
    \text{Sum-of-Costs (SOC):}\quad& \obj(\JointPlan) = \sum_{\plan\in\JointPlan} \cost(\plan), \\ \text{Makespan:}\quad& \obj(\JointPlan) = \underset{\plan\in\JointPlan}{\max} \;\cost(\plan).
\end{align*}

Under the common setting that MAPF is applied to mobile robots in a shared space,
the assumptions in this classical MAPF formulation bring with them a number of implications.
First, the physical shape and size of an agent---referred to as the agent's \emph{volume}---is not considered;
conflicts depend only on the location of agents on the graph.
A real-world deployment would require $\Graph$ to be constructed so that no two agents conflict when located at different vertices or traversing different edges, thus setting a lower bound on the distance between two graph elements (assuming the agents represent non-zero sized entities) and therefore also the resolution at which the real environment can be modelled.
Concretely, if $\Graph$ represents the environment with a higher resolution than the physical size of the agents, then it would be possible for two agents to be located on different vertices in the graph while still colliding in the real system. 

Second, since time is discrete, all edges take exactly one timestep to traverse.
Real-world environments do not always follow such a grid-like structure that this assumption implies. Consequently, long corridors must be modeled with multiple edges when otherwise a single longer edge would suffice; 
if not all edges are unit-length, a real deployment would require agents to wait idly until every agent has reached its destination vertex, else the agents de-synchronize.

\subsubsection{MAPF in Continuous Time}
\label{sec:Background:ExistingProblems:MAPFR}

Multi-Agent Path Finding in Continuous-time (\MAPFR)~\cite{CCBS} drops several of the assumptions made in classical MAPF, most prominently, that time is discrete.
The \MAPFR is posed over a directed graph $\Graph = \tuple{\Vertices, \Edges}$ where each vertex in $\Vertices$ corresponds to a position in a metric space $\MetricSpace$, and each edge in $\Edges$ corresponds to at least one (but possibly several) trajectories in $\MetricSpace$, starting and ending at the source and target vertex's position, respectively.
Just as in classical MAPF, agents start at a unique vertex and are assigned a unique goal vertex. 
However, agents can now traverse edge trajectories at any real-valued time and wait at vertices for any real-valued duration.
Time being continuous means that agents can be moved arbitrarily close to each other on the graph; 
for instance, an agent's traversal can be timed such that it arrives at the target vertex at some miniscule $\epsilon$ time after another agent arrives there.
Thus, for a \MAPFR solution to be practically relevant in a setting where the agents represent physical robots and conflicts are collisions, it is necessary to consider the agents physical volumes. 
However,~\cite{OCCBS} highlights that conflicts need not be strictly geometric-based; instead, conflicts can be defined by an abstract conflict-function, which is then definable for each specific application.
In this work, we adopt this abstract definition of conflicts.

The \MAPFR formulation is more expressive than classical MAPF, however, there still remains several limiting assumptions.
Since all agents exist on the same graph, it is assumed that every agent is able to follow every edge trajectory; the edge trajectories are therefore bounded by the most limiting motion model among the agents.
Furthermore, it is assumed that every agent is able to remain indefinitely at every vertex in the graph. 
These assumptions on edge trajectories and vertices together practically limit the types of agents that can be modeled.
Consider fixed-wing aircraft that cannot remain still while airborne. 
First, these cannot be modeled together with ground vehicles since ground vehicles cannot follow flight trajectories. 
Second, even in a setting with only fixed-wing aircraft, the only vertices that can be included are those along the ground; fix-wing aircraft generally cannot remain for an arbitrary amount of time at a position in the air. 
Instead, every possible trajectory from one ground-based vertex to another would have to be included as an edge-trajectory---from takeoff to landing---which is likely infeasible in real-world applications.

\subsubsection{Multi-Agent Motion Planning}
\label{sec:Background:ExistingProblems:MAMP}

The \emph{Multi-Agent Motion Planning} (MAMP) problem~\cite{ECBS-CT} bares many similarities with the \MAPFR problem.
Time is continuous, such that
edges (representing trajectories) can be traversed at any real-valued time and vertices can be occupied for any real-valued duration.
However, each agent $\agent\in\Agents$ exists on it's \emph{own} graph $\Graph^\agent = \tuple{\Vertices^\agent, \Edges^\agent}$ defined on a state space $\StateSpace$.
Therefore, edge trajectories and vertex positions are defined in $\StateSpace$. 
This allows for an agent's graph to be tailored specifically to the agent's motion model, thereby avoiding many of the issues with all agents existing on the same graph that are present in the \MAPFR formulation.
In the MAMP formulation from~\cite{ECBS-CT},
conflicts are inherently collision-based: 
space is discretized into cells; 
each vertex and edge is associated with the set of cells that are swept by the agent's volume as they are used.
A collision occurs when two agents occupy the same cell at the same time.
Thus, this formulation also supports arbitrary agent volumes by the set of cells that are occupied when an agent uses a specific graph element. 

Under finite-precision arithmetic, which we assume all practical implementations operate under, some level of discretization will always be present.
However, the discretization of space into cells directly couple precision with computational overhead.
Collision detection will always be conservative.
For instance, at any level of discretization, it will always be possible for two agents to occupy the same cell without their volumes actually overlapping; 
this would still register as a collision.
Using course cells would exacerbate the conservativeness, while using granular cells increases the computational overhead associated with tracking the usage of each of those cells.
This is one motivation behind the way conflicts are handled in~\cite{OCCBS} and in this work: instead of discretizing space into cells, we instead store the incompatibility of graph elements directly. Conceptually, this is like storing a single cell for each pair of graph elements where agents on each respective element could conflict.
% As a final point, basing conflict detection on discretized space inherently assumes that all conflicts are geometric-based.

\subsubsection{Task/goal-structure and hybrid approaches}
\label{sec:Background:ExistingProblems:tasks}

Classical MAPF, continuous-time MAPF, and MAMP all share a further restriction, independent of how they model time, space, or graphs: every agent is assigned exactly one goal vertex, and a valid plan must end there. This single requirement carries several consequences that limit real-world applicability.

First, since every agent's plan must end at its goal, an agent cannot reach its goal and then move out of the way for other agents without returning to it afterward. 
Consequently, the solver is not merely tasked with planning each agent to visit its goal vertex, but with finding a mapping from the start configuration, in which every agent is simultaneously at its start, to the goal configuration, in which every agent is simultaneously at its goal. 
It is conceivable that some applications do not require such tight coordination; 
it may be sufficient for each agent to have reached its goal at some point during the plan, after which it is free to continue moving. 
Second, requiring every plan to end at a goal vertex also requires goals to be unique, which is incompatible with applications where several agents must, for instance, deliver goods to the same location. 
Third, this requirement leaves no way to represent idle agents---those that have either completed their task or had no task to begin with---without assigning them a goal. These idle agents are sometimes handled by assigning \emph{dummy goals}~\cite{GabayReachability, FMScheduler}, which introduces the further problem of deciding where a dummy goal should be placed.

Several works relax one or more of these consequences directly. The pickup-and-delivery problem~\cite{Ma_PickupAndDelivery} allows each agent two goal locations rather than one, and Multi-Goal MAPF~\cite{MultiGoalMAPF} generalizes this to an arbitrary sequence of goals per agent, addressing goal uniqueness and multiplicity. Separately, \cite{TransientMAPF} critiques the requirement that a solution end with all agents simultaneously at rest at their goals, arguing that many real-world problems only require each agent to visit its goal at some point, not to remain there in synchrony with every other agent. Closest to our own treatment of this issue, \cite{GabayReachability} proposes a reachability objective in which an agent's plan is permitted to continue moving after its goal has been visited, rather than assuming its final position is fixed at the goal in advance.

\subsubsection{Summary of the Remaining Gaps}
Throughout Sections~\ref{sec:Background:ExistingProblems:MAPF}--\ref{sec:Background:ExistingProblems:tasks} 
we have discussed the various assumptions made on a select set of MAPF variants regarding time, space, conflicts, and goal structure.
We now provide a summary in Table~\ref{tab:background:formulation_comparison} of these aspects and how this work departs from each.
\begin{table}[h]
    \centering
    \caption{Comparison of existing MAPF formulations along the discussed axes, and where this work departs from each.}
    \label{tab:background:formulation_comparison}
    \resizebox{\textwidth}{!}{
    \begin{tabular}{lcccc}
        \toprule
        & \textbf{Classical MAPF} & \textbf{Continuous-time MAPF} & \textbf{MAMP} & \textbf{This work} \\
        \midrule
        Time & Discrete & Continuous & Continuous & Continuous \\
        Graph & Shared across agents & Shared across agents & Per-agent & Per-agent \\
        Vertex waitability & All waitable & All waitable & All waitable & May be non-waitable \\
        Conflict basis & Shared vertex/edge only & Geometric (abstractable) & Geometric (cell-based) & Abstract \\
        Goals per agent & One, unique & One, unique & One, unique & Sequence \\
        Post-goal behavior & Must remain at goal & Must remain at goal & Must remain at goal & May continue moving \\
        Idle agents & Not handled & Not handled & Not handled & Handled \\
        Task type & Vertex only & Vertex only & Vertex only & Vertex (with service time) or edge \\
        \bottomrule
    \end{tabular}
    }
\end{table}

\subsection{Conflict-Based Search Solvers}
\label{sec:Background:OCCBS}

This section introduces the lineage of discrete-time Conflict-Based Search, leading up to Optimal Continuous-time Conflict-Based Search.

\subsubsection{Discrete-time Conflict-Based Search}
\label{sec:Background:CBS}

The work in~\cite{CBS} introduces the seminal \emph{Conflict-Based Search} (CBS), 
a solver for the classical MAPF problem that is \emph{exact} (returns only optimal solutions) and \emph{solution complete} (returns a solution if one exists, but does not necessarily report a problem to be unsolvable).
CBS consist of two levels; 
a high-level search over a binary constraint tree (CT), and a low-level path planner.
Informally, the search starts at the CT root node, representing the problem without any conflict-avoidance constraints.
Every node in the CT contains a lower bound on the global cost that any solution in its sub-tree can achieve; 
the leaf node minimising this lower bound is selected for expansion, at which point either it contains a solution (which is shown to be the optimal solution) or new conflict-avoidance constraints are added to its two newly spawned children. 
The low-level path planner is invoked at every child node to compute new plans satisfying the newly added constraints.

Specifically, every CT node $\CTnode$ contains 
a set of constraints $\CTnode_\CTconst$
and a valid joint plan $\CTnode_\JointPlan$, 
where every plan $\plan\in\CTnode_\JointPlan$ minimises the plan cost $\cost(\plan)$ over all valid plans satisfying $\CTnode_\CTconst$.
The search starts at the CT root $\CTroot$ with $\CTroot_\CTconst=\varnothing$;
every plan $\plan^\agent \in \CTnode_\JointPlan$ is the lowest-cost plan for $\agent$ without considering conflicts with other agents.
At every iteration of the search, the leaf node $\CTnode$ minimizing $\obj(\CTnode_\JointPlan)$ (initially $\CTroot$) is selected for expansion.
During expansion, an arbitrary conflict in $\CTnode_\JointPlan$ is selected, 
say involving agents $i$ and $j$, 
and a \emph{branching rule} is applied to the conflict to create constraints $c_i$ and $c_j$ for each respective agent is created.
These constraints ensure that their corresponding agent avoids this specific conflict. 
For instance, $c_i$ and $c_j$ could respectively forbid $i$ and $j$ from occupying the vertex where and when the collision occurred, if the collision involved using the same vertex at the same timestep.
Two child nodes to $\CTnode$, $\CTnode^i$ and $\CTnode^j$, are subsequently spawned with $\CTnode^i_\CTconst = \CTnode_\CTconst \cup \left\{ c_i \right\}$ and $\CTnode^j_\CTconst = \CTnode_\CTconst \cup \left\{ c_j \right\}$.
That is, $i$ is constrained in $\CTnode^i$ to give way to $j$, while $j$ is constrained in $\CTnode^j$ to give way to $i$.
The joint plans $\CTnode^i_\JointPlan$ and $\CTnode^j_\JointPlan$ are recomputed by the low-level path planner to satisfy $\CTnode^i_\CTconst$ and $\CTnode^j_\CTconst$, respectively.
In practice, only $i$ and $j$'s plans need recomputing, since they are the only two agents with new constraints; the remaining agents' plans in $\CTnode_\JointPlan$ already satisfy the child nodes' constraints.

Let $\Solutions(\CTnode)$ be the set of all solutions consistent with the constraints in $\CTnode_\CTconst$, and
\begin{equation}
    \CTnodeLB{\CTnode} \leq \min_{\JointPlan\in\Solutions(\CTnode)} \; \obj(\JointPlan)
\end{equation}
denote a lower bound on all such solutions.
The two properties underpinning CBS's exactness and solution completeness are
\begin{align}
    \Solutions(\CTnode) \;&=\; \Solutions\left(\CTnode^i\right) \cup \Solutions\left(\CTnode^j\right) \label{CBS:property:soundness} \\
    \CTnodeLB{\CTnode} \;&\leq\; \min\left(\CTnodeLB{\CTnode^i}, \CTnodeLB{\CTnode^j}\right) \label{CBS:property:admissibility}.
\end{align}
Property~\eqref{CBS:property:soundness} ensures soundness in the branching, such that no solutions are removed during the search. 
This property being satisfied depends entirely on the branching rule used to create the constraints that are added to the child nodes.
For a branching to be sound, it must hold that any solution must satisfy at least one of the constraints that the branching rule creates.
For instance, if the constraint $c^i$ forbids $i$ from occupying vertex $\vertex$ at timestep $t$, and $c^j$ forbids $j$ from occupying $\vertex$ at $t$ also, then any solution must satisfy at least one of these constraints since otherwise $i$ and $j$ conflict.

Property~\eqref{CBS:property:admissibility} ensures admissibility in the global cost lower bound, such that the first solution encountered during the best-first CT search is indeed the optimal.
In CBS, the joint plan of a CT node, $\CTnode_\JointPlan$, is used as the global cost lower bound: $\CTnodeLB{\CTnode} = \obj(\CTnode_\JointPlan)$.
Whether this holds or not depends on two things.
First, the global cost function $\obj(\JointPlan)$ must be decomposable into a function over, and non-decreasing in, individual plan costs $\cost(\plan), \plan\in\JointPlan$~\cite{OCCBS}.
With these properties, ensuring that every individual plan cost is optimal ensures that the global cost is optimal.
Second, the low-level planner must be able to find the optimal plan for a single agent given a set of constraints.
If the low-level planner is able to do so, 
then at every CT node $\CTnode$ we know that every plan $\plan\in\CTnode_\JointPlan$ is optimal under $\CTnode_\CTconst$, implying $\CTnode_\JointPlan$ is optimal under $\CTnode_\CTconst$.
Finally, note that for every node $\CTnode'$ in $\CTnode$'s sub-tree it holds that $\CTnode_\CTconst \subseteq \CTnode'_\CTconst$, since constraints are only added and never removed during branchings.
Adding constraints can never lead to lower-cost solutions, hence $\obj(\CTnode_\JointPlan) \leq \obj(\CTnode^i_\JointPlan)$ and $\obj(\CTnode_\JointPlan) \leq \obj(\CTnode^i_\JointPlan)$. 
Property~\eqref{CBS:property:admissibility} is satisfied.

\subsubsection{Continuous-time Conflict-Based Search}

\emph{Continuous-time Conflict Based Search} (CCBS)~\cite{CCBS} is an inexact solver~\cite{CCBS_revisit, OCCBS} for the \MAPFR problem, 
extending CBS to allow for edge traversal times, edges durations, and vertex wait durations, in real-valued time.

Besides allowing for arbitrary agent volumes and edge trajectories (theoretically), the overall structure of the CCBS solver matches closely with CBS:
a two-levelled search consisting of a high-level search in the binary CT and a low-level search for finding optimal single agent plans. 
The main difference between CCBS and CBS can be found in the branching rule. 
Consider two agents $i$ and $j$ that conflict when respectively performing an action, an action being to either wait at a vertex or traverse an edge.
Agent $i$ executes action $\action^i$ starting at time $t^i$, and $j$ executes $\action^j$ starting $t^j$.
As described theoretically, 
CCBS's branching rule then finds for agent $i$ the earliest time $t^i_u>t^i$ when $\action^i$ can be executed without conflicting with $j$ performing $\action^j$ at $t^j$.
Likewise, $t^j_u>t^j$ is found for $j$.
The two constraints $c^i = \langle i, \action^i, [t^i, t^i_u) \rangle$ forbidding $i$ from executing $\action^i$ starting at any time in $[t^i, t^i_u)$, and $c^j = \langle j, \action^j, [t^j, t^j_u) \rangle$ forbidding $j$ from executing $\action^j$ starting at any time in $[t^j, t^j_u)$, are created and added to the constraint set of each respective child node. 

The branching rule that is actually used in the publicly available implementation applies the same procedure to move actions, actions involving traversing an edge, as theoretically described.
However, wait actions, involving waiting at a vertex, are treated in a different way.
Consider an agent $j$ waiting at a vertex and while doing so conflicting with agent $i$ traversing an edge. 
Agent $j$ waits at the vertex during a specific time interval, however, if $j$ waited at the vertex for all time instead, then the interval during which $i$ and $j$ conflict can be recorded as the \emph{intersection interval}.
In the implemented branching rule, the constraint applied to $j$ forbids it from occupying the vertex at any time in the intersection interval.

The low-level planner, CSIPP~\cite{CCBS}, used in CCBS is based on \emph{Safe Interval Path Planning} (SIPP)~\cite{SIPP}, which splits the times when a vertex can be occupied or an edge can be traversed by an agent into safe and unsafe intervals.
The unsafe intervals are defined by the constraints in a given constraint set, and the safe intervals are the unsafe intervals' complement. 
An \Astar~\cite{Astar} search is then done in the vertex-safe interval state space, collapsing the space of dense times into discrete states, and thereby returning an optimal single agent plan in continuous time.

The work in~\cite{CCBS_revisit} finds that the theoretically described branching rule can lead to an infinite expansion of CT nodes between the root and a node representing an optimal solution, thus potentially leading to non-termination. 
They also find that the implemented branching rule in some cases removes solutions from the search, thus violating property~\eqref{CBS:property:soundness}.
Counter-examples are provided in~\cite{CCBS_revisit, OCCBS} showing that CCBS returns a sub-optimal solutions.
Despite the lack of formal guarantees, comparisons with the exact \OCCBS (described in the next section) show that CCBS is often able to find the optimal solution faster due to its more aggressive search space pruning~\cite{OCCBS}.

\subsubsection{Optimal Continuous-time Conflict-Based Search}

\emph{Optimal Continuous-time Conflict-Based Search} (\OCCBS)~\cite{OCCBS} restores exactness and solution completeness to CCBS by introducing a new branching rule, \deltaBR, that changes how conflicting wait actions are handled. 

\deltaBR is based on shifting constraints from~\cite{CCBS_revisit} and applies the same rule as CCBS to conflicts involving two move actions.
However, when applied to a conflict involving a move and a wait action,
\deltaBR finds a specific value $\delta$ based on the intersection interval $I=[t^I_1, t^I_2)$ and the time interval during which the waiting agent occupies its vertex.
Then, for the moving agent $i$ executing the move action $\action^i$ at time $t^i$, a constraint $c^i$ is created forbidding $i$ from executing $\action^i$ starting at any in $[t^i, t^i + \delta)$.
For the other agent $j$ waiting at $\vertex$, a constraint $c^j$ is created forbidding $j$ 
from occupying $\vertex$ at any time in $[t^I_1 + \delta, t^I_2)$ and 
from executing any move action from $\vertex$ starting at any time in $[t^I_1 + \delta, t^I_2)$.
It is shown in~\cite{OCCBS} that \OCCBS is exact and guaranteed to return a solution in finite time, under the assumption that a solution exists.

Essentially, the work in~\cite{CCBS_revisit} identifies that properties~\eqref{CBS:property:soundness} and~\eqref{CBS:property:admissibility} are not sufficient in continuous time to guarantee that a solution will be returned.
That is, the theoretically described CCBS branching rule satisfies both, yet allows for the possibility of an infinite number of nodes between the CT root and one representing an optimal solution.
This problem is not encountered in CBS since time is discretised, thereby allowing for only a finite number of steps between the root and an optimal solution.
Consequently, an additional property is introduced in~\cite{OCCBS} to ensure that an optimal solution is reachable in a finite number of steps in the CT.
That is, the following three properties must be satisfied for \OCCBS to return an optimal solution in finite many iterations (restating properties~\eqref{CBS:property:soundness} and~\eqref{CBS:property:admissibility} for ease of reference):
\begin{enumerate}[label=\textbf{CT-\Roman*}, align=left, leftmargin=1.5cm]
    \item \label{OCCBS:property:soundness} No solutions are removed from the search:
    \begin{equation*}
        \Solutions(\CTnode) \;=\; \Solutions\left(\CTnode^i\right) \cup \Solutions\left(\CTnode^j\right).
    \end{equation*}
    \item \label{OCCBS:property:admissibility} The global cost lower bound monotonically increases when descending the CT: 
    \begin{equation*}
        \CTnodeLB{\CTnode} \;\leq\; \min\left(\CTnodeLB{\CTnode^i}, \CTnodeLB{\CTnode^j}\right).
    \end{equation*}
    \item \label{OCCBS:property:progress} For every infinite path of nodes $\CTnode_1, \CTnode_2, \CTnode_3, \dots$ in the CT, where $\CTnode_i$ is the parent of $\CTnode_{i+1}$, and any $c\in\Real$,
    \begin{equation*}
        \exists k \in \Natural: \forall i\geq k: c < \CTnodeLB{\CTnode_i}.
    \end{equation*}
\end{enumerate}
That is, for~\ref{OCCBS:property:progress}, constraints are accumulated along the path $\CTnode_1, \CTnode_2, \CTnode_3, \dots$;
it must hold that the constraints that accumulate must eventually push the global cost lower bound beyond some threshold $c$, wherever that threshold may be.
Intuitively, this means that as the search continues, 
either such a path will lead to the optimal solution, or it will be ignored in the best-first search in favour of other paths with lower global cost lower bounds in which the optimal solution lies.

\subsection{Suboptimal, Bounded-Suboptimal, and Anytime Methods}
We introduce methods that find sub-optimal solutions, with additionally bounds on sub-optimality that sub-optimality, an \emph{anytime} property of providing improving solutions over runtime, and possibly guarantees of eventual convergence to an optimal solution.

\subsubsection{Prioritized SIPP}
\label{sec:Bakground:PSIPP}

Prioritized SIPP (PSIPP)~\cite{PSIPP} applies prioritized planning to a continuous-time MAPF formulation with circular agents and straight line constant-speed traversals. 
Agents can traverse edges at any real-valued time, and edge traversals can have any positive real-valued duration.
Each agent starts at a unique vertex and is assigned a unique goal vertex, with no two agents at their goal vertices conflicting with each other.
Given some ordering of the agents, 
PSIPP plans the movement of one agent at a time, using SIPP, to move from its starting vertex to its goal vertex while avoiding collisions with other previously planned agents.

PSIPP outperforms CCBS for fast suboptimal solutions, on some maps planning more than 500 agents in less than a second. 
However, PSIPP is inherently a heuristic method that is subject to the initial agent ordering, meaning that the returned solutions have no guarantees of optimality or of even finding a solution.

\subsubsection{Enhanced Conflict-Based Search in Continuous Time}

\emph{Enhanced Conflict-Based Search in Continuous Time} (ECBS-CT)~\cite{ECBS-CT} is a bounded suboptimal solver for the MAMP problem, 
building on the \emph{Enhanced CBS} solver~\cite{barer2014suboptimal}. 

Given a suboptimality bound $w\geq1$, ECBS-CT searches for a solution $\JointPlan$ with $\frac{\obj(\JointPlan)}{\obj^*} \leq w$, where $\obj^*$ is the optimal global cost, such that $\JointPlan$'s suboptimality is at most $w$. 
Setting $w=1$ collapses the method into searching for an optimal solution.
It does this by applying the general CBS framework, with a high-level search in a binary constraint tree and a low-level single-agent path planner, \emph{Soft Conflict Interval Path Planning} (SCIPP), based on SIPP.
However, instead of always selecting the CT node minimising the global cost, ECBS-CT applies a focal search:
with $\obj_\mathit{min}$ being the minimum $\CTnodeLB{\CTnode}$ over all CT leaf nodes $\CTnode$, 
the search selects from the set of nodes satisfying $\CTnodeLB{\CTnode} \leq w\obj_\mathit{min}$.
Among these satisfying leaf nodes, a secondary heuristic minimising the total duration of time intervals in which two or more agents collide is applied.
That is, ECBS-CT focuses the search not on the nodes with highest potential solution quality, but instead on the nodes with satisfactory cost bounds and which are nearest to representing a solution. 

Like SIPP, SCIPP performs an \Astar search in the vertex-safe interval state space, with safe intervals defined by the set of \emph{hard} constraints given by the high-level search.
However, SCIPP applies a similar focal search as in the high-level search:
instead of always selecting the state with minium total estimated path cost $f_\mathit{min}$, any state with $f \leq w\cdot f_\mathit{min}$ is eligable for selection.
Among these, the secondary heuristic selects the state minimising the number of \emph{soft conflicts}, that is, conflicts with other agents in the high-level node's joint plan. 

On these two search levels, ECBS-CT trades some of the slack afforded by $w$ for joint plans that are closer to being solutions, and plans that pre-emptively avoid conflicts. 
The results show that by increasing from $w=1$ to $w=1.2$ ($20\%$ optimality gap) and beyond, runtime decreases and success rate increases, such that plans can be found in some cases for over $10$ times as many agents than otherwise.

Although this focal search mechanism allows ECBS-CT to trade optimality for a substantial gain in scalability, the suboptimality bound $w$ must be fixed before the search begins, and lowering it to tighten the bound requires restarting the search from scratch. 
Anytime or incremental variants that retain this work across values of $w$ are not part of the ECBS-CT framework.

\subsubsection{Anytime Methods in Discrete Time}
\label{sec:Bakground:AnytimeDiscrete}

An anytime solver makes a solution available early and improves it as time allows.
Where a bounded-suboptimal solver fixes its quality target before the search begins,
an anytime solver maintains an incumbent solution whose quality improves monotonically and, ideally, converges to optimal. 
Work on anytime MAPF in discrete time follows three broad approaches.

The first keeps the CBS structure and makes the bounded-suboptimal search itself anytime.
Anytime Focal Search (AFS)~\cite{CohenAnytimeFocal} iteratively tightens the suboptimality factor $w$ while reusing the search effort already expended under looser factors. 
Applied by its authors to the high level of BCBS~\cite{barer2014suboptimal}, it gives what later work~\cite{li2021lns} calls anytime BCBS, and what~\cite{CohenAnytimeFocal} describes as the first anytime
MAPF solver.
Notably, AFS reports each incumbent together with the ratio of its cost to the smallest $f$-value remaining in the open list, and prunes states whose cost exceeds the current incumbent---so the incumbent, the bound certifying it, and the pruning rule all derive from the single focal search over the constraint tree.

The second approach replaces systematic search with local search over complete solutions.
MAPF-LNS~\cite{li2021lns} computes an initial solution with a fast suboptimal method,
then repeatedly removes the plans of a subset of agents, selected by a randomized removal heuristic, and repairs them by prioritized planning, keeping the result whenever it lowers the cost. 
Subsequent work largely differs in how the subset is chosen
or how the repair is performed: 
MAPF-LNS2~\cite{li2022lns2} repairs plans that still contain collisions rather than requiring a feasible starting solution, while~\cite{huang2022mllns} and~\cite{phan2024bandit} learn, 
respectively via supervised learning and multi-armed bandits, which neighbourhoods are worth replanning. 
These methods scale to hundreds of agents, but the incumbent is never accompanied by a lower bound on the optimal cost, so nothing is known about how far from optimal a returned solution is, nor is optimality reached in the limit.

The third performs an anytime search directly over the joint configuration space.
LaCAM$^*$~\cite{okumura2023lacamstar} uses lazy successor generation to return an initial solution very quickly, and then continues searching while revising parent relations among search nodes, converging to an optimal solution for costs accumulated over transitions.
It is thus both scalable and eventually optimal, though---as with the first approach---the guarantee is stated over the enormous discrete-time configuration space.

Two observations carry over to this work. 
First, all of the above assume the classical MAPF setting of Section~\ref{sec:Background:ExistingProblems:MAPF}: discrete time, a shared graph, and one goal per agent at which the plan ends. 
Second, the destroy-and-repair idea underlying LNS is close in spirit to the repair methods in \ours's portfolio (Section~\ref{sec:Method:TPSIPP}): 
both take an existing joint plan and replan a subset of agents around the remainder. 
The difference lies in where the plan comes from and what is known about the result. 
In LNS the repaired plan is the incumbent itself and the repair is the only source of information about solution quality; 
in \ours the plan being repaired is a CT node's joint plan, and the search that produced it simultaneously supplies a lower bound, so any repaired solution can be reported together with a certified optimality gap.

\subsubsection{Anytime Methods in Continuous Time}

To our knowledge, no anytime solver has been proposed for the continuous-time formulations of Section~\ref{sec:Background:ExistingProblems}. 
That literature is instead split between exact solvers that return nothing at all before the optimum is found (CCBS~\cite{CCBS}, OC-CBS~\cite{OCCBS}), 
bounded-suboptimal solvers whose bound is fixed before the search (ECBS-CT~\cite{ECBS-CT}), 
and unbounded-suboptimal prioritized planners that return a single solution and stop
(PSIPP~\cite{PSIPP}, CPLP~\cite{CPLP}).

The reason is largely one of timing. Any search-based anytime method
reports its incumbent against a lower bound, and in continuous time that
bound is trustworthy only under
\ref{OCCBS:property:soundness}--\ref{OCCBS:property:progress}. CCBS was
believed to satisfy these until recently, when~\cite{CCBS_revisit} showed its branching rule violates the first and
third; only with \OCCBS's branching rule were the properties established
on a continuous-time constraint tree at all. \ours\ is, to our knowledge,
the first anytime method built on this foundation.

LNS-style repair faces no comparable obstacle when agents are required to
remain at their goal once reached: a continuous-time prioritized planner
such as PSIPP~\cite{PSIPP} is a perfectly serviceable repair primitive
under that restriction, and nothing here rules out a continuous-time
analogue of MAPF-LNS. 
To our knowledge none has been published. PSIPP's
actual limitation is orthogonal to LNS suitability: it assumes each
agent's final position is fixed at its goal, so it cannot vacate a
completed task to let another agent through
(Section~\ref{sec:Method:TPSIPP}). 
This is what our modified version of
PSIPP, \emph{Tier-Prioritized SIPP}, is built to resolve. The limitation
is not specific to LNS or to PSIPP: it would affect any repair method,
prioritized or otherwise, once agents are allowed to move on after
completing their tasks. 
Configuration-space search of the LaCAM$^*$ kind
relies on a finite successor set per configuration, which is not
immediate when action durations are real-valued.

A focal-search adaptation in the manner of~\cite{CohenAnytimeFocal} is now
equally possible, and would in fact be expressible within \ours as one
selection policy $\policy_\policySelection$ among others, retaining correctness through Theorems~\ref{theorem:AOCCBS:IncumbentDrivenTermination} and~\ref{theorem:AOCCBS:PolicyDrivenTermination}.
What distinguishes \ours is not the constraint tree but the decoupling: 
focal search derives incumbent and bound from a single node ordering, whereas here the bound comes from the CT search and the incumbent from repair, so heuristics may be added freely without bearing on correctness.

\subsubsection{Kinodynamic Feasibility: Repairing After Search versus Searching Within the Feasible Space}

A separate line of work attains continuous, kinodynamically feasible motion by planning in a simpler space and recovering feasibility afterwards. 
\emph{Concrete Multi-Agent Path Planning}~\cite{ConcreteMAPP} plans discrete paths with Tree-LaCAM, an anytime and complete solver built on LaCAM/LaCAM$^*$~\cite{okumura2023lacamstar}, 
converts them into trajectories using learned robot dynamics, 
and executes them under optimal control that repairs residual collisions; 
a heterogeneous fleet of 40 physical ground and aerial robots is coordinated this way, concurrently and onboard, 
at a hardware scale no continuous-time planner operating directly on trajectories has been deployed at. 
In multi-robot kinodynamic motion planning, db-CBS~\cite{moldagalieva2024dbcbs} takes a related route, 
searching over precomputed motion primitives with a bounded discontinuity between them, 
optimizing the resulting trajectory in the joint space, and repeating with a reduced discontinuity bound; 
the method is anytime and asymptotically optimal.

\ours makes the opposite trade. 
Kinodynamic feasibility is a property of graph construction (Section~\ref{sec:ProblemFormulation:generality}) rather than of planning, so the search runs directly in the space of motions the agents can realize. 
Any incumbent is executable by construction, with no trajectory-generation or repair stage after the fact, and the reported optimality gap is with respect to the agents' graphs. 
What is given up is the ability to improve upon that fixed discrete abstraction: solution quality is bounded by the trajectories present in the graphs, 
and enriching them moves cost into the precomputation stage (Section~\ref{sec:Implementation}). 
Such a formulation is appealing in tightly constrained or safety-critical settings, where stronger end-to-end feasibility guarantees may justify the additional computational burden.

Taken together, anytime methods are well developed for discrete-time MAPF but absent from its continuous-time counterparts, 
while the methods that do plan kinodynamically feasible motion obtain their anytime behaviour from a discrete abstraction and restore feasibility only afterwards. 
\ours sits at the intersection: 
anytime, with a certified optimality gap, over a formulation that is continuous in time and feasible by construction.

\section{Problem Formulation}
\label{sec:ProblemFormulation}

% Our problem allows for
% \begin{enumerate}
%     \item Arbitrary and heterogeneous agent motion models (over a finite, discrete set of trajectories).
%     \item Arbitrary and heterogeneous agent volumes. 
%     \item Abstract conflict detection that includes but is not limited to geometric collisions.
%     \item Arbitrary number of tasks (including none), with each task either involving waiting for a service time at a vertex or traversing an edge.
%     \item Infinitely waiting at target (common MAPF assumption) or becoming idle and adjusting to other agents after completing all tasks.
% \end{enumerate}

This problem formulation combines and extends those introduced in Section~\ref{sec:Background:ExistingProblems} to allow for continuous-time movement of a fleet of heterogeneous agents along a discrete set of trajectories in state spaces, with arbitrary definitions of conflicts, sequences of tasks at vertices and edges, and several other practically motivated generalizations.

\textbf{A graph}
$\Graph = \tuple{\Vertices, \Edges}$ is a directed multi-graph over a state space $\StateSpace$ which may be unique to $\Graph$ (although we do not denote this for brevity). 
Every vertex $\vertex\in\Vertices$ corresponds to a unique state $\vertex_\state\in\StateSpace$, 
and $\waitable{\vertex}\in\{\True, \False\}$ denotes if $\vertex_\state$ is a state that can be remained at indefinitely or not.
For instance, it is not possible to remain in a state that contains both a measure and its non-zero derivative.
Each edge $\edge\in\Edges$ has a source vertex $\source{\edge}\in\Vertices$ and a target vertex $\target{\edge}$; 
since $\Graph$ is a multi-graph, there may exist multiple edges with the same source and target vertices.
Edge $\edge$ corresponds to a trajectory $\edge_\traj: [0, \edge_\dur)\mapsto\StateSpace$ mapping time to the state space, with duration $\edge_\dur\in\Real_{>0}$, starting at state $\edge_\traj(0)=\source{\edge}_\state$ and ending at state $\edge_\traj(\edge_\dur)=\target{\edge}_\state$.

\textbf{An agent}
$\agent$ from the set of all agents $\Agents$ exists on a possibly unique graph $\Graph^\agent = \tuple{\Vertices^\agent, \Edges^\agent}$ defined over a state space $\StateSpace^\agent$.
The agent $\agent$ being located at vertex $\vertex\in\Vertices^\agent$ means that it is in the state $\vertex_\state$;
traversing edge $\edge\in\Edges^\agent$ means that it follows the state trajectory~$\edge_\traj$.
% We assume for all $\edge\in\Edges^\agent$ that $\agent$'s underlying motion model permits following $\edge_\traj$.
Initially, agent $\agent$ is located at $\AgentStart^\agent\in\Vertices^\agent$ at time $t=0$; we do not assume $\waitable{\vertex}$.

\textbf{An action} 
$\action$ is executed by an agent and can mean to either traverse an edge at a specific time or wait at a vertex during some time interval.
A \emph{move action} $\move=\tuple{\edge, t}\in\Edges\times\Real$ means to begin traversing $\edge$ at time $t$.
That is, $\move$'s trajectory is the same as $\edge$'s trajectory but shifted in time: $\move_\traj(t') = \edge_\traj(t'-t)$ defined over $t'\in[t, t+\edge_\dur)$.
We say that $\move$'s source and target vertices are the same as $\edge$'s: $\source{\move} = \source{\edge}$ and $\target{\move}=\target{\edge}$.
We also denote $\move$'s start and end time with $\timestart{\move}=t$ and $\timeend{\move}=t+\edge_\dur$, and its duration matching $\edge$'s duration, $\move_\dur = \edge_\dur$.

A \emph{wait action} $\wait=\tuple{\vertex, t_1, t_2}\in\Vertices\times\Real\times(\Real\cup\{\infty\})$ means to wait at $\vertex$ during the time interval $[t_1, t_2)$.
Only wait actions with a positive duration ($t_1 < t_2$) at a vertex that can be waited at ($\waitable{\vertex}$) are considered.
A wait action with $t_2 = \infty$ is referred to as an \emph{infinite wait action}.
Matching the notation of move actions, $\wait$'s trajectory maps to the state of $\vertex$ over the full waiting interval: $\wait_\traj(t') = \vertex_\state$ defined over $t'\in[t_1, t_2)$.
We also denote $\source{\wait} =\target{\wait}=\vertex$, the start time $\timestart{\wait}=t_1$, the end time $\timeend{\wait}=t_2$, and the duration $\wait_\dur = t_2 - t_1$.

Every action $\action$ has a \emph{cost} $\action_\obj$.
In the majority of the MAPF literature, this cost equals the action's duration, meaning $\move_\obj = \move_\dur$ for a move action and $\wait_\obj = \wait_\dur$ for a wait action.
However, in some settings the cost is not measured in time but energy, such that waiting does not cost while moving does. 
We leave action cost as a domain-specific attribute. 

With superscript $\action^\agent$ we mean that the action is executed specifically by agent $\agent\in\Agents$, such that $\edge\in\Edges^\agent$ (for a move action) and $\vertex\in\Vertices^\agent$ (for a wait action). 
However, we omit the superscript for notational simplicity whenever it is otherwise clear.

\textbf{A task sequence}
$\Tasks = \tuple{\task_1, \task_2, \dots, \task_N} \in \left( \left(\Vertices \cup \Edges \right) \times \Real_{\geq0}\right)^N$ is a sequence of $N$ \emph{tasks}, 
with a task meaning to either traverse an edge or to wait at a vertex for at least a specific amount of time.
Tasks in a task sequence must be completed in order and one-at-a-time; 
only once task $\task_i$ is completed can the traversing of an edge or waiting at a vertex to complete $\task_{i+1}$ begin.
For instance, if both $\task_i$ and $\task_{i+1}$ involve waiting at the same vertex, only once $\task_i$ is complete does the time start counting toward $\task_{i+1}$.

An \emph{edge task} $\task = \tuple{\edge, \edge_\dur}$ with $\edge\in\Edges$ is completed by an agent once the agent traverses $\edge$.
Continuing with the same notation as for actions, we say that $\task$'s source and target vertices match the edge $e$: $\source{\task} = \source{\edge}$ and $\target{\task} = \target{\edge}$.
Such an edge task's duration is denoted $\task_\dur = \edge_\dur$.

A \emph{vertex task} $\task = \tuple{\vertex, \dur}$ with $\vertex\in\Vertices$ is completed by an agent once the agent has been located at $\vertex$ for at least $\dur$ time.
For such a vertex task we denote $\source{\task} = \target{\task} = \vertex$ and its duration $\task_\dur = \dur$.

Just as with actions, the superscript $\Tasks^\agent = \tuple{\task_1^\agent, \task_2^\agent, \dots, \task_N^\agent}$ denotes that agent $\agent\in\Agents$ is assigned to $\Tasks^\agent$.
This means that $\edge\in\Edges^\agent$ (for edge tasks) and $\vertex\in\Vertices^\agent$ (for vertex tasks).
Also, number of tasks $N$ may be unique to each agent. 
When otherwise clear from context, we omit the superscript for notational simplicity.

\textbf{A plan}  
$\plan^\agent = \tuple{\action_1, \action_2,\dots,\action_n}$ (with $n$ possibly unique to $\agent$) is a sequence of actions for agent $\agent\in\Agents$ to execute in series. 
The duration of $\plan^\agent$ is the sum of its actions' durations, $\plan^\agent_\dur = \sum_{\action\in\plan^\agent} \action_\dur$, however excluding any final infinite wait action if present.
We say that $\plan^\agent$ is \emph{valid} if the following conditions are true:
\begin{itemize}
    \item The first action $\action_1$ starts at $\agent$'s starting vertex and at time $t=0$:
    \begin{equation*}
        \source{\action_1} = \AgentStart^\agent \;\wedge\; \timestart{\action_1} = 0.
    \end{equation*}
    
    \item Every successive action is connected to the previous in time and place: 
    \begin{equation*}
        \timeend{\action_{k-1}} = \timestart{\action_k} \: \wedge \;  \target{\action_{k-1}} = \source{\action_k}
        \quad \forall k \in [2,..,n].
    \end{equation*}
    
    \item All tasks in $\Tasks^\agent = \tuple{\task_1, \task_2,\dots,\task_N}$ are completed in order and one-at-a-time:
    there exists an ordered set of times $\tuple{t_1, t_2, ..., t_N} \in \Real_{\geq0}^N$ (with $t_i$ being the start time of $\task_i$) where 
    \begin{equation*}
        t_1 \leq t_1 + \task_{1,\dur} \leq t_2 \leq t_2 + \task_{2.\dur}\leq\dots\leq t_N
    \end{equation*}
    for which the following holds:
    \begin{itemize}
        \item For every edge task $\task_i = \tuple{\edge, \edge_\dur}\in\Tasks^\agent$ there exists a move action $\move=\tuple{\edge, t_i}\in\plan^\agent$.
        \item For every vertex task $\task_i = \tuple{\vertex, \dur} \in \Tasks^\agent$, it holds that if $\dur=0$ then the agent either starts in $\vertex$ or at some time moves to it,
        \begin{equation*}
            \AgentStart^\agent = \vertex \;\vee\; \exists \move\in\plan^\agent: \target{\move}=\vertex,
        \end{equation*}
        otherwise if $\dur>0$ then the agent waits at $\vertex$ for at least $\dur$,
        \begin{equation*}
            \exists\wait\in\plan^\agent:\; \source{\wait}=\vertex \;\wedge\; \timestart{\wait} \leq t_i < t_i+\task_{i,\dur} \leq \timeend{\wait}
        \end{equation*}
    \end{itemize}
    
    \item The final action $\action_n$ is an infinite wait action ($\action_{n,\dur} = \infty$). This is to model that agents are assumed to remain idle once completing all of their given instructions.
\end{itemize}
Without loss of generality, we assume that no two subsequent actions $\action_i, \action_{i+1} \in\plan^\agent$ are both wait actions, since they can equivalently be merged into one.

\textbf{Conflicts} are fundamentally defined by a function $\inconflict{i, \state_i, j, \state_j} \in \{\True,\False\}$, denoting if agent $i\in\Agents$ in state $\state_i\in\StateSpace^i$ conflicts with agent $j$ in state $\state_j\in\StateSpace^j$ or not.
We leave it implicit that no agent can conflict with itself.
This problem formulation makes no assumptions on the reason for such a joint state to be a conflict.
A conflict in the typical MAPF setting would mean that the agents' physical volumes intersect each other---a \emph{geometric collision}.
However, this formulation also allows for other types of conflicts.
It may for instance be that $i$ and $j$ are both airborne drones, and this configuration unacceptably places one agent in the downwash of the other. 
All application-specific reasons for joint states being a conflict or not are abstracted into $\inconflictfunc$. 
Thus, this formulation generalizes beyond agent volumes.

A \emph{conflict} $\tuple{\action^i, \action^j}$ is defined to mean that the states of $i$ and $j$ conflict while respectively executing $\action^i$ and $\action^j$.
That is, for some time $t$ when both $\action^i_\traj$ and $\action^j_\traj$ are defined, $\inconflict{i, \action^i_\traj(t), j, \action^j_\traj(t)}$ holds.
We use the mapping $\isconflict{\action^i, \action^j}\in\{\True, \False\}$ (\emph{is} conflict instead of \emph{in} conflict) to denote if $\tuple{\action^i, \action^j}$ is a conflict or not.
Similarly, we abuse the notion to let $\isconflict{\plan^i, \plan^j}$ denote if $\exists \tuple{\action^i, \action^j}\in\plan^i\times\plan^j$ such that $\isconflict{\action^i, \action^j}$.

\textbf{A joint plan}
$\JointPlan = \left\{ \plan^\agent \mid \agent\in\Agents \right\}$ contains one plan for each agent.
If all plans $\plan^\agent\in\JointPlan$ are valid, then $\JointPlan$ is valid.

\textbf{A solution}
is a joint plan $\JointPlan$ for which $\neg\isconflict{\plan^i, \plan^j}$ for all $\plan^i,\plan^j\in\JointPlan$.
An \emph{optimal solution} $\JointPlan^*$ is a solution minimising a \emph{global cost} function $\obj$ over all solutions.

As discussed in Section~\ref{sec:Background:CBS}, 
the correctness guarantees of \OCCBS rely on properties~\ref{OCCBS:property:soundness}--\ref{OCCBS:property:progress} being satisfied.  
This requires $\CTnodeLB{\CTnode} = \obj(\CTnode_\JointPlan)$ for every CT node $\CTnode$.
That is, no solution consistent with $\CTnode_\CTconst$ can have a lower cost than $\obj(\CTnode_\JointPlan)$.
For this to be satisfied, two things must hold.
First, we assume that $\obj(\JointPlan)$ is decomposable into a function over, and non-decreasing in, individual plan costs $\cost(\plan), \plan\in\JointPlan$.
Second, we assume that a low-level planner is available to compute a single-agent plan $\plan$ minimising $\cost(\plan)$ over all plans consistent with a given set of constraints $\CTnode_\CTconst$.
In our low-level planner, \oursTPSIPP (Section~\ref{sec:Method:TPSIPP}), our assumptions are consistent with the majority of the MAPF literature:
action costs equal their duration, $\action_\obj = \action_\dur$, and
a plan's cost is equal to the sum of its actions' costs, $\cost(\plan) = \sum_{\action\in\plan} \action_\obj$.
However, any other cost definitions may require adapting the low-level planner to guarantee that it returns an optimal single-agent plan.

\subsection{Model Generality}
\label{sec:ProblemFormulation:generality}

Based on the problem formulation above, each agent moves on a possibly unique graph $\Graph^\agent$ defined over a possibly unique state space $\StateSpace^\agent$. 
Nothing in this definition ties $\StateSpace^\agent$ to a spatial pose, a rigid body, or even a physical robot: 
a state space is an arbitrary set, and a graph is nothing more than states and inter-state trajectories the agent can realize. 
The formulation is therefore agnostic to what an agent \emph{is}---it may be a vehicle, a robotic arm operating in joint-configuration space, an aircraft, or an entity with no physical embodiment at all, such as a scheduled job for which $\inconflictfunc$ encodes resource contention rather than spatial overlap.

All differential and kinematic constraints on an agent's motion---nonholonomy, actuation limits, dynamics, whatever they may be---are captured entirely by which vertices and edges are present in $\Graph^\agent$, not by any separate mechanism in the plan or solution definitions.
Kinodynamic feasibility is thus a property of graph construction, not of planning: 
once $\Graph^\agent$ is fixed, every edge in it is by definition traversable, and the rest of the formulation (actions, tasks, plans) reasons only over this discrete abstraction. 
This separates the agent-specific, possibly continuous problem of generating feasible motion primitives from the agent-agnostic, discrete problem of coordinating them.

Conflicts are likewise defined without reference to what an agent is, only to the states two agents occupy: $\inconflict{i,\state_i,j,\state_j}$ takes just a pair of states as argument, internalizing whatever embodiment-specific test (e.g., geometric intersection of swept volumes, a downwash cone, an interference footprint)
determines whether that pairing is acceptable. 
Agent volumes are thus not primitive to the formulation, only one possible implementation of $\inconflictfunc$ for a particular class of agents. 
This is what admits heterogeneous agents in the first place: 
two agents need not share a state space, a motion model, or a notion of physicality, only a well-defined $\inconflictfunc$ function relating their respective states.

We restrict conflicts to be pairwise, and take $\inconflictfunc$ to be symmetric: $\inconflict{i,\state_i,j,\state_j} = \inconflict{j,\state_j,i,\state_i}$ for all arguments. 
This is not a claim that both agents are equally \emph{affected}, only that the joint configuration itself is unacceptable, regardless of which agent, if either, bears the consequence.
A drone flying through another's downwash may disturb only the former; the joint behavior is still what is considered unacceptable.
Symmetry follows because a conflict records only whether unacceptable behavior occurred, not who caused it or who suffered it.

Because each agent's state space may be unique, defining $\inconflict{i,\state_i,j,\state_j}$ requires being able to relate $\state_i$ and $\state_j$ to each other: 
it must be stated concretely what about their joint occupation of these states is unacceptable. 
We take this as a modeling obligation rather than a limitation of the formulation: 
a conflict that cannot be described in terms of the two agents' states is not a well-defined conflict.
%, regardless of whether a common representation happens to be convenient. 
This obligation grows with the number of distinct agent-type pairings in $\Agents$, but that cost is intrinsic to any formulation admitting heterogeneous agents and non-geometric conflicts, not an artifact of how this formulation is posed.

The final action of a valid plan is required to be an infinite wait action at a waitable vertex.
This requirement may seem like an artificial restriction; 
after all,
an agent having completed all of its tasks is free to do what it wants thereafter.
We impose it deliberately as a hard stand-in for a softer requirement that a solution must not leave the system in a configuration from which a conflict could still arise later. 
Terminating every plan in an indefinitely-waitable state (i.e., leaving the system at rest) closes off that possibility without requiring the solver to reason explicitly about an unbounded horizon past each agent's last action.

\section{Anytime-Optimal Continuous-Time Conflict-Based Search}
\label{sec:AOCCBS}

\ours extends \OCCBS to a richer problem formulation, with the main difference lying in the high-level search.

\OCCBS targets only optimal solutions. At each iteration, it selects the CT node minimizing the objective function for expansion. If that node is a solution, it is returned as optimal; otherwise, it spawns two children and the iteration ends. This continues until an expanded node represents a solution.

\OCCBS's design leaves much of the information already present in the CT unused. 
A node $\CTnode$ is selected because $\CTnodeLB{\CTnode}$ is currently known solution lower bound.
During expansion, two children are spawned, \OCCBS does not check whether any of these children represent a solution;
it a child does represent a solution $\JointPlan$, and $\obj(\JointPlan)=\CTnodeLB{\CTnode}$, then already here we know that $\JointPlan$ is optimal.
Instead of halting with this optimal solution, \OCCBS will continue iterating until the child node is selected.
Even if $\obj(\JointPlan) > \CTnodeLB{\CTnode}$, $\JointPlan$ may still be optimal, only that its optimality has not yet been confirmed.
Regardless, $\JointPlan$ can be made available as an incumbent with an optimality gap upper bound $\obj(\JointPlan) / \CTnodeLB{\CTnode}$.
Furthermore, \OCCBS is also inherently sequential, expanding one CT node at a time, which does not allow for the use of parallel computations.

\ours addresses both limitations. 
First, the child nodes at every expansion are checked for a solution; 
if one represents a solution that is better than the incumbent, then it replaces the incumbent.
Second, \ours allows for the expansion of multiple CT nodes simultaneously on separate processors, with the results of these expansions aggregated to maintain correctness guarantees (Section~\ref{sec:TheoreticalGuarantees} provides correctness proofs).
Informally, the nodes are selected and expanded according to a configurable \emph{portfolio}.
This portfolio consists of an arbitrary number of \emph{policies}, with each policy managing the expansion of one node. 
The policy operations are independent of each other, thus, opening the possibility for parallel computations.
A policy specifies a \emph{selection policy} that determines which of the CT leaf nodes are selected, and a \emph{set of repair methods} that are applied to the CT node's joint plan, in addition to the regular branching using \deltaBR.
\ours is agnostic to the selection policy and repair methods, hence it remains open to future repair methods. 
These kinds of configurable methods have been widely used in the literature, such as RHCR~\cite{RHCR} and LNS~\cite{li2021lns}, 
while the ability to incorporate other CT node selection policies than by minimum cost lower bound has been useful when searching for sub-optimal solutions fast~\cite{CohenAnytimeFocal, ECBS-CT}.

% Also common for these problems is to assume that each agent $\agent$ is given exactly one goal vertex at which a valid plan must end---consequently the cost $\cost(\plan_\agent)$ being the arrival time at the goal vertex coincides with the duration of $\plan_\agent$.
% In this problem, however, tasks consist of multiple sub-tasks and agent may continue to move after completing their task, such as to avoid obstructing other, active agents.
% For this reason, a plan's duration does not necessarily coincide with the task completion time.
% Considering this, there are multiple possible and practically motivated cost functions. 
% Some applications may value completing tasks as soon as possible, while others may aim to reduce the time that agents spend engaged with a plan.

\subsection{High-Level Search}
\label{sec:AOCCBS:highlevel}

The high-level search is done over the CT using a modular search portfolio $\Portfolio = \tuple{\policy^1, \policy^2, \dots, \policy^p}$ containing a finite number of policies.
Each policy $\policy^i\in\Portfolio$ may be unique, consisting of a selection policy $\policy^i_\policySelection$ and a set of repair methods $\policy^i_\policyRepairSet$:
\begin{itemize}
    \item \textbf{Selection policy}: $\policy^i_\policySelection$ controls which unexpanded CT node is selected for expansion. In \OCCBS, the default policy is to expand a node minimizing the global cost lower bound $\CTnodeLB{\CTnode}$.
    In contrast, \ours allows for any selection policy, including, for instance, random selection, selection by fewest collisions, selection by least amount of time in conflict, and so on. Note however that \ours's correctness guarantees (Section~\ref{sec:TheoreticalGuarantees}) requires at least one selection policy to be by minimum $\CTnodeLB{\CTnode}$, that is until an incumbent solution has been found, after which any selection policy can be used.
    
    \item \textbf{Set of repair methods}: $\policy^i_\policyRepairSet = \set{r^{i,1}, r^{i,2},\dots}$ contains a finite number of repair methods. Each repair method $r^{i,j}\in\policy^i$ takes an existing joint plan as input and produces either a solution---which may or may not be optimal---or a report of failure.
\end{itemize}
All policies in $\Portfolio$ are independent of each other, and all repair functions in $\policy_\policyRepairSet$ are also independent, allowing for parallel execution.

At iteration $k$ of \ours, the set $\NodeQueue_k$ contains all CT leaf nodes, and $\incumbentSolution_k$ denotes the best solution found so far, with $\objectiveUB_k = \obj\left(\incumbentSolution_k\right)$ if a solution has been found and $\objectiveUB_k = \infty$ otherwise.
Initially $\NodeQueue_0 = \left\{ \CTroot\right\}$ contains the root node.
Let $\objectiveLB_k = \min_{\CTnode\in\NodeQueue_k} \CTnodeLB{\CTnode}$ denote the minimum cost lower bound over all leaf nodes if $\NodeQueue_k\neq\varnothing$, otherwise $\objectiveLB_k = \objectiveUB_k$.
Each policy $\policy\in\Portfolio$, in turn, selects a node from $\NodeQueue_k$ to expand according to its selection rule $\policy_\policySelection$; 
once selected, a node becomes unavailable to subsequent policies. 
We denote by $\NodeSelection$ the set of nodes selected this way. 
This continues until every policy has selected a node or no selectable nodes remain in $\NodeQueue_k$.
Each policy $\policy$ then independently applies \deltaBR to its selected node, spawning two children; the set of all such children is denoted $\NewNodes$. 
Each policy also applies its repair methods in $\policy_\policyRepairSet$
to its selected CT node,
% ---depending on configuration, either to the parent's joint plan or to each child's joint plan---
and additionally checks whether each child's joint plan is a solution. 
Any solutions produced this way, whether through repair or from the children, are collected in the set $\NewSolutions$.
After all selected nodes have been processed, 
the incumbent solution for the next iteration is
\begin{equation*}
    \incumbentSolution_{k+1} = \argmin{\JointPlan \in \NewSolutions \cup \{\incumbentSolution_k\}} \obj\left( \JointPlan \right)
\end{equation*}
with the new upper bound $\objectiveUB_{k+1} = \obj\left(\incumbentSolution_{k+1}\right)$ and the set of unexpanded nodes
\begin{equation*}
    \NodeQueue_{k+1} = \left\{ \CTnode \in (\NodeQueue_{k} \setminus \NodeSelection) \cup \NewNodes \mid \CTnodeLB{\CTnode} < \objectiveUB_{k+1} \right\}
\end{equation*}
where all nodes with global cost lower bound greater or equal to $\objectiveUB_{k+1}$ pruned. 

Section~\ref{sec:TheoreticalGuarantees} shows that at every iteration $k$ of \ours on a solvable problem instance with optimal solution $\JointPlan^*$ and optimal objective value $\obj^*$,
it holds that 
\begin{itemize}
    \item $\objectiveLB_k \leq \obj^* \leq \objectiveUB_k$;
    \item $\gapUB = \objectiveUB_k / \objectiveLB_k - 1$ is a true optimality gap upper bound on $\incumbentSolution_k$;
    \item if $\NodeQueue_k = \varnothing$ then $\obj(\incumbentSolution_k)=\obj^*$ and $\incumbentSolution_k$ is optimal; otherwise
    \item if $\NodeQueue_k \neq \varnothing$ then $\NodeQueue_k$ contains a CT node from which $\JointPlan^*$ is reachable.
\end{itemize}
It is also shown that if either an incumbent solution has been found or at least one policy $\policy\in\Portfolio$ always selects a node $\CTnode\in\NodeQueue_k$ minimizing $\CTnodeLB{\CTnode}$, then \ours is guaranteed to return an optimal solution in a finite number of iterations.

Once an initial incumbent solution has been found, \ours provides a solution at any time with a known upper bound on its optimality gap. 
Given time, 
$\objectiveLB_k$ will progressively increase due to the expansion of CT nodes, 
and $\objectiveUB_k$ will progressively decrease due to better solutions being found (be that via the repair functions or due to the expansion of the CT nodes). 
Within finitely many iterations, these bounds will meet and an optimal solution will be found.

We make a final note on the portfolio: $\Portfolio$ does not need to be static over a search; each iteration of the search may use a unique portfolio. 
For instance, Section~\ref{sec:Experiments} details experiments using 
\begin{itemize}
    \item a portfolio for quickly finding an incumbent (mixed selection policies for increased diversity in the nodes being selected, and every policy using a repair function),
    \item then, at an increasing rate as \gapUB decreases, using a portfolio for raising the global solution cost lower bound $\objectiveLB$ as fast as possible (all policies selecting by minimum cost lower bound $\CTnodeLB{\CTnode}$, no repair function to decrease iteration time).
\end{itemize}
This provides merely one example of many possible non-static portfolios.

\subsubsection{Cardinal Conflicts}
\label{sec:AOCCBS:highlevel:conflict_selection}

A \emph{cardinal} conflict $\tuple{\action^i, \action^j}$ is a conflict which upon branching results in an increase in the plan costs $\cost(\plan^i)$ and $\cost(\plan^j)$ for each agent $i$ and $j$~\cite{ICBS}.
A \emph{semi-cardinal} conflict results in an increase for only one agent, 
and a \emph{non-cardinal} conflict does not increase either.
Branching on non-cardinal conflicts increases the number of leaf nodes without raising the CT objective lower bound.
Cardinal conflicts, on the other hand, contribute to raising the objective lower bound and are therefore preferred.  
When branching on a CT node, our \ours implementation supports four modes:
\begin{itemize}
    \item \ConflictSelectionArbitrary: selects an arbitrary conflict.
    \item \ConflictSelectionSemiCard: selects an arbitrary semi-cardinal or cardinal conflict. If none exist, an arbitrary conflict is selected.
    \item \ConflictSelectionBestof{n}: arbitrarily selects a set of at most $n$ conflicts and selects the best cardinal conflict, falling back to the best semi-cardinal conflict, falling back to an arbitrary conflict. We define the best conflict as one maximising over all conflicts the minimum plan cost increase over the two resulting child nodes.  
\end{itemize}
Note that $n=\infty$ returns the best conflict over all conflicts at a node, and $n=1$ reduces to arbitrary selection.
A lower $n$ makes the conflict selection faster, while a higher $n$ increases the chance of selecting a better conflict.

\subsection{Low-Level Path Planner}
\label{sec:Method:low-level}

Our low-level path planner \emph{Multi-Goal CSIPP} (\oursPP) extends CSIPP~\cite{CCBS}, itself based on SIPP~\cite{SIPP}, however incorporating the treatment of multiple sequential tasks in the style of~\cite{RHCR}, while also allowing for the handling of idle agents. 

As discussed in Section~\ref{sec:ProblemFormulation}, for the correctness guarantees to hold, the low-level path planner must find for a given agent the lowest-cost plan that is consistent with a given set of constraints. 
This means that the path planner is coupled with the definition of a plan's cost.
\oursPP assumes that an action's cost is equal to its duration, and that a plan's cost is equal to the sum of its actions' costs. 
Modifications may be required if other plan cost definitions are used. 

In \oursPP, when planning for agent $\agent\in\Agents$, 
a state $\tuple{\vertex, \SIPPsafeI_\vertex, \task}\in\Vertices^\agent\times(\Real\times\Real)\times(\Tasks^\agent\cup\{\None\})$
contains a vertex $\vertex$ and a safe interval $\SIPPsafeI_\vertex = [\SIPPstarttime_\vertex, \SIPPendtime_\vertex)$, as in SIPP and CSIPP, 
but also the next task $\task_i$ the agent must complete ($\None$ if no tasks remain). 
The \Astar search seeks a path from the agent's current state to a state $\tuple{\vertex,[\cdot,\infty),\None}$ with $\waitable{\vertex}$, i.e., a state in which all tasks are completed and the agent can safely remain at $\vertex$ indefinitely.

From a state $\tuple{\vertex, \SIPPsafeI_\vertex, \task_i}$ with arrival time $\SIPParrivalTime$, two kinds of transitions are possible:
\begin{enumerate}
    \item \emph{Move transitions} are constructed exactly as in CSIPP:
    for every edge $\edge$ with $\source{\edge}=\vertex$, the set of safe times to traverse $\edge$ and arrive at $\vertex' = \target{\edge}$ is
    \begin{equation}
        P_\edge = [\SIPParrivalTime, \SIPPendtime_\vertex) \cap S_\edge \cap S_{\vertex'}^{+\edge_\dur}
    \end{equation}
    where 
    $S_\edge$ is the set of safe times to traverse $\edge$ and 
    $S_{\vertex'}^{+\edge_\dur}$ is the set of safe times to occupy $\vertex'$ but shifted by $\edge_\dur$ to account for the agent arriving there only after traversing $\edge$.
    A transition to 
    $\tuple{\vertex', \SIPPsafeI_{\vertex'}, \task_i}$ 
    exists at the earliest time of each maximally connected interval in $P_\edge$.
    If $\task_i$ is a vertex task at $\vertex$, $P_\edge$ is additionally restricted to $[\SIPParrivalTime, \SIPParrivalTime + \task_{i,\dur})$:
    at $\SIPParrivalTime + \task_{i,\dur}$ the task is completed regardless, so later departures are already captured at state $\tuple{\vertex, \SIPPsafeI_\vertex, \task_{i+1}}$. 
    Otherwise, if $\task_i$ is an edge task on $\edge$, then all states reached by a transition along $\edge$ lead to a state containing the next task $\task_{i+1}$ (or $\None$ if $\task_i$ was the last).
    \item A \emph{vertex task completion transition} exists if $\task_i$ is a vertex task at $\vertex$ and $\SIPParrivalTime + \task_{i,\dur} < \SIPPendtime_\vertex$.
    This is a single transition to $\tuple{\vertex, \SIPPsafeI_\vertex, \task_{i+1}}$ with arrival time $\SIPParrivalTime + \task_{i,\dur}$.
    Move and task-completion transitions out of this new state are generated in the same way as any other state, so a task may be completed and the agent may immediately depart within the same safe interval.
\end{enumerate}

Let $\MinCost(\vertex,\vertex')$ be the minimum possible time to move from $\vertex\in\Vertices^\agent$ to $\vertex'\in\Vertices^\agent$ ignoring constraints and unsafe intervals. 
We define the heuristic at state $\tuple{\vertex, \SIPPsafeI_\vertex, \task_i}$ with arrival time $\SIPParrivalTime$ as
\begin{align*}
    f\left(\tuple{\vertex, \SIPPsafeI, \task_i}\right) &= \SIPParrivalTime + \MinCost\left(\vertex, \source{\task_i}\right) + \task_{i,\dur} \\ 
    &+ \sum_{j=i+1}^{|\Tasks^\agent|} \MinCost\left(\target{\task_{j-1}}, \source{\task_j}\right) + \task_{j,\dur}
\end{align*}
which captures the shortest possible time from state $\tuple{\vertex, \SIPPsafeI_\vertex, \task_i}$ for the agent to complete all remaining tasks.
In common \Astar terminology, $f = g + h$ where $g = \SIPParrivalTime$ is the cost so far and $h$ is the sum of the remaining terms with captures a lower bound on the remaining cost to a goal state.
% That is, $f$ lower-bounds the time to reach a goal state, since it corresponds to moving directly to $\task_i$ and completing all remaining tasks in sequence while ignoring unsafe intervals; no valid plan can do better. Hence $f$ is admissible.

\oursPP retains the two properties that make SIPP optimal: 
(i) the safe-interval decomposition is exact, so the state space contains every arrival time at which a distinct set of future transitions becomes available, and no reachable plan is excluded; 
(ii) $f$ is admissible since the estimated remaining cost to a goal state cannot be better than taking the shortest path through all tasks. 
Since \oursPP is an \Astar search over this state space with an admissible heuristic, it is guaranteed to return a plan of minimum duration (i.e., the earliest possible completion time of all tasks) whenever one exists.

A state with $\SIPPsafeI_\vertex=[\cdot,\infty)$ and $\task=\None$ requires no further move action and is a valid goal; an infinite-wait action is implicitly available there.

% \begin{enumerate}
%     \item Include pseudo-code
% \end{enumerate}

\subsection{Repair Method: Tier-Prioritized Safe Interval Path Planning}
\label{sec:Method:TPSIPP}

Tier-Prioritized SIPP (\oursTPSIPP) resembles PSIPP~\cite{PSIPP} in that both apply prioritized planning. Although PSIPP does not natively support multiple vertex and edge tasks, each with durations, this is easy to address by modifying the low-level planner, as shown with \oursPP in the previous section. 
PSIPP is nonetheless unsuitable here since it lacks a mechanism for moving an idle agent out of the way and instead plans each agent to a single, specific vertex.
For instance, consider a problem where two agents are each assigned a vertex task of non-infinite duration, and the two tasks conflict with each other. Without modification, PSIPP first plans one agent to its task and leaves it there, making the other agent's task inaccessible from the first agent's arrival time onward. 
When PSIPP then plans the second agent, it can only succeed if the second agent can reach its task, complete it, and clear the vertex before the first agent arrives---otherwise PSIPP fails to find a solution, even though one exists: the first agent could simply move out of the way once its own task is complete. Trying a different planning order does not fix this in general, and the number of possible orders grows exponentially in the number of agents.

\newcommand{\Deadline}{\text{Deadline}}
\newcommand{\Pin}{\text{Pin}}
\newcommand{\PlanTier}{\text{PlanTier}}
\newcommand{\PlanAgent}{\text{PlanAgent}}
\newcommand{\Failed}{\text{Failed}}
\newcommand{\Success}{\text{Success}}
\begin{algorithm}
\caption{Tier-Prioritized SIPP}
\label{alg:TPSIPP}
    \begin{algorithmic}[1]
        \Function{\oursTPSIPP}{$\JointPlan$}
            \State $P \gets \{ \agent\in\Agents \mid \plan^\agent\in\JointPlan \text{ is conflict-free}\}$ \Comment{Planned agents, mutable below} \label{alg:TPSIPP:b1s}
            \State $\JointPlan' \gets \{\plan^\agent\in\JointPlan \mid \agent\in P\}$ \Comment{Decided plans, mutable below}
            \State $U \gets \Agents \setminus P$ \Comment{Unplanned agents} \label{alg:TPSIPP:b1e}
            \State
            \For{$\agent \in U$} \Comment{Deadlines reflect only planned agents in $P$} \label{alg:TPSIPP:b2s}
                \State $\delta(\agent) \gets \Call{\Deadline}{\agent, \JointPlan'}$ \label{alg:TPSIPP:b2e}
            \EndFor
            \State $U^{0} \gets \Call{RandomOrder}{\tuple{\agent\in U \mid \delta(\agent) = 0}}$ \label{alg:TPSIPP:b3s}
            \State $U^{(0, \infty)} \gets \Call{IncreasingOrder}{\tuple{ \agent\in U \mid 0 < \delta(\agent) < \infty }, \delta}$
            \State $U^{\infty} \gets \Call{RandomOrder}{\tuple{ \agent\in U \mid \delta(\agent) = \infty }}$ \label{alg:TPSIPP:b3e}
            \State
            \If{$\Call{\PlanTier}{U^{0}} = \Failed$} \Comment{Plan severely constrained agents}
                \State \Return \Failed
            \EndIf
            \If{$\Call{\PlanTier}{U^{(0,\infty)}} = \Failed$} \Comment{Plan constrained agents}
                \State \Return \Failed
            \EndIf
            \If{$\Call{\PlanTier}{U^\infty} = \Failed$} \Comment{Plan unconstrained agents}
                \State \Return \Failed
            \EndIf
            \State \Return $\JointPlan'$ \label{alg:TPSIPP:return}
        \EndFunction
        \State
        \Function{\Deadline}{$\agent, \JointPlan'$}
            \If{$\neg\waitable{\AgentStart^\agent}$}
                \Return $0$
            \EndIf
            \State \Return the safe interval end at $\AgentStart^\agent$ starting at $t=0$, w.r.t.\ $\JointPlan'$
        \EndFunction
        \State
        \Function{\PlanTier}{$U'$}
            \State marker$=\None$
            \While{$U'\neq\varnothing$}
                \State $\agent = U'.\Call{PopFront}{}$ \label{alg:TPSIPP:plan1}
                \If{marker$=\agent$}
                    \Return \Failed \Comment{$U'$ cycled without change}
                \EndIf
                \State
                \State $\plan^\agent = \Call{\PlanAgent}{\agent, \JointPlan', P}$ \label{alg:TPSIPP:plan2}
                \If{$\plan^\agent \neq \Failed$}
                    \State $\JointPlan'[\agent] \gets \plan^\agent$
                    \State $P \gets P \cup \{\agent\}$
                    \State marker$=\None$
                \Else
                    \State $U'.\Call{AppendToBack}{\agent}$
                    \If{marker$=\None$}
                        marker$=\agent$
                    \EndIf
                \EndIf
            \EndWhile
            \State \Return \Success
        \EndFunction
    \end{algorithmic}
\end{algorithm}

\oursTPSIPP is a repair function that takes a valid but non-solution joint plan $\JointPlan$ as input and returns either a repaired solution $\JointPlan'$ or failure.
The pseudo-code is found in Algorithm~\ref{alg:TPSIPP}, on which the following description relies.
Since $\JointPlan$ is valid, no conflict-free plan $\plan \in \JointPlan$ blocks another agent from completing its tasks. 
If it did, $\plan$ would conflict with that agent's own valid plan and therefore not be conflict-free. 
So every conflict-free plan in $\JointPlan$ is copied directly into $\JointPlan'$, and only the conflicting agents are left unplanned (lines~\ref{alg:TPSIPP:b1s}--\ref{alg:TPSIPP:b1e}); 
the better $\JointPlan$ is to begin with, the fewer repairs are needed.
Each unplanned agent has a deadline: the latest time it may still occupy its starting vertex, determined by whether that vertex is waitable and, if so, whether some plan already in $\JointPlan'$ eventually conflicts with it (lines~\ref{alg:TPSIPP:b2s}--\ref{alg:TPSIPP:b2e}). 
Unplanned agents are grouped into three tiers by this deadline: $t=0$
(severely constrained), $t \in (0, \infty)$ (constrained), and $t = \infty$ (unconstrained) (lines~\ref{alg:TPSIPP:b3s}--\ref{alg:TPSIPP:b3e}).
Starting with the severely constrained tier, the agents are arbitrarily ordered and an agent $\agent$ is selected (line~\ref{alg:TPSIPP:plan1}).
A valid plan $\plan^\agent$ is sought using \oursPP that 
(i) avoids conflicts with plans already in $\JointPlan'$, 
(ii) avoids conflicting with any still-unplanned agent at that agent's starting vertex before its deadline, and 
(iii) does not end at a vertex that conflicts with any still-unplanned agent's tasks (line~\ref{alg:TPSIPP:plan2}). 
Together, these conditions guarantee that $\plan^\agent$ leaves every remaining unplanned agent free to use its full allotted time at its start, and does not permanently block any of their tasks. 
If such a plan is found, it is added to $\JointPlan'$ and $\agent$ is no longer unplanned. 
However, unlike PSIPP, if such a plan cannot be found then instead of reporting failure, \oursTPSIPP requeues the agent and the tier continues, so failure only occurs once every agent in the tier has been cycled through without a single plan being found.
Once the entire tier is planned, the same procedure repeats for the constrained tier (agents taken in increasing order of deadline) and then the unconstrained tier (agents taken in arbitrary order).
Since plans are only ever added to $\JointPlan'$ if they do not conflict with any existing plan in $\JointPlan'$, we know that $\JointPlan'$ is a solution if it contains a plan for every agent.
At that point, $\JointPlan'$ is returned (line~\ref{alg:TPSIPP:return}).

When using \oursTPSIPP as a repair function in the high-level search's portfolio, it is applied to many different CT nodes and their possibly unique joint plans.
Additionally, the severely constrained and unconstrained agents are ordered arbitrarily and can therefore be ordered, e.g., randomly. 
This introduces substantial variation in how repairs unfold across calls and therefore provides many chances throughout solving a problem to find a solution.

\section{Theoretical Guarantees}
\label{sec:TheoreticalGuarantees}

In this work, we say that an algorithm is
\begin{itemize}
    \item \textbf{sound} if it only returns solutions,
    \item \textbf{exact} if it only returns optimal solutions, and 
    \item \textbf{solution complete} if it returns a solution in finite time if one exists.
\end{itemize}
Observe that exactness implies soundness, and that a solution complete algorithm is not required to report the non-existence of a solution while a \emph{complete} algorithm is.
We show in the section that \ours is exact and solution complete.

Just as in \OCCBS, this work makes explicit that real values cannot be represented exactly under finite-precision arithmetic. 
Since our implementation (like all implementations) is subject to a precision $\epsmach$, our guarantees are given in that context. 
Thus, it is possible that a solver in this setting is unable to find solutions that would require a finer precision than what the underlying machine can provide, 
or that a returned solution actually contains conflicts on the time scale of a few orders of magnitude above the machine's finest precision. 
Nonetheless, we argue that this accuracy is well beyond what most real-world applications require, where, for instance, sensor and actuator noise would dominate. 
In Section~\ref{sec:Implementation:conflict_information}, we discuss how conflict information is practically computed and stored for use during runtime.
Since this conflict information is the sole connection between \ours and the specific application that it plans over, all correctness guarantees are with respect to the accuracy of this conflict information.

The remainder of this section provides various useful properties of \ours, including the conditions required for guaranteed termination in a finite number of iterations with an optimal solution.
The invariant (Theorem~\ref{theorem:AOCCBS:invariant}) shows that at every iteration, the search either still has access to an optimal solution through some unexpanded leaf node, or has already found one; 
termination (Corollary~\ref{cor:AOCCBS:Termination}) shows that once the node queue empties, this invariant forces the incumbent to be optimal; 
and bound correctness (Corollary~\ref{cor:AOCCBS:BoundCorrectness}) confirms that the lower and upper bounds maintained throughout the search are valid, giving a meaningful optimality gap at any point of interruption.
The remaining two theorems then give the conditions under which an optimal solution guaranteed in a finite number of iterations:
either an incumbent solution has been found (Theorem~\ref{theorem:AOCCBS:IncumbentDrivenTermination}), 
or at least one policy in the portfolio always expands the node minimizing the lower bound (Theorem~\ref{theorem:AOCCBS:PolicyDrivenTermination}), inheriting termination guarantees from \OCCBS.

All proofs assume that \ours is applied to a solvable problem instance, with  $\JointPlan^*$ being an optimal solution and $\obj^* = \obj(\JointPlan^*)$.
Recall the following notation from Sections~\ref{sec:Background:OCCBS} and~\ref{sec:AOCCBS:highlevel}:
for a CT node $\CTnode$, 
$\CTreachableSolutions{\CTnode}$ is the set of all solutions permitted under $\CTnode_\CTconst$;
$\CTnodeLB{\CTnode} \leq \min_{\JointPlan\in\CTreachableSolutions{\CTnode}} \obj(\JointPlan)$ is a cost lower bound over all solutions in $\CTreachableSolutions{\CTnode}$;
$\incumbentSolution_k$ is the best-found incumbent solution at iteration $k$, with $\objectiveUB_k = \obj(\incumbentSolution_k)$;
$\objectiveLB_k = \min_{\CTnode\in\NodeQueue_k} \CTnodeLB{\CTnode}$ if $\NodeQueue_k\neq\varnothing$, otherwise $\objectiveLB_k = \objectiveUB_k$;
and $\gapUB_k = \objectiveUB_k / \objectiveLB_k - 1$ is an upper bound on $\incumbentSolution_k$'s optimality gap.

\begin{theorem}[\ours Invariant]
    \label{theorem:AOCCBS:invariant}
    At the start of every iteration $k$ of \ours, 
    either
    $\exists \CTnode^* \in \NodeQueue_k: \JointPlan^* \in \CTreachableSolutions{\CTnode^*}$ (Reachability),
    $\CTnodeUB_k=\obj^*$ (Optimality),
    or both are true.
\end{theorem}
\begin{proof}
    The invariant is proven by induction:
    \begin{itemize}
        \item \textit{Base Case}: 
        At $k=0$ we have $\NodeQueue_0 = \{\CTroot\}$.
        Since the empty constraint set $\CTroot_\CTconst=\varnothing$ permits all solutions, $\JointPlan^* \in \CTreachableSolutions{\CTroot}$.
        Reachability holds at $k=0$.
    
        \item \textit{Inductive step}:
        Assume the invariant holds at the start of iteration $k$.
        Recall from Section~\ref{sec:AOCCBS:highlevel} that during iteration $k$, 
        a set of nodes $\NodeSelection\subseteq\NodeQueue_k$ is selected for expansion, 
        producing a set of child nodes nodes $\NewNodes$ and a set of newly found solutions $\NewSolutions$.
        At the start of iteration $k+1$, the incumbent solution is
        \begin{equation*}
            \incumbentSolution_{k+1} = \argmin{\JointPlan \in \NewSolutions \cup \{\incumbentSolution_k\}} \obj\left( \JointPlan \right)
        \end{equation*}
        and the set of unexpanded nodes is
        \begin{equation*}
            \NodeQueue_{k+1} = \left\{ \CTnode \in (\NodeQueue_{k} \setminus \NodeSelection) \cup \NewNodes \mid \CTnodeLB{\CTnode} < \CTnodeUB_{k+1} \right\}
        \end{equation*}
        where all nodes $\CTnode$ with $\CTnodeLB{\CTnode} \geq \CTnodeUB_{k+1}$ are pruned.
        There are three cases:
        \begin{enumerate}
            \item If $\CTnodeUB_k = \obj^*$, then $\CTnodeUB_{k+1} = \obj^*$ 
            since no solution can have a cost lower than $\obj^*$. Optimality holds at the start of iteration $k+1$.
            
            \item If instead $\CTnodeUB_k > \obj^*$ but during the expansion a solution with cost $\obj^*$ was found, then $\CTnodeUB_{k+1} = \obj^*$. Optimality holds at the start of iteration $k+1$.
            
            \item Otherwise, $\CTnodeUB_k \geq \CTnodeUB_{k+1} > \obj^*$. 
            Then optimality does not hold, and therefore by the invariant there exists $\CTnode^* \in \NodeQueue_k$ such that $\JointPlan^* \in \CTreachableSolutions{\CTnode^*}$. This gives two subcases:
            \begin{enumerate}[label=(\roman*)]
                \item \textbf{$\CTnode^*$ is not expanded}: $\CTnode^*\not\in\NodeSelection$.
                Since $\JointPlan^*\in\CTreachableSolutions{\CTnode^*}$ and $\obj^* <\CTnodeUB_{k+1}$, 
                \begin{equation*}
                    \CTnodeLB{\CTnode^*} \leq \obj(\JointPlan^*) = \obj^* < \CTnodeUB_{k+1}.
                \end{equation*}
                Hence, $\CTnode^*$ was not pruned and therefore $\CTnode^*\in\NodeQueue_{k+1}$.
                Reachability holds at the start of iteration $k+1$.
                
                \item \textbf{$\CTnode^*$ is expanded}: $\CTnode^*\in\NodeSelection$.
                Since no solutions are removed during the expansion of a node, 
                there exists a child $\CTnode^c$ of $\CTnode^*$ such that 
                \begin{equation*}
                    \JointPlan^*\in\CTreachableSolutions{\CTnode^c},
                \end{equation*}
                meaning that
                \begin{equation*}
                    \CTnodeLB{\CTnode^c} \leq \obj(\JointPlan^*) = \obj^* \leq \CTnodeUB_{k+1}.
                \end{equation*}
                Thus, $\CTnode^c$ is not pruned and therefore $\CTnode^c\in\NodeQueue_{k+1}$. Reachability holds at the start of iteration $k+1$.
            \end{enumerate}
        \end{enumerate}
        In all cases, at the start of $k+1$, either optimality holds or reachability holds. Thus, the invariant is preserved.
    \end{itemize}
    By induction, the invariant holds at the start of every iteration.
\end{proof}

\begin{corollary}[\ours Termination]
    \label{cor:AOCCBS:Termination}
    If $\NodeQueue_k=\varnothing$, then $\CTnodeUB_k = \obj^*$ and $\incumbentSolution_k$ is an optimal solution.
\end{corollary}
\begin{proof}
    The invariant from Theorem~\ref{theorem:AOCCBS:invariant} holds for this iteration $k$.
    However, since reachability is impossible, optimality must hold.
    That is, $\CTnodeUB_k = \obj^*$ and $\incumbentSolution_k$ is an optimal solution.
\end{proof}

\begin{corollary}[\ours Bound Correctness]
    \label{cor:AOCCBS:BoundCorrectness}
    At every iteration $k$, 
    the optimal solution lower and upper bounds $\objectiveLB_k$ and $\objectiveUB_k = \obj(\incumbentSolution_k)$,
    and the incumbent solution $\incumbentSolution_k$'s optimality gap upper bound $\gapUB_k = \objectiveUB_k/\objectiveLB_k - 1$, are correct.
\end{corollary}    
\begin{proof}
    If no incumbent solution $\incumbentSolution_k$ has been found then $\objectiveUB_k = \infty$.
    The cost of any solution cannot be greater than $\infty$.
    Otherwise, $\objectiveUB_k=\obj(\incumbentSolution_k)$, and it must hold that $\obj(\JointPlan^*) \leq \objectiveUB_k$, else $\incumbentSolution_k$ provides a counterexample. 
    Thus, at every iteration $k$, $\objectiveUB_k$ is a true upper bound on the optimal cost.

    If \emph{optimality} in the invariant of Theorem~\ref{theorem:AOCCBS:invariant} holds,
    such that $\objectiveUB_k = \obj^*$,
    then either $\NodeQueue_k=\varnothing$ and $\objectiveLB_k = \objectiveUB_k$ which is a true lower bound on $\obj^*$,
    or $\NodeQueue_k\neq\varnothing$ and $\objectiveLB_k = \min_{\CTnode\in\NodeQueue_k} \CTnodeLB{\CTnode} < \objectiveUB_k$ 
    since $\forall\CTnode\in\NodeQueue_k: \CTnodeLB{\CTnode} < \objectiveUB_k$.
    If instead \emph{optimality} in the invariant of Theorem~\ref{theorem:AOCCBS:invariant} doesn't hold, then \emph{reachability} must hold: $\exists\CTnode^*\in\NodeQueue_k: \JointPlan^*\in\CTreachableSolutions{\CTnode^*}$. Then $\objectiveLB_k = \min_{\CTnode\in\NodeQueue_k} \CTnodeLB{\CTnode} \leq \CTnodeLB{\CTnode^*} \leq \obj^*$.
    Thus, at every iteration $k$, $\objectiveLB_k$ is a true lower bound on the optimal cost.

    Finally, since $\objectiveLB \leq \obj^*$, 
    it must hold that 
    \begin{equation*}
        \frac{\obj(\incumbentSolution_k)}{\objectiveLB_k} = \frac{\objectiveUB_k}{\objectiveLB_k} \geq \frac{\objectiveUB_k}{\obj^*}
    \end{equation*}
    Thus, at every iteration $k$, $\gapUB_k = \objectiveUB_k / \objectiveLB_k - 1$ is a true upper bound on the optimality gap of the incumbent solution $\incumbentSolution_k$.
\end{proof}

\begin{theorem}[Incumbent-driven Termination]
    \label{theorem:AOCCBS:IncumbentDrivenTermination}
    If an incumbent solution $\incumbentSolution_k$ has been found,
    \ours is guaranteed to terminate in a finite number of iterations with an optimal solution.
\end{theorem}
\begin{proof}
    Termination follows from the CT in \ours satisfying property~\ref{OCCBS:property:progress} as a consequence of using \deltaBR:
    for every infinite path $\CTnode_1, \CTnode_2, \CTnode_3,\dots$ in the CT (where $\CTnode_{i}$ is the parent of $\CTnode_{i+1}$), and any $c\in\Real$, 
    \begin{equation*}
        \exists i \in \Natural: \forall j\geq i: c < \CTnodeLB{\CTnode_j}.
    \end{equation*}
    This means that for any descending path in the CT, eventually the lower bound on the cost will be pushed beyond some threshold $c$. 
    Consider here that $c = \objectiveUB_k$.
    Since all nodes $\CTnode$ with $\CTnodeLB{\CTnode} \geq \objectiveUB_k$ are pruned and therefore not included in $\NodeQueue_k$, every node in $\NodeQueue_k$ will either be pruned or expanded to eventually reach children that are pruned.
    Consequently, $\NodeQueue_K=\varnothing$ for some bounded iteration $K$.
    At that point, by the invariant in Theorem~\ref{theorem:AOCCBS:invariant}, the incumbent $\incumbentSolution_K$ will be an optimal solution.
\end{proof}

\begin{theorem}[Policy-driven Termination]
    \label{theorem:AOCCBS:PolicyDrivenTermination}
    If there exists at least one policy $\policy\in\Portfolio$ that always selects $\CTnode\in\NodeQueue_k$ minimizing $\CTnodeLB{\CTnode}$ for expansion, then \ours is guaranteed to terminate in a finite number of iterations with an optimal solution.
\end{theorem}
\begin{proof}
    If at least one such policy $\policy$ exists, 
    then at every iteration at least one node $\CTnode\in\NodeQueue_k$ minimizing $\obj$ will be expanded.
    This is a best-first search (with possibly additional node expansions) using \deltaBR, which is consistent with an \OCCBS search and therefore inherits \OCCBS's termination guarantees~\cite{OCCBS}.

    Specifically, since properties~\ref{OCCBS:property:soundness}--\ref{OCCBS:property:progress} are satisfied by the CT in \ours and a node $\CTnode\in\NodeQueue_k$ minimizing $\CTnodeLB{\CTnode}$ is selected for expansion at every iteration $k$,
    $\min_{\CTnode\in\NodeQueue_k} \CTnodeLB{\CTnode}$ will necessarily increase as $k$ grows. 
    Since no solution is removed during the search, 
    and property~\ref{OCCBS:property:progress} ensures progress is made,
    eventually an optimal solution will be found.
\end{proof}

\section{Experimental Evaluation}
\label{sec:Experiments}

This section presents \AC{a preliminary} experimental evaluation of \ours, \AC{with further testing ongoing}. 
All experiments are run on a 2025 Mac Studio, 16-core M4~Max CPU (12 performance cores, 4 efficiency cores), 64~GB RAM, macOS Tahoe~(26.5.2). 
\ours is implemented in Python; the source code, benchmark sets, results, and demonstrations \AC{will be made available.}
% on our repository\footnote{\GithubAOCCBS}. 

We demonstrate \ours on three difference MAPF-based models.
The first operates under the common assumptions of circular agents and straight-line traversals at constant speed. It is under this model that we perform an ablation study in Section~\ref{sec:Experiments:Ablations} using various solver configurations, 
with some comparison with \OCCBS,
followed by additional comparisons in Section~\ref{sec:Experiments:OCCBS_comparison}.
The second model includes arbitrary continuous and connected edge trajectories, 
and the third model further introduces polygonal (possibly non-convex) agent shapes. 
We provide a demonstration of this richer model in Section~\ref{sec:Experiments:ComplexDemonstration}.
Appendix~\ref{sec:Implementation} describes our implementation of these three models, specifically regarding their interface with \ours and how to preprocess and store information for lookup during runtime. 

Figure~\ref{fig:runtime_example} illustrates an example of the solver progress over runtime, showing how the optimal cost lower bound (derived from the CT) increases while the upper bound (the incumbent solution's cost) decreases. The incumbent solution at first comes from the repair function (\oursTPSIPP), however, later in the search the solutions from the child nodes during CT branchings are of higher quality. The solution found at $20.3$ seconds (with an optimality gap upper bound of $0.033\%$) was verified optimal by $22.5$ seconds, at which point the upper and lower bound on the objective value coincide.
\begin{figure}[h]
    \centering
    \includegraphics[width=1\linewidth]{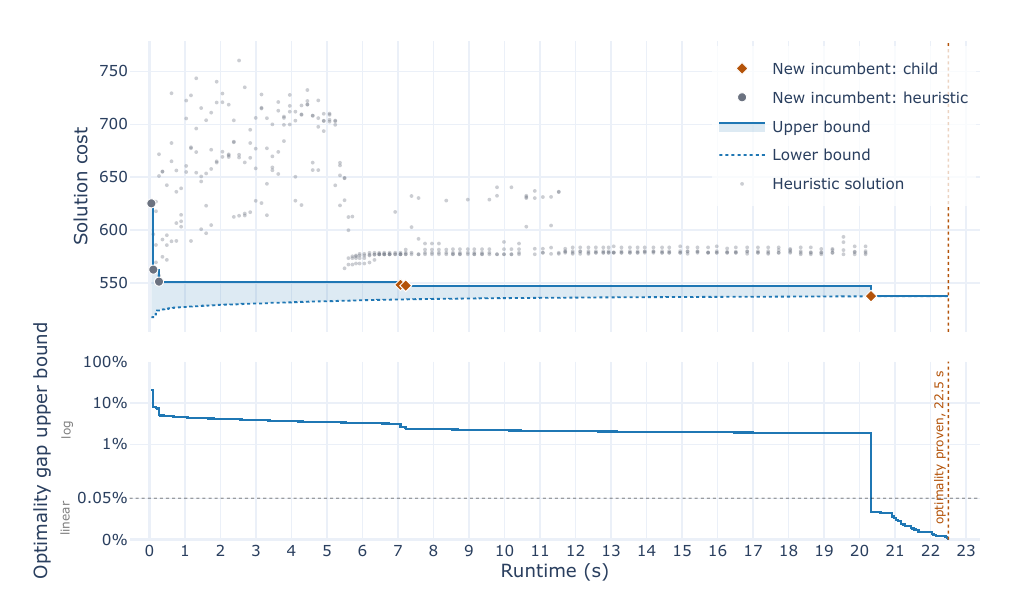}
    \caption{
    The solver progress over runtime. 
    The top plot shows the lower bound (from the CT tree) and the upper bound (the incumbent solution cost) of the optimal cost, forming the shaded region that corresponds to the incumbent solution's optimality gap upper bound. The bottom plot shows this optimality gap upper bound in percent (note the mixed linear and log scale). The upper plot also shows the solutions found by the repair function (\oursTPSIPP), as well as the incumbent solution's origin, being either from the repair function or from a child node during a branching.
    Finally, the solution found at $20.3$~seconds was verified optimal at $22.5$~seconds.}
    \label{fig:runtime_example}
\end{figure}

\subsection{Benchmark Setup: Maps, Scenarios, and Pre-processing}

The ablations and comparisons with \OCCBS are done using the standard MovingAI benchmark \emph{gridmaps}~\cite{SturtevantBenchmarks}, and \emph{roadmaps} sampled from these gridmaps.
Gridmaps are maps with unit-length edges, typically used in discrete-time MAPF, while roadmaps may include non-unit-length edges.
We use the \MovingAIMap{Den520d Sparse} roadmap from~\cite{CCBS}, 
and additionally our roadmaps sampled with our method detailed in Appendix~\ref{sec:BenchmarkGeneration}. 
A full roadmap benchmarking set using this method, along with scenario files, \AC{which will be made available on our repository}.
Figure~\ref{fig:Experiments:roadmaps} shows three such sampled roadmaps.
\begin{figure}[htbp]
    \begin{subfigure}[t]{0.31\textwidth}
        \centering
        \includegraphics[width=\linewidth]{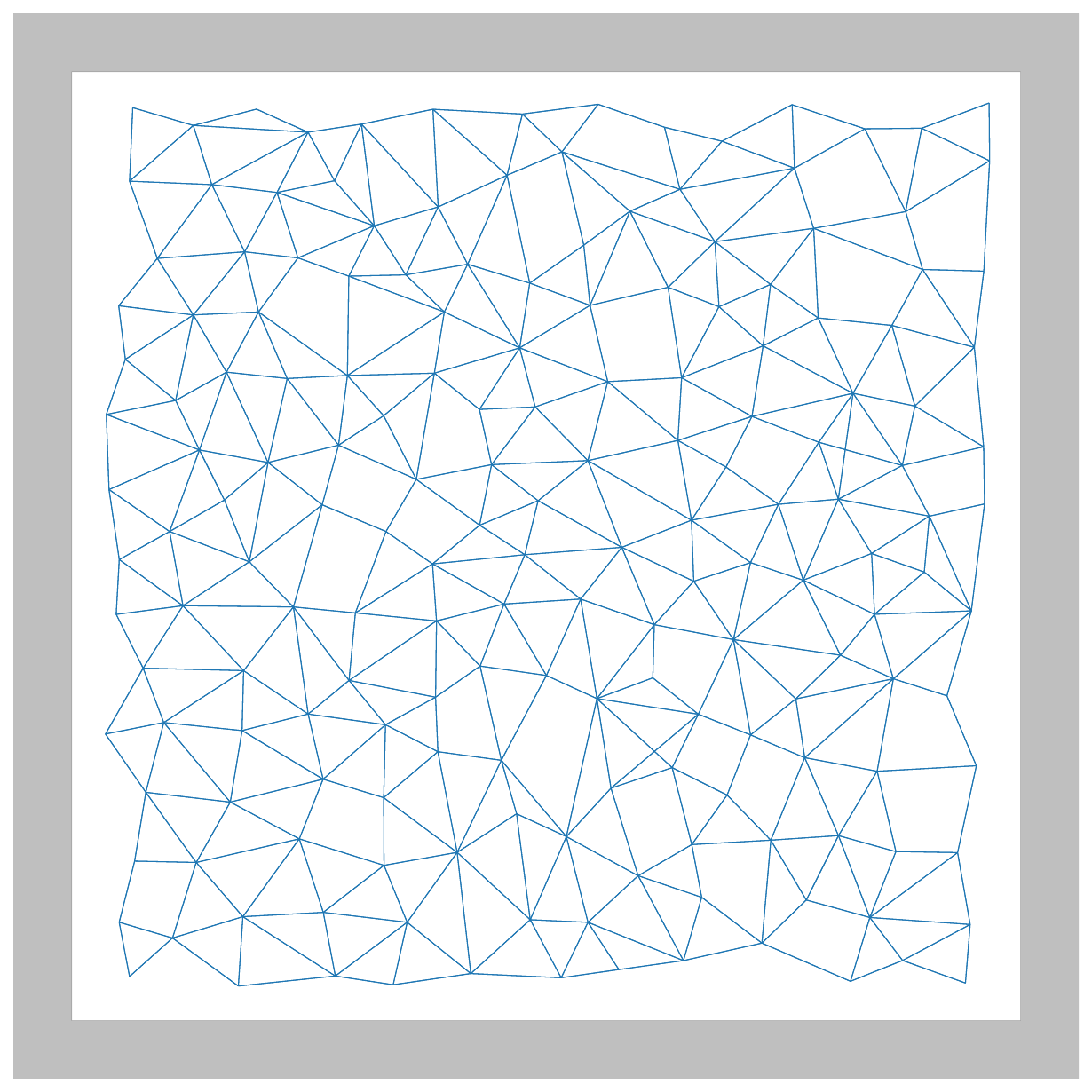}
        \caption{\MovingAIMap{Empty-16-16} with $|\Vertices|= 178$ and $|\Edges|= 916$.}
        \label{fig:Experiments:roadmaps:1}
    \end{subfigure}
    \hfill
    \begin{subfigure}[t]{0.31\textwidth}
        \centering
        \includegraphics[width=\linewidth]{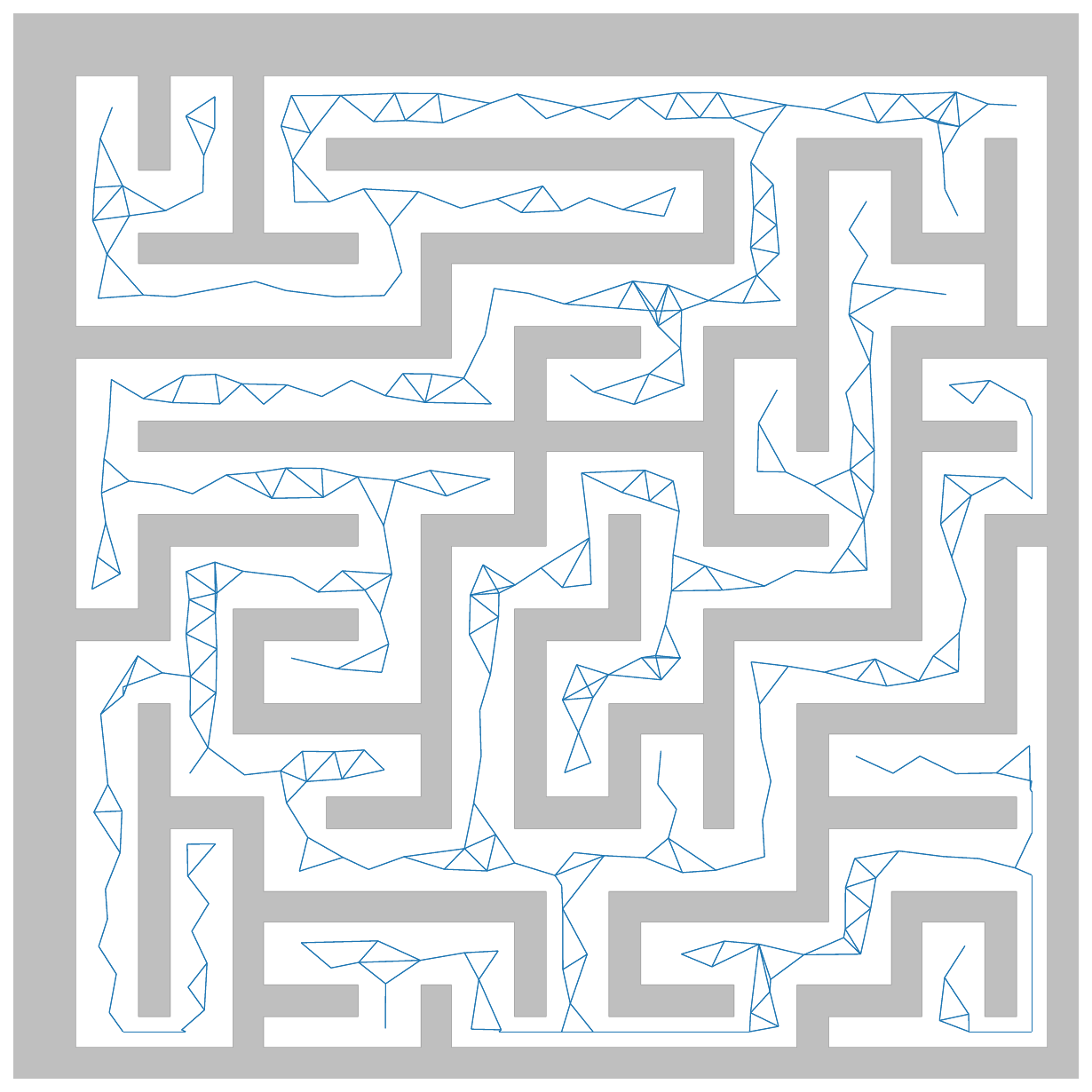}
        \caption{\MovingAIMap{Maze-32-32-2} with $|\Vertices|= 362$ and $|\Edges|= 1\;088$.}
        \label{fig:Experiments:roadmaps:2}
    \end{subfigure}
    \hfill
    \begin{subfigure}[t]{0.31\textwidth}
        \centering
        \includegraphics[width=\linewidth]{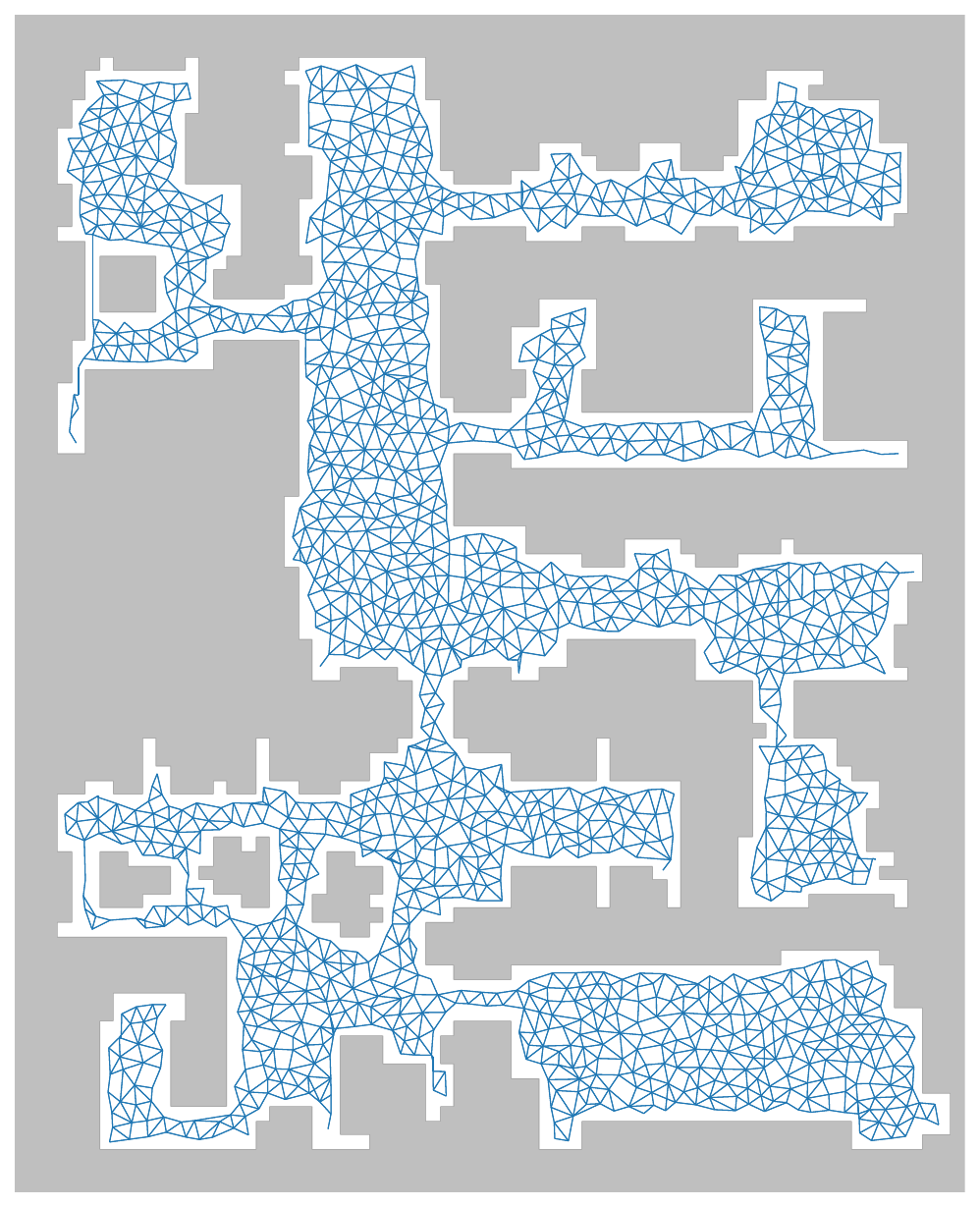}
        \caption{\MovingAIMap{Den312d} with $|\Vertices|= 1\;607$ and $|\Edges|= 7\;894$.}
        \label{fig:Experiments:roadmaps:3}
    \end{subfigure}
    \caption{Roadmaps sampled from MovingAI gridmaps, using our method detailed in Appendix~\ref{sec:BenchmarkGeneration} with a density $\density=1.0$.}
    \label{fig:Experiments:roadmaps}
\end{figure}

The benchmarking scheme from~\cite{SurveyStern2019} is used here:
for each map, a \emph{scenario} defines a list of start and goal vertex pairs.
A problem with $n$ agents is created from the first $n$ start and goal vertex pairs, 
starting with $n=2$ and incrementing until no solution (including sub-optimal) can be found for a problem within $30$ seconds. 
We used the
\emph{random} scenarios from~\cite{SturtevantBenchmarks} for the gridmaps, 
the scenarios for \MovingAIMap{Den520d Sparse} from~\cite{CCBS},
and our own generated scenarios (see Appendix~\ref{sec:BenchmarkGeneration}) for our sampled roadmaps.
To match prior work, all experiments on the gridmaps and \MovingAIMap{Den520d Sparse} use circular agent with radius $1/(2\sqrt{2})$.
The agents on our roadmaps have radius $0.5$. 

As described in Appendix~\ref{sec:Implementation}, pre-processing of graph and conflict information is done using knowledge that is assumed to be known beforehand. This includes information about the graphs and types of agents within the system.
The pre-processed graph information contains the all-pairs shortest path lengths, and the conflict information distils all domain specifics, such as geometry and conflict definitions, and stores that information for use during solving.
Information about a specific instance, such as the agent's starting positions and task sequences are not assumed to be known beforehand and therefore are not used for any preprocessing.
Thus, the reported runtime in all experiments begins from the time instance-specific information is given to the solver. 

Figure~\ref{fig:Experiments:preprocessing_times} visualise the graph and conflict information preprocessing time for each map, using parallel processing on 12 CPU cores. 
Although the pre-processing times remain well within a practical range,
these are done specifically for the circular agent, straight-line, constant-velocity model.
\AC{The time required to perform this pre-processing for, e.g., curved trajectories and non-circular agents, is the focus of ongoing work.}
Nonetheless, the measured pre-processing times appear to follow a predictable trend under this model, providing an estimate for other maps.
\begin{figure}
    \centering
    \includegraphics[width=1\linewidth]{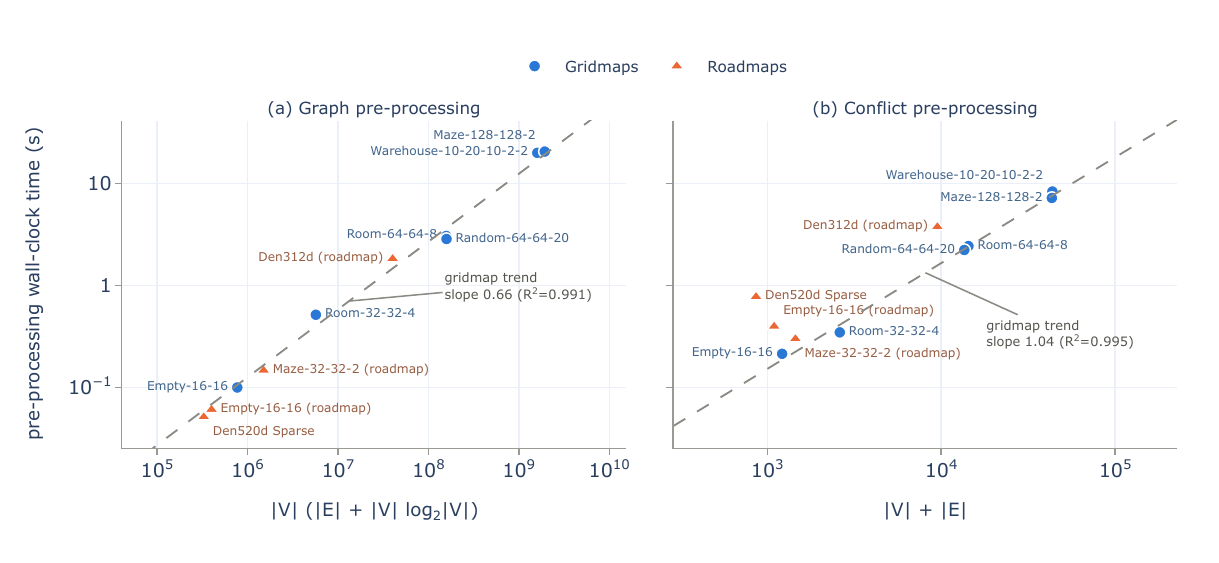}
    \caption{The graph and conflict information pre-processing runtime (see Appendix~\ref{sec:Implementation}), for each map used in the ablation and benchmark sections. The graph pre-processing computes the all-pairs shortest path lengths, which scales with $|\Vertices|(|\Edges| + |\Vertices|\log_2 |\Vertices|)$. 
    Conflict information pre-processing distils all domain-specific information, such as geometry and conflict definitions, and stores that information for use during solving.}
    \label{fig:Experiments:preprocessing_times}
\end{figure}

\subsection{Ablations}
\label{sec:Experiments:Ablations}

This section presents the results of an ablation study over \ours solver configurations together with comparisons with \OCCBS.
To match the problem formulation that \OCCBS is implemented for, 
all instances in this section use
maps with straight-line trajectories,
circular agents with radius $1/2\sqrt{2}$ (such that agents do not collide under the classical MAPF formulation on a grid map with unit length edges),
and only one vertex task per agent with an infinite service time.

These ablations are limited to maximum $200$ agents.
We use all $25$ scenarios for each of the grid maps \mapRoom (Room-64-64-8) and \mapRandom (Random-64-64-20) from~\cite{SturtevantBenchmarks} using connectedness $2$ (grid layout)
and the road map \mapDenSparse from~\cite{CCBS}.
We compare \ours implemented in Python, making extensive use of data-structure optimizations, with the original C++ \OCCBS implementation\footnote{\GithubOCCBS}.
Note that \OCCBS returns only optimal solutions;
in contrast, when \ours returns a suboptimal solution, we report on its optimality gap \emph{upper bound} \gapUB, the true optimality gap may be lower.

\subsubsection{Ablations on \ours Configurations}
\label{sec:Experiments:Ablations:Configurations}

\begin{table}[]
    \centering
    \caption{Overview of the experimental \ours configurations.}
    \label{tab:experiments:ablation:configurations}
    \begin{tabular}{lcccc}
        \toprule
        & \textbf{Nr of Policies} & \textbf{Repair Function} & \textbf{Node Selection} & \textbf{Portfolio Type} \\ 
        \midrule
        \textbf{C1} & 1  & ---     & Min. Cost & Static      \\
        \textbf{C2} & 1  & TP-SIPP & Min. Cost & Static      \\
        \textbf{C3} & 11 & ---     & Min. Cost & Static      \\
        \textbf{C4} & 11 & TP-SIPP & Min. Cost & Static      \\
        \textbf{C5} & 11 & TP-SIPP & Mixed     & Static      \\
        \textbf{C6} & 11 & w/wo TP-SIPP & Mixed/Min. Cost     & Gap Cadence \\
        \bottomrule
    \end{tabular}
\end{table}
Table~\ref{tab:experiments:ablation:configurations} details the six \ours configurations used in these ablations. They vary along three axes: single- versus multi-policy portfolios (each policy run on a separate CPU core), presence or absence of a repair function (\oursTPSIPP), and the node-selection rule. Selection is either minimum cost throughout, or a mix of three policies selecting by minimum cost, four by minimum conflict count, and four by minimum summed conflict time---where conflict time is the duration an agent's action must shift to avoid a move-move conflict, or the overlap duration for a move-wait conflict.
The multi-policy configurations use $11$ worker processes alongside the coordinating process, matching the $12$ performance cores of our hardware. \OCCBS is single-threaded, so only the single-policy configurations use comparable hardware.

Configurations \textbf{C1}--\textbf{C5} use a static portfolio, while \textbf{C6} uses a \emph{gap-cadence} portfolio that shifts between two sub-portfolios as the optimality gap upper bound \gapUB tightens: a \emph{solution-finding} portfolio (\textbf{C5}) that prioritizes finding a solution quickly while \gapUB is high, and an \emph{LB-raising} portfolio (\textbf{C3}) that prioritizes raising the CT solution lower bound as \gapUB falls. The solution-finding portfolio runs every $\lceil 0.05 / \gapUB \rceil$ iterations (every iteration when $\gapUB=\infty$), so it is used continuously while 
$\gapUB \geq 5\%$ and is progressively phased out in favour of the LB-raising portfolio as \gapUB decreases. At $\gapUB=1\%$ it runs only every fifth iteration.

\begin{figure}[]
    \centering
    \includegraphics[width=1\linewidth]{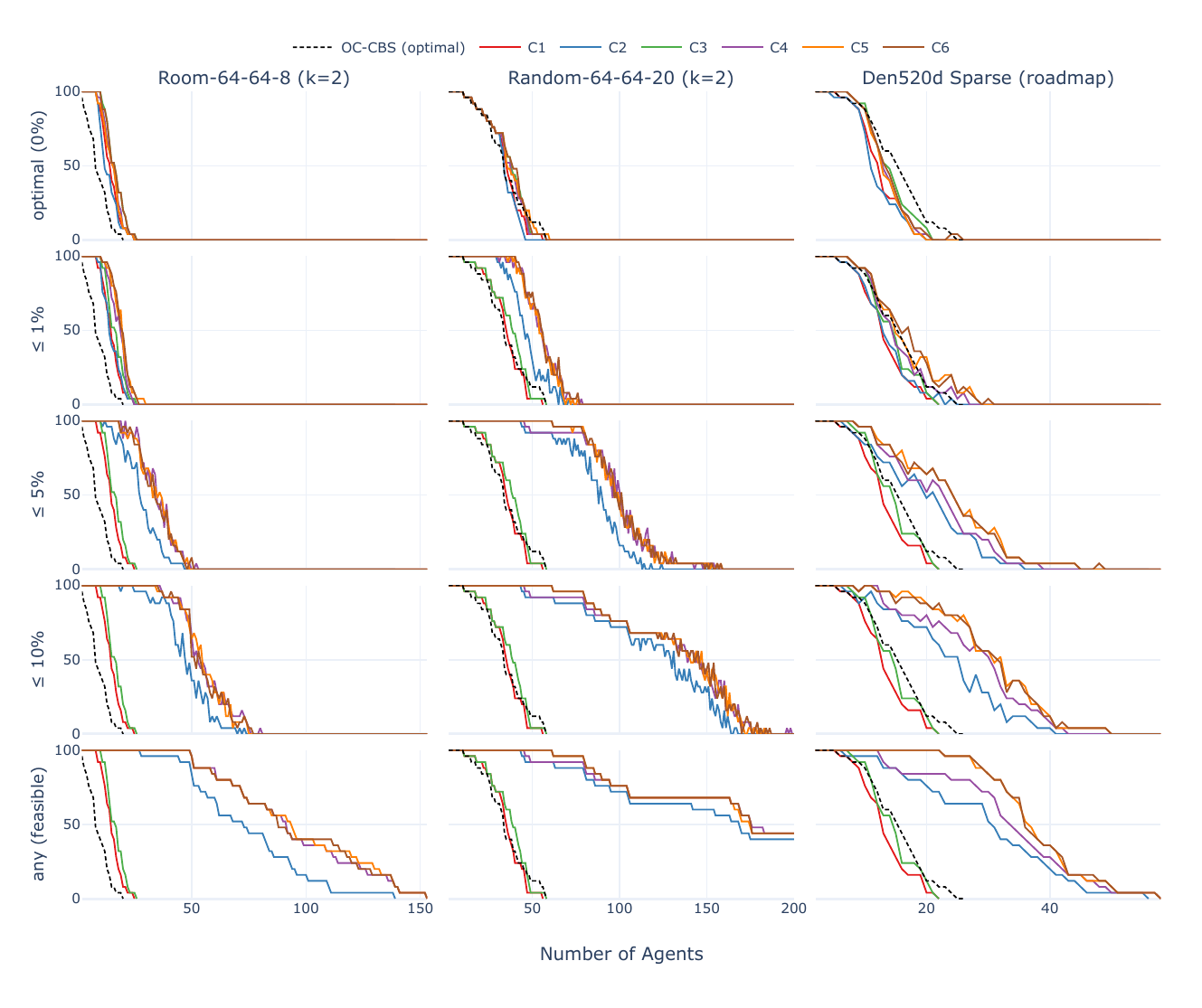}
    \caption{Success rate of \ours configurations \textbf{C1}--\textbf{C6} and \OCCBS, as the share of the $25$ scenarios per map solved at each agent count. Columns are maps; rows are optimality gap upper bounds \gapUB, from proven optimal (top) to any feasible solution (bottom). \OCCBS returns only optimal solutions, so its curve is identical in every row.}
    \label{fig:ablation:01_success_rate}
\end{figure}
Figure~\ref{fig:ablation:01_success_rate} visualises the success rate of solving to \gapUB values on each of the maps, using \ours with configurations \textbf{C1}--\textbf{C6} and \OCCBS.
Success rate is measured by the share of the $25$ scenarios that could be solved to the given \gapUB at each agent count.
Configuration \textbf{C1} is most alike \OCCBS (both single-core, no repair, Min. Cost node selection), however, \OCCBS is implemented in C++ which is often faster than
Python~\cite{PythonVSCpp}.
Despite this, \textbf{C1} outperforms \OCCBS on \mapRoom, is on par on \mapRandom, and lags behind on \mapDenSparse.
This indicates that the various implementation optimizations in \ours are advantageous; that these same optimizations in \OCCBS could yield meaningful improvements in runtime performance; and that a C++ implementation of \ours could further increase the scale of problems that can be solved.
Extending \textbf{C1} to multi-core with \textbf{C3} shows, in this figure, only slight improvements. We revisit this below.
Interestingly, by comparing the performance of \textbf{C1} and \textbf{C3} for optimal solutions with their performance for any solution, it can be seen most clearly on \mapDenSparse that despite not using a repair function, they are still able to find suboptimal solutions by only checking the child nodes during expansions.

The repair function is what separates the two groups. At loose \gapUB the deepest sweep on \mapRoom reaches $153$ agents for the repair-using configurations against $25$ for \textbf{C1}, on \mapDenSparse $58$ against $22$, and on \mapRandom the $200$-agent limit against $56$;
\textbf{C2} remains slightly behind \textbf{C4}, the difference being 1 vs 11 policies.
That gain is not free, however. At $\gapUB=0$ the repair-using configurations are consistently \emph{worse} than their non-repairing counterparts---\textbf{C2} below \textbf{C1} and \textbf{C4} below \textbf{C3} on all three maps---because time spent in \oursTPSIPP is time not spent expanding CT nodes to raise the lower bound.
Configuration \textbf{C6} is designed to avoid exactly this trade-off, and the results show that it does: it matches \textbf{C5} at loose \gapUB while recovering \textbf{C3}'s performance at $\gapUB=0$, which no static configuration achieves on both ends.

\begin{figure}[]
    \centering
    \includegraphics[width=1\linewidth]{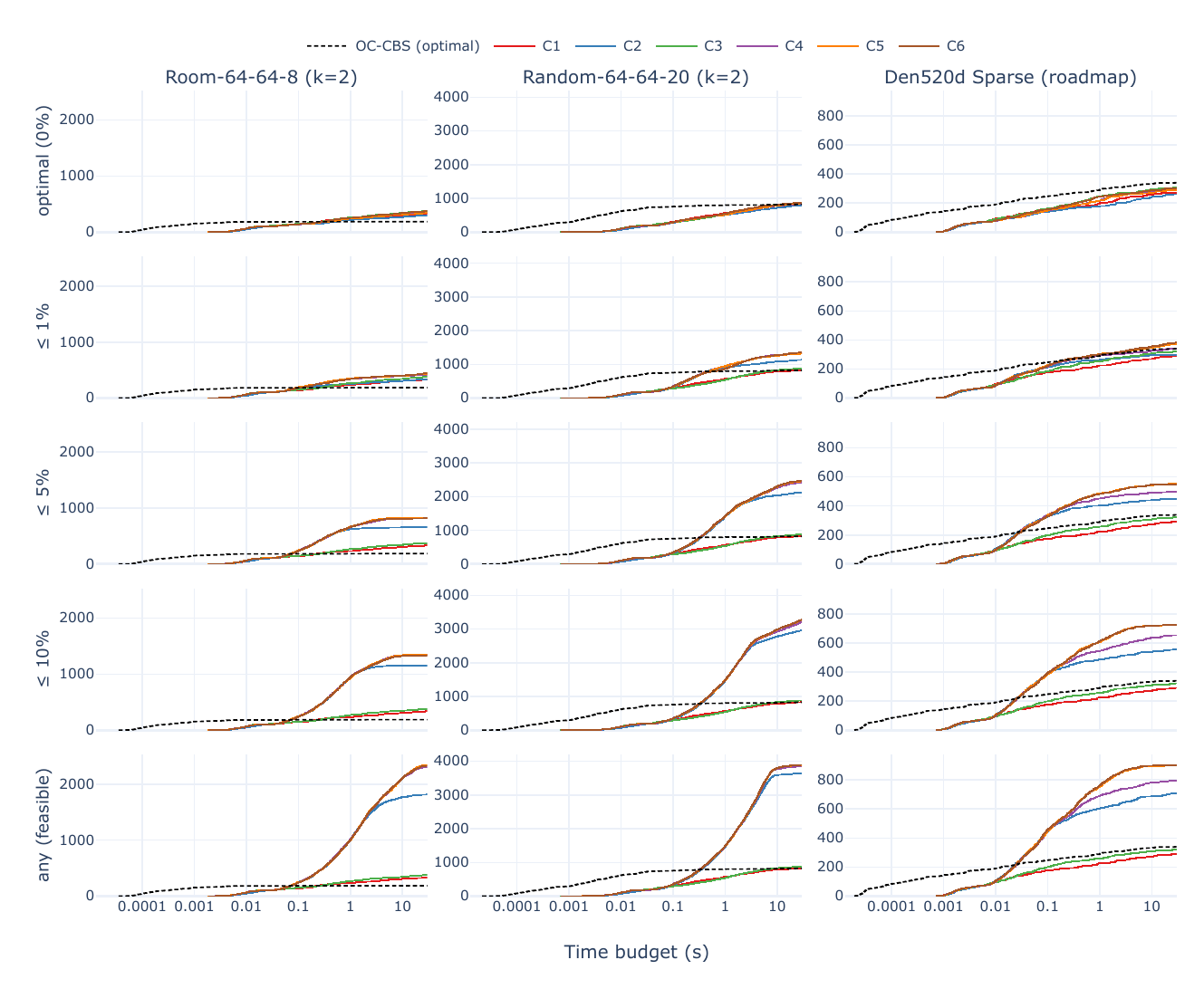}
    \caption{Number of problems solved to a given optimality gap upper bound \gapUB within a given time budget, for \ours configurations \textbf{C1}--\textbf{C6} and \OCCBS. Columns are maps; rows are \gapUB, from proven optimal (top) to any feasible solution (bottom). The vertical axis is an absolute instance count, shared down each column and independent across columns; the horizontal axis is logarithmic.}
    \label{fig:ablation:02_solved_vs_runtime}
\end{figure}
Figure~\ref{fig:ablation:02_solved_vs_runtime} visualises instead the number of problems that each configuration can solve within a given runtime, for a given optimality gap and map.
These results highlight how fast the solvers are at solving problems.
The per-map ordering between \textbf{C1}, \textbf{C3} and \OCCBS is the same as in
Figure~\ref{fig:ablation:01_success_rate}, but the runtime axis shows how it arises: \OCCBS solves the easy instances one to two orders of magnitude faster, consistent with a C++ constant-factor advantage, while \textbf{C1} and \textbf{C3} close most of that gap over the remaining budget.
Among the repair-using configurations, \textbf{C5} and \textbf{C6} are equivalent, \textbf{C4} is equivalent on the grid maps (\mapRoom and \mapRandom) but lags slightly on the road map (\mapDenSparse), and \textbf{C2} is just behind.
Taken over the full $30$ second budget, \textbf{C6} finds a feasible solution to 
$6.9\times$ as many problems as \textbf{C1} on \mapRoom ($2329$ versus $336$), 
$4.7\times$ as many on \mapRandom ($3896$ versus $827$), and
$3.1\times$ as many on \mapDenSparse ($904$ versus $291$).
Specifically on \mapRandom, $4.7\times$ is a lower bound since $11$ of the $25$ scenarios reach the $200$-agent limit rather than a failure.
The multi-core effect is far smaller at this aggregate level: \textbf{C3} ends $6$--$15\%$ ahead of \textbf{C1} in solved instances on each map. The next section shows that this modest aggregate figure hides a large effect concentrated in the hardest instances.

A natural question that arises from looking at the results in Figures~\ref{fig:ablation:01_success_rate} and~\ref{fig:ablation:02_solved_vs_runtime} is: \emph{what is the advantage of using multiple policies (CPU processors), and where can that advantage be found?}
The performance difference between \textbf{C2} and \textbf{C4} could possibly be explained by \textbf{C4} running \oursTPSIPP more often to sample more random repair orders, thereby finding better solutions through chance.
\textbf{C1} and \textbf{C3}, on the other hand, do not use \oursTPSIPP. Therefore, the additional policies in \textbf{C3} contribute primarily toward the search for, and verification of, an optimal solution, yet Figures~\ref{fig:ablation:01_success_rate} and~\ref{fig:ablation:02_solved_vs_runtime} show only a modest aggregate difference between the two. 
The following analysis shows that this aggregate view conceals a large effect confined to the hardest instances.

\begin{figure}
    \centering
    \includegraphics[width=1\linewidth]{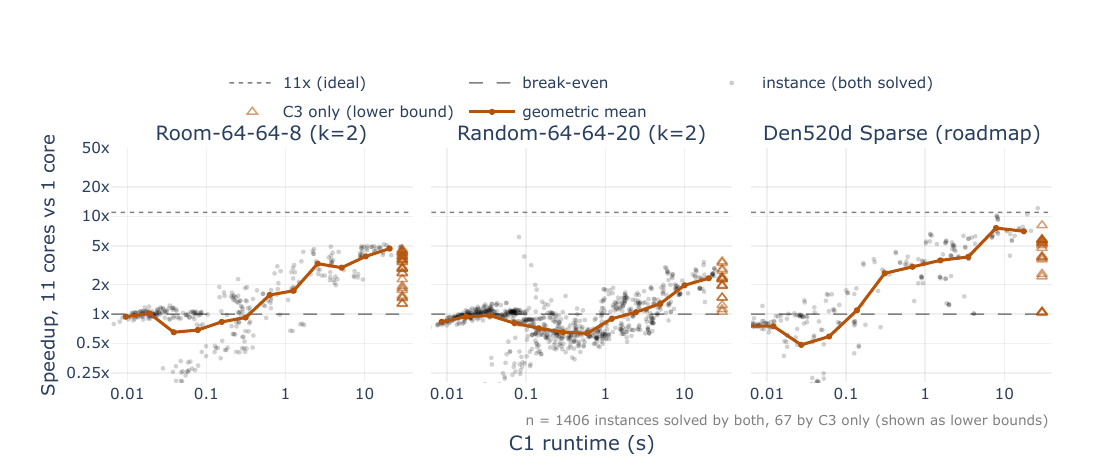}
    \caption{Per-instance speedup of \textbf{C3} ($11$ policies) over \textbf{C1} ($1$ policy), $t_\textbf{C1}/t_\textbf{C3}$, plotted against \textbf{C1}'s solve time as a proxy for instance difficulty. Both axes are logarithmic, and the vertical axis is symmetric about break-even so that equal speedups and slowdowns are equally far from it. Each dot is one instance solved optimally by both configurations; open triangles are the $67$ instances solved optimally by \textbf{C3} alone, drawn at the lower bound $30\,\mathrm{s}/t_\textbf{C3}$ with the true value somewhere above. The amber curve is the geometric mean over $12$ log-spaced difficulty bins, excluding bins with fewer than $8$ instances and excluding the censored instances---so it understates the advantage at its right end.}
    \label{fig:ablation:03_multicore_paired_speedup}
\end{figure}
The results of our investigation into the above question are shown in Figure~\ref{fig:ablation:03_multicore_paired_speedup} and Table~\ref{tab:ablation:03_multicore_paired_speedup}.
Figure~\ref{fig:ablation:03_multicore_paired_speedup} visualises, for each map, the speedup of using \textbf{C3} over \textbf{C1} for each instance, plotted over the time it took \textbf{C1} to solve.
That is, we use \textbf{C1}'s solve time $t_\textbf{C1}$ as a proxy for a problem's difficulty, and then show how much faster \textbf{C3} could solve the same problem, $t_\textbf{C1}/t_\textbf{C3}$. 
The $67$ instances that \textbf{C3} was able to solve to optimality but not \textbf{C1} (there were no cases of the opposite) are shown as lower bounds on the speedup, using the time limit $t_\textbf{C1}=30$ seconds.
The geometric mean curve is also shown, computed by dividing the data into 12 log-spaced bins and computing the geometric mean over the per-instance speedup within the bin, with bins containing less than $8$ instances dropped.
All three maps share the same shape---at or below break-even on the easiest instances, a dip below it, then a monotone rise---but differ in where the crossover falls and how large the final gain is. 
On \mapDenSparse the trend crosses break-even at roughly $0.15$ seconds and reaches $7.0\times$; on \mapRoom at roughly $0.5$ seconds, reaching $4.7\times$; 
and on \mapRandom only at roughly $2$ seconds, reaching $2.3\times$. 
The multi-policy portfolio thus pays off earliest and most on the road map and latest and least on the random grid.
Note also that the trend is conservative at its right end: the censored instances, where \textbf{C3} did best, are by construction excluded from the bins, as are the $123$ instances that only \textbf{C3}'s sweep reached and which therefore have no \textbf{C1} counterpart to
pair with.

\begin{table}[]
    \centering
    \begin{tabular}{lrrrrrrr}
        \toprule
        bucket & $n$ & $t_\textbf{C1}$ & $t_\textbf{C3}$ & speedup & throughput & overhead \\
        \midrule
            $0$--$0.1$ & 572 & 0.02 & 0.03 & 0.78$\times$ & 0.93$\times$ & 1.19$\times$ \\
            $0.1$--$0.5$ & 311 & 0.22 & 0.28 & 0.78$\times$ & 2.02$\times$ & 2.59$\times$ \\
            $0.5$--$2$ & 231 & 1.06 & 0.91 & 1.16$\times$ & 3.68$\times$ & 3.16$\times$ \\
            $2$--$5$ & 145 & 3.18 & 1.87 & 1.70$\times$ & 4.01$\times$ & 2.36$\times$ \\
            $5$--$15$ & 101 & 8.30 & 2.88 & 2.88$\times$ & 4.43$\times$ & 1.54$\times$ \\
            $15$--$30$ & 46 & 19.65 & 5.76 & 3.41$\times$ & 4.77$\times$ & 1.40$\times$ \\
            $>30$ (timeout) & 67 & $>30$ & 10.88 & $>2.76\times$ & -- & -- \\
        \bottomrule
    \end{tabular}
    \caption{Paired comparison of \textbf{C1} and \textbf{C3} over all three maps. Instances are grouped into buckets by \textbf{C1}'s solve time, with $n$ the number of instances in each bucket. $t_\textbf{C1}$ and $t_\textbf{C3}$ are geometric mean solve times, so $t_\textbf{C1}/t_\textbf{C3}$ equals the speedup column. The remaining columns are geometric means over the instance-wise ratios between \textbf{C3} and \textbf{C1}: solve time $t_\textbf{C1}/t_\textbf{C3}$ (speedup), CT nodes expanded per unit time (throughput), and total CT nodes expanded (overhead), so that speedup $=$ throughput $/$ overhead. The final row holds instances \textbf{C1} did not solve optimally within the limit; their speedups are lower bounds and no node ratios are available. There were no instances solved optimally by \textbf{C1} alone.}
    \label{tab:ablation:03_multicore_paired_speedup}
\end{table}
Table~\ref{tab:ablation:03_multicore_paired_speedup} aggregates the data from all maps, showing for each bucket the geometric mean solve time for \textbf{C1} and \textbf{C3}, the speedup, the increase in node expansion rate (throughput), and the increase in the total number of expanded nodes (overhead). Since all three are geometric means over the same instances, the speedup is exactly the throughput divided by the overhead, which separates the two competing effects of running $11$ policies in parallel.

The throughput shows the raw benefit of using multiple processors: more nodes are expanded per unit time, peaking at $4.77\times$ that of \textbf{C1} in the $15$--$30$ second bucket. 
This is well short of the $11\times$ ideal, giving at most $43\%$ parallel efficiency on our hardware, which we attribute to the cost of managing a pool of policies on separate processes with synchronization at every iteration. On the easiest instances that cost dominates outright: at a geometric mean solve time of $0.02$ seconds the pool expands nodes \emph{slower} than a single process ($0.93\times$), as the whole solve is shorter than the time needed to make the pool
productive.

The overhead shows that the parallel search expands more CT nodes than the single-core search to solve the same problem, so part of the throughput is spent on nodes that are not needed to find and verify an optimal solution. This is not surprising: expanding one node at a time means every choice is made with the latest information, whereas expanding several nodes in parallel means that all but the first are chosen without it. The overhead is largest on mid-difficulty problems ($3.16\times$ in the $0.5$--$2$ second bucket) and falls to $1.40\times$ on the hardest, where the CT is large enough that the parallel policies explore genuinely distinct regions rather than duplicating a short search.

The speedup is what remains after these two effects cancel. On the two easiest buckets the overhead exceeds the throughput and \textbf{C1} is the faster solver ($0.78\times$); from the $0.5$--$2$ second bucket onward the ordering reverses and the advantage grows monotonically, reaching $3.41\times$ in the $15$--$30$ second bucket. Both factors move favourably with difficulty---throughput rises while overhead falls---which is why the effect compounds rather than saturating.

The most consequential row is the last. On $67$ instances \textbf{C3} found and verified an optimal solution within the time limit where \textbf{C1} did not, with no instance where the reverse held, raising coverage of the paired instances from $92.0\%$ to $96.3\%$. The principal benefit of the additional policies is therefore not that the same problems are solved faster, but that problems out of reach of a single core come into reach.

\subsubsection{Ablation on Conflict Selection}
\label{sec:Experiments:Ablations:ConflictSelection}

Every configuration in Section~\ref{sec:Experiments:Ablations:Configurations} uses the \ConflictSelectionBestof{\infty} conflict selection strategy.
This section isolates that choice.
All five configurations below are \textbf{C5} with only the conflict-selection rule varied, so any difference between them is attributable to that rule alone.
Recall from Section~\ref{sec:AOCCBS:highlevel:conflict_selection}:
\ConflictSelectionArbitrary branches on the first conflict found and performs no probing;
\ConflictSelectionSemiCard scans the conflict list and returns the first conflict for which at least one child's cost increases;
\ConflictSelectionBestof{n} probes up to $n$ candidate conflicts and returns the one maximising the smaller of the two children's cost increases.
We test \textbf{C5} using \ConflictSelectionArbitrary, \ConflictSelectionSemiCard, and \ConflictSelectionBestof{n} with $n \in \{8, 16, \infty\}$.
The ablation is run on \mapRoom ($2535$ instances over $25$ scenarios) and \mapDenSparse ($968$ instances over $25$ scenarios).

\begin{figure}
    \centering
    \includegraphics[width=1\linewidth]{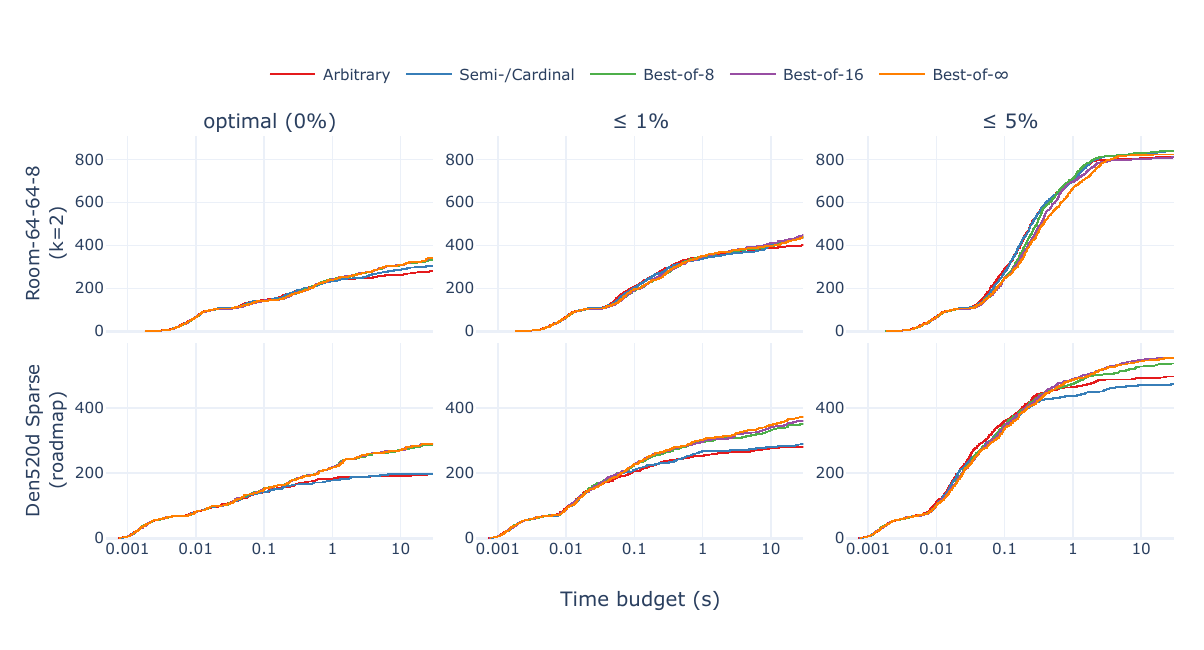}
    \caption{Conflict-selection ablation: instances proved to within each optimality gap as a function of the time budget, on \mapRoom (top row, $2535$ instances) and \mapDenSparse (bottom row, $968$ instances). All five configurations are \textbf{C5} with only the conflict-selection strategy varied. At the optimal threshold the curves fall into two groups on both maps, with the three \ConflictSelectionBestof{n} rules mutually indistinguishable and \ConflictSelectionArbitrary below them; the separation narrows as the gap threshold loosens, and on \mapRoom it has closed entirely by $\le 5\%$.}
    \label{fig:04_conflict_selection}
\end{figure}
Figure~\ref{fig:04_conflict_selection} visualises the number of problems solved using each of the conflict selection strategies, for various \gapUB-values.
At the optimal threshold a clear separation emerges on both maps between \ConflictSelectionArbitrary and the \ConflictSelectionBestof{n} strategies.
Within the $30$-second time limit on \mapDenSparse,
\ConflictSelectionArbitrary and \ConflictSelectionSemiCard prove optimality on $196$ and $197$ instances respectively, while \ConflictSelectionBestof{8}, \ConflictSelectionBestof{16} and \ConflictSelectionBestof{\infty} reach $286$, $291$ and $291$---a $48\%$ increase over branching at random.
On \mapRoom the same ordering holds but with a smaller margin: $280$ and $305$ against $334$, $342$ and $342$, a $22\%$ increase.
The advantage shrinks as the threshold loosens.
On \mapDenSparse it survives at $\le 1\%$ ($280$ and $289$ against $350$, $359$ and $372$) and at $\le 5\%$ ($497$ and $473$ against $537$, $555$ and $554$).
On \mapRoom it is already marginal at $\le 1\%$ ($401$ and $434$ against $438$, $450$ and $441$) and absent at $\le 5\%$, where all five configurations lie within $4\%$ of one another ($815$ and $839$ against $842$, $811$ and $824$) and \ConflictSelectionBestof{16} is nominally the weakest.
These results suggest that reasoning over conflict cardinality is most useful when aiming for optimal solutions, however, once a few percent of suboptimality is acceptable, the choice of conflict to branch on matters less, and on the map where good feasible solutions are easy to come by it stops mattering sooner.

The behaviour of \ConflictSelectionSemiCard is map-dependent, and this is the one place the two maps disagree qualitatively.
On \mapDenSparse it is not merely a weak heuristic but a net loss: it costs $0.84\times$ against \ConflictSelectionArbitrary in paired geometric-mean runtime, expands $18\%$ \emph{more} constraint-tree nodes, and at $\le 5\%$ falls behind \ConflictSelectionArbitrary outright ($473$ against $497$).
On \mapRoom it is a modest gain instead, at $1.05\times$ paired geometric-mean runtime with $10\%$ \emph{fewer} nodes, and ahead of \ConflictSelectionArbitrary at every threshold.
What is consistent across both maps is that it captures only a small part of the available benefit: it recovers $25$ of the $62$-instance gap to \ConflictSelectionBestof{16} on \mapRoom and $1$ of $95$ on \mapDenSparse.
We attribute this to the rule satisficing rather than maximising.
Accepting the first conflict that raises either child's cost is barely more informative than just accepting the first conflict, yet it pays for a probe at every candidate it scans along the way, and unlike \ConflictSelectionBestof{n} that scan is not capped---so whether it comes out slightly ahead or slightly behind depends on how expensive the uncapped scan happens to be on a given map.
The benefit of cardinality reasoning in this setting comes from \emph{ranking} candidate conflicts, not from \emph{classifying} them.

These results suggest that the \ConflictSelectionBestof{n} candidate sample saturates almost immediately, and this holds on both maps.
At the optimal threshold \ConflictSelectionBestof{8}, \ConflictSelectionBestof{16} and \ConflictSelectionBestof{\infty} are separated by $5$ and $0$ instances on \mapDenSparse and by $8$ and $0$ on \mapRoom.
Pairwise, their geometric-mean speedups are within $3\%$ of each other on both maps ($1.02\times$ and $1.00\times$ on \mapDenSparse, $1.00\times$ and $1.03\times$ on \mapRoom), with node counts within $2\%$.
Since every probe of a candidate conflict costs two low-level searches (to determine the raise in cost),
and the list of candidate conflicts grows with the number of agents,
this matters in practice:
apparently even a small candidate conflict set ($n\in\{8, 16\}$) suffices to find a good candidate,
while ensuring that the amount of work is bounded, unlike with $n=\infty$.

\begin{figure}
    \centering
    \includegraphics[width=1\linewidth]{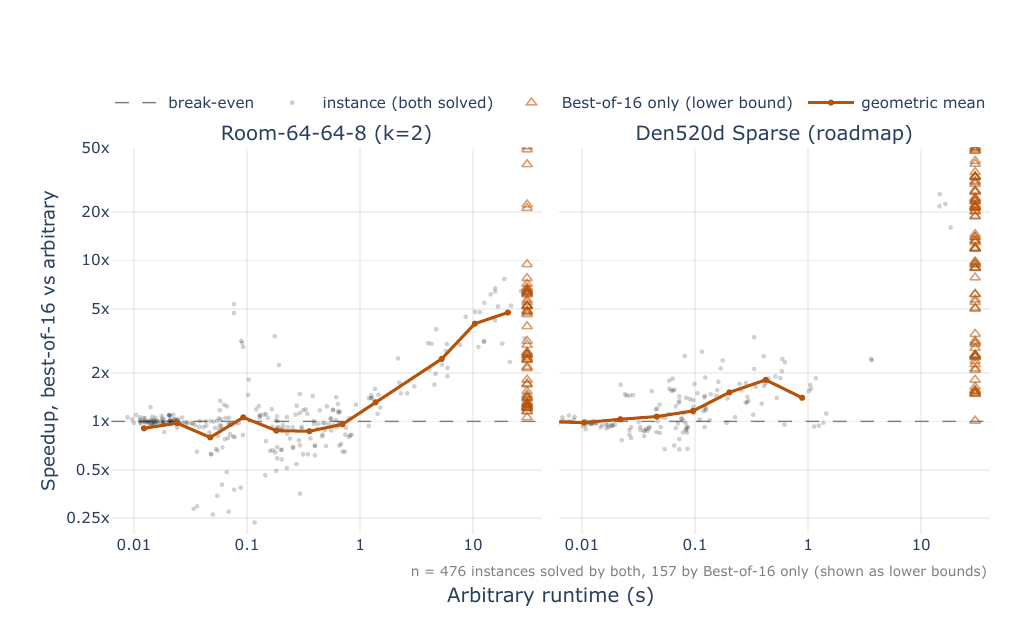}
    \caption{Per-instance speedup of \ConflictSelectionBestof{16} over \ConflictSelectionArbitrary against instance hardness, on \mapRoom and \mapDenSparse. Open triangles at the right edge are the instances only \ConflictSelectionBestof{16} proved optimal---$62$ on \mapRoom and $95$ on \mapDenSparse---drawn at the lower bound $30\,\mathrm{s} / t_{\text{\ConflictSelectionBestof{16}}}$; there are none in the opposite direction on either map. The trend line covers only jointly-solved instances.}
    \label{fig:05_conflict_selection_paired_speedup}
\end{figure}
\begin{table}[]
    \centering
    \begin{tabular}{lrrrrrr}
        \toprule
        bucket & $n$ & $t_{\text{\ConflictSelectionArbitrary}}$ & $t_{\text{\ConflictSelectionBestof{16}}}$ & speedup & throughput & overhead \\
        \midrule
        \multicolumn{7}{l}{\emph{\mapRoom}} \\
            $0$--$0.1$ & 134 & 0.03 & 0.03 & 0.92$\times$ & 0.97$\times$ & 1.05$\times$ \\
            $0.1$--$1$ & 102 & 0.31 & 0.35 & 0.90$\times$ & 0.84$\times$ & 0.93$\times$ \\
            $1$--$30$ & 44 & 5.60 & 2.10 & 2.67$\times$ & 0.75$\times$ & 0.28$\times$ \\
            $>30$ (timeout) & 62 & $>30$ & 8.74 & $>3.43\times$ & -- & -- \\
        \midrule
        \multicolumn{7}{l}{\emph{\mapDenSparse}} \\
            $0$--$0.1$ & 137 & 0.02 & 0.02 & 1.03$\times$ & 0.79$\times$ & 0.77$\times$ \\
            $0.1$--$1$ & 44 & 0.24 & 0.15 & 1.55$\times$ & 0.63$\times$ & 0.41$\times$ \\
            $1$--$30$ & 15 & 3.35 & 0.85 & 3.94$\times$ & 0.89$\times$ & 0.23$\times$ \\
            $>30$ (timeout) & 95 & $>30$ & 2.30 & $>13.06\times$ & -- & -- \\
        \bottomrule
    \end{tabular}
    \caption{Paired comparison of \ConflictSelectionBestof{16} against \ConflictSelectionArbitrary by instance hardness. Speedup decomposes as throughput divided by search overhead, where throughput is the ratio of node expansion rates and search overhead the ratio of nodes expanded. All entries are per-instance ratios averaged geometrically within a bucket. The final row of each block holds instances \ConflictSelectionArbitrary did not prove optimal within the budget, whose speedups are right-censored and reported as lower bounds.}
    \label{tab:ablation:05_conflict_selection_paired_speedup}
\end{table}
Figure~\ref{fig:05_conflict_selection_paired_speedup} and Table~\ref{tab:ablation:05_conflict_selection_paired_speedup} compare \ConflictSelectionArbitrary and \ConflictSelectionBestof{16} on speedup, throughput, and overhead over $t_\text{\ConflictSelectionArbitrary}$ as a proxy for instance difficulty.
Probing costs node throughput on both maps:
\ConflictSelectionBestof{16} expands nodes at $0.63$ to $0.89$ times \ConflictSelectionArbitrary's rate on \mapDenSparse and $0.75$ to $0.97$ times on \mapRoom.
What it buys is a smaller constraint tree, and it buys progressively more of it as instances get harder---down to $0.77$, $0.41$ and $0.23$ of \ConflictSelectionArbitrary's node count on \mapDenSparse, and $1.05$, $0.93$ and $0.28$ on \mapRoom.
The net effect is therefore a crossover rather than a uniform gain, and \mapRoom shows it more starkly than \mapDenSparse: there, probing does not repay its cost at all below one second ($0.92\times$ and $0.90\times$), and only on instances taking \ConflictSelectionArbitrary more than a second does it pay off, at $2.67\times$.
On \mapDenSparse the corresponding progression is $1.03\times$, $1.55\times$ and $3.94\times$.
Consequently the aggregate geometric mean over jointly-solved instances---$1.25\times$ on \mapDenSparse and $1.08\times$ on \mapRoom---understates the effect substantially, being dominated by the trivial instances that make up $70\%$ and $48\%$ of those sets respectively.
The strongest evidence is the censored row.
Of the instances \ConflictSelectionArbitrary failed to prove optimal within the budget, \ConflictSelectionBestof{16} proved $95$ on \mapDenSparse, taking a geometric-mean $2.30$ seconds and therefore at least $13\times$ less time, with extreme cases exceeding $100\times$; on \mapRoom it proved a further $62$ at a geometric-mean $8.74$ seconds, at least $3.4\times$ less time, with extreme cases exceeding $50\times$.
On neither map is there a single instance in the opposite direction.

\subsection{Comparisons with \OCCBS on Additional Maps}
\label{sec:Experiments:OCCBS_comparison}

We apply the benchmarking procedure of Section~\ref{sec:Experiments:Ablations} to size MovingAI gridmaps and the three sampled roadmaps of Figure~\ref{fig:Experiments:roadmaps}, 
using configuration \textbf{C6} with \ConflictSelectionBestof{16} conflict selection.

Figure~\ref{fig:Experiments:MainBenchmark:Gridmaps} compares \ours and \OCCBS on the gridmaps.
At $\gapUB=0$ the two are roughly equivalent, \ours proving optimality on $3\;414$ instances against \OCCBS's $3\;269$.
\ours is ahead on five of the six maps, by as much as 
$2.0\times$ on \MovingAIMap{Room-64-64-8} ($376$ against $187$) and 
$1.7\times$ on \MovingAIMap{Room-32-32-4} ($440$ against $262$),
and behind on \MovingAIMap{Warehouse-10-20-10-2-2} ($949$ against $1\;296$).
Given that \OCCBS is implemented in C++ and \ours in Python, 
we read this as evidnce that the implementation optimizations described in Appendix~\ref{sec:Implementation} offset the language gap, rather than as a claim of superiority at proving optimality.

The separation appears once suboptimality is permitted.
Within the same $30$ second time limit, \ours finds a feasible solution to $5.8\times$ as many instances as \OCCBS solves optimally ($19\;078$ against $3\;269$), and the agent count at which half the scnearios are still solved rises accordingly:
from $22$ to $76$ on \MovingAIMap{Empty-16-16},
$7$ to $91$ on \MovingAIMap{Room-64-64-8}, 
$6$ to $70$ on \MovingAIMap{Maze-128-128-2}, 
and $33$ to $181$ on \MovingAIMap{Random-64-64-20}. 
On \MovingAIMap{Warehouse-10-20-10-2-2}, \ours solves every scenario at $300$ agents (the limit of these preliminary experiments, reached by $24$ of the
$25$ scenarios) while \OCCBS is down to a single scenario beyond $124$ agents and solves none at $135$. 
That map's $5.6\times$ is therefore a lower bound.
\begin{figure}
    \centering
    \includegraphics[width=1\linewidth]{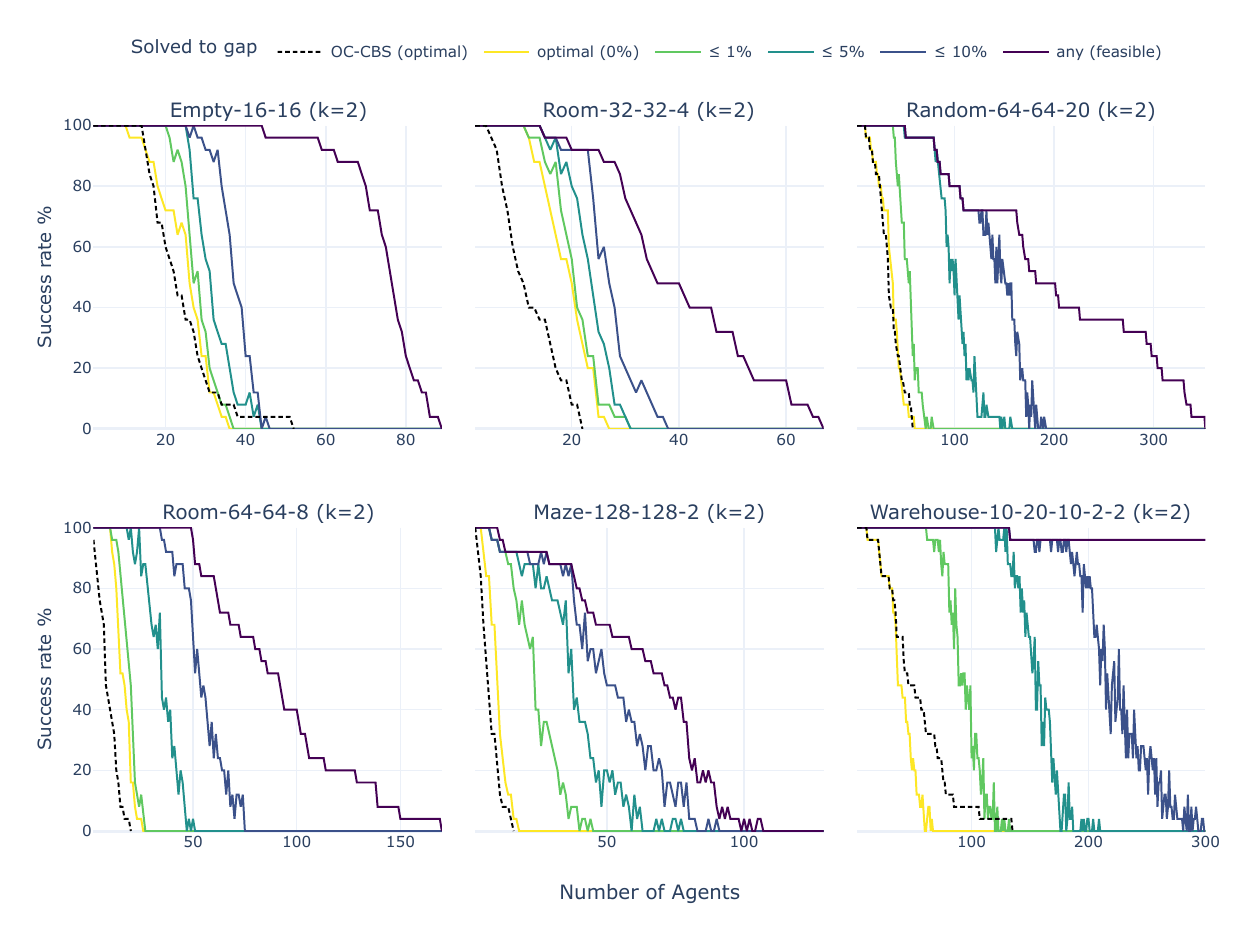}
    \caption{A comparison between \ours and \OCCBS on success rate at various optimality gap upper-bound values, on six different gridmaps. The map \MovingAIMap{Warehouse-10-20-10-2-2} was tested only up to $300$ agents (experiments are
    ongoing).}
    \label{fig:Experiments:MainBenchmark:Gridmaps}
\end{figure}

Figure~\ref{fig:Experiments:MainBenchmark:Roadmaps} repeats the comparison on the roadmaps, where the same pattern holds more starkly. 
At $\gapUB=0$ the two methods are equivalent: 
\ours proves optimality on $616$ instances against $592$, ahead on all three maps but by margins that do not extend the reachable range:
the success-rate curves reach zero at $7$ agents on \MovingAIMap{Maze-32-32-2} and 
$19$ on \MovingAIMap{Den312d} for both methods, 
and at $18$ against \OCCBS's $20$ on \MovingAIMap{Empty-16-16}. 
This is consistent with Section~\ref{sec:Experiments:Ablations:Configurations}, where \textbf{C1} outperformed \OCCBS on both gridmaps but lagged on the roadmap \mapDenSparse.

Permitting suboptimality again separates them, 
by $5.4\times$ overall ($3\,220$ instances against $592$). 
On \MovingAIMap{Den312d} the agent count at which half the scenarios are solved rises from $10$ to $85$, 
and the deepest scenario reaches $132$ agents against \OCCBS's $19$. 
\begin{figure}
    \centering
    \includegraphics[width=1\linewidth]{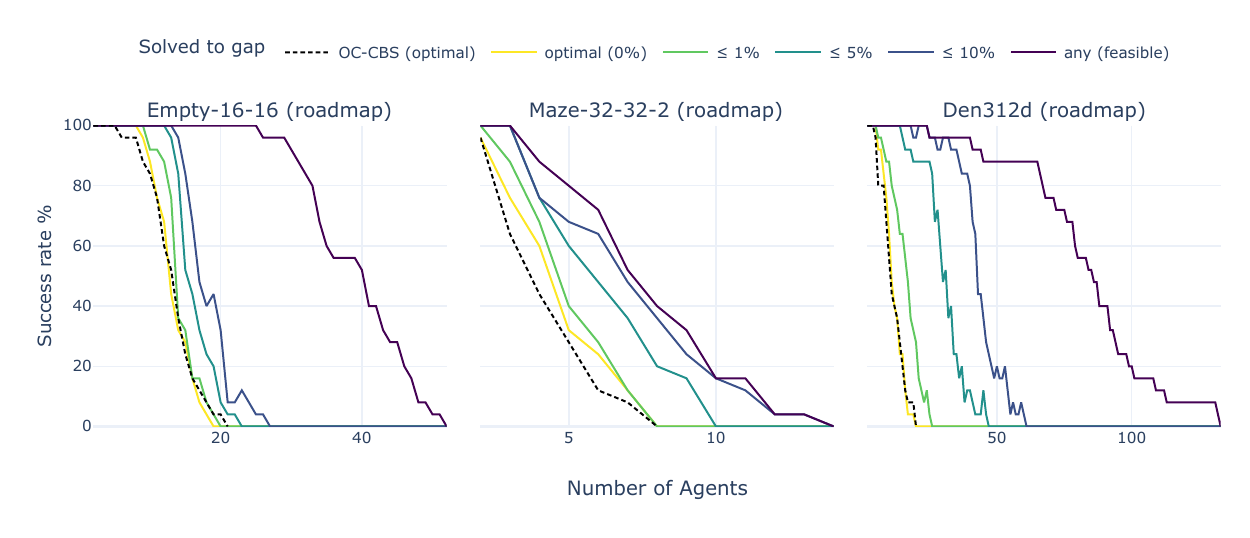}
    \caption{A comparison between \ours and \OCCBS on success rate at various optimality gap upper-bound values, on three different roadmaps.}
    \label{fig:Experiments:MainBenchmark:Roadmaps}
\end{figure}
Across both figures, the benefit of relaxing \gapUB grows with the size of the map. 
On the smallest maps (\MovingAIMap{Empty-16-16}, \MovingAIMap{Room-32-32-4}, and the \MovingAIMap{Maze-32-32-2} roadmap), 
\ours solves $2.8$ to $3.8\times$ as many instances as \OCCBS, 
while on the largest (\MovingAIMap{Room-64-64-8}, \MovingAIMap{Maze-128-128-2}, and the \MovingAIMap{Den312d} roadmap),
the factor is $8.6$ to $12.3\times$. 
Accepting a bounded optimality gap is thus most valuable precisely where optimal methods scale worst.

\subsection{Demonstration of Non-Convex Agents and Smooth Motion}
\label{sec:Experiments:ComplexDemonstration}

This section describes the demonstration in Figure~\ref{fig:Demo} that was first introduced in Section~\ref{sec:Introduction}.
The problem includes four agents, two of which are forklift-shaped and the other two are crane-shaped.
The problem uses the implementation in Section~\ref{sec:Implementation:models:3}, which supports arbitrary 2D polygonal shaped agents and arbitrary connected and continuous trajectories.

Preprocessing consist of two parts, the 
the graph information (Section~\ref{sec:Implementation:graph_information})
and the conflict information (Section~\ref{sec:Implementation:conflict_information}).
All agents move on the same graph, with preprocessing of this graph taking $27$~milliseconds.
The conflict information is computed for each unique agent-type pair:
between two forklifts ($175$~seconds), 
between two cranes ($180$~seconds),
and between a forklift and a crane ($306$~seconds),
all to a precision of $10^{-2}$~second.
This graph and conflict preprocessing is a one-time cost that is valid for any instance on this same map, with any number of forklift and crane agents.
If new types of agents are introduced, only the information involving these new agents must be computed.

\ours took $5.50$~milliseconds to find the optimal solution.

% All-pairs shortest paths time: 0.02261154167354107
% II crane crane: 179.94088991731405
% II fork fork: 174.5344339585863
% II crane fork: 306.1525507499464

\section{Discussion}
\label{sec:Discussion}

The comparison in Section~\ref{sec:Experiments} shows that \ours is able to solve around $5\times$ as many problems across the entire benchmark compared to \OCCBS.
It is important to remember, however, that \OCCBS is designed to only return optimal solutions.
Read as a comparison of exact solvers, the two are roughly equivalent.
\ours proves optimality on $3\,414$ gridmap instances against $3\,269$,
and on $616$ roadmap instances against $592$, only differing by about $4$\%.
The result we draw from this is therefore not that \ours is the stronger optimal solver, 
but that accepting a bounded optimality gap buys roughly a $5\times$ increase in the number of instances solvable within the same budget ($19\,078$ against $3\,269$ on the gridmaps, $3\,220$ against $592$ on the roadmaps), 
and that no existing continuous-time solver offers that trade at a known optimality gap upper bound and gurantees of eventual optimality. 
The size of the gain grows with the map:
$2.8$ to $3.8\times$ on the three smallest maps against $8.6$ to $12.3\times$ on the three largest, 
so bounded suboptimality is worth most where optimal methods scale worst. 
Note however that the reported optimality gap is an upper bound, so the true suboptimality of the returned solutions is unmeasured and may be considerably smaller.

The ablations point consistently to a single design principle: 
finding a solution and proving it optimal are different workloads, and no fixed configuration serves both. 
The repair function \oursTPSIPP extends the reachable agent count several-fold at loose \gapUB, 
yet at $\gapUB=0$ every repairing configuration is worse than its non-repairing counterpart on all three maps (\textbf{C2} below \textbf{C1}, \textbf{C4} below \textbf{C3}), because time spent repairing is time not spent expanding CT nodes and raising the lower bound. 
Careful conflict selection, on the other hand, appears to greatly benefit the search for optimal solutions, which is consistent with prior work~\cite{andreychuk2021improving}.
This effect is less noticeable at higher $\gapUB$.
Parallelism improves solver performance with instance difficulty.
From our limited experiments across this dimension, the throughput increase outweighs the overhead of managing a pool of multiple processes only beyond a certain problem difficulty threshold.
Since \gapUB is observable at runtime, a solver can schedule its policies against it, which is what \textbf{C6} does, 
matching \textbf{C5} at loose \gapUB while recovering \textbf{C3}'s behaviour at
$\gapUB=0$, which no static configuration achieves at both ends. 
We expect this to apply to anytime bounded-suboptimal search more broadly than to \ours specifically.

The benefit of multiple policies is coverage rather than speed, and the aggregate statistics obscure this. 
Eleven workers reach at most $4.77\times$ the node expansion rate of one, at most $43\%$ parallel efficiency, and part of that is spent on nodes a sequential search would not expand ($1.4$ to $3.2\times$ overhead). 
Resultingly, \textbf{C3} solves only $6$--$15\%$ more instances than \textbf{C1}, and is the slower solver below roughly half a second. 
What the policies do buy is $67$ instances that \textbf{C3} proved optimal and \textbf{C1} did not, 
with none in the opposite direction, raising coverage of the paired instances from $92.0\%$ to $96.3\%$. 
For an anytime solver the relevant question is which instances become solvable at all, not how much faster the already-solvable ones are handled.

The conflict-selection results suggest that the value of cardinality reasoning lies in ranking candidate conflicts rather than classifying them, and that a small candidate sample suffices.
\ConflictSelectionSemiCard, which accepts the first conflict raising either child's cost, recovers only $1$ of the $95$-instance gap between \ConflictSelectionArbitrary and \ConflictSelectionBestof{16} on \mapDenSparse and $25$ of $62$ on \mapRoom, 
while paying for an uncapped amount of work; 
whether it is a net gain or a net loss depends on the map. 
By contrast \ConflictSelectionBestof{8} lies within $5$ instances of \ConflictSelectionBestof{\infty}. 
Since each probe costs two low-level searches and the candidate list grows with the agent count, a bounded probe budget is both cheaper and, in these experiments, no less effective. 
Both observations apply to CBS variants in general, including the discrete-time case.

As detailed in Appendix~\ref{sec:Implementation}, the generality of \ours rests on moving domain-specific reasoning offline. 
Geometry, agent shape, and the conflict definition are distilled during pre-processing into information the high-level search queries without interpreting, which is why the same solver handles circular agents on straight-line edges, arbitrary continuous edge trajectories, and non-convex polygonal agents without change to the main solver. 
The cost of generality is therefore paid in pre-processing rather than in solving.
This also bounds what we can currently claim: 
the pre-processing measurements in Figure~\ref{fig:Experiments:preprocessing_times} cover only the simple model of circular agents with straight-line traversals at constant speed, and establishing the corresponding cost for curved trajectories and polygonal agents is ongoing work.

Finally, \ours is implemented in Python and is level at $\gapUB=0$ with a C++ implementation of \OCCBS. 
We read this as evidence that data-structure and pre-processing choices dominate language choice at this scale, rather than as a claim about the algorithms, and it suggests that a compiled implementation of \ours would extend the range of solvable instances further. 
We leave an \ours C++ implementation for future work.

Further work could also investigate the incorporation of disjunctive conflicts~\cite{andreychuk2021improving} and conflict bypassing~\cite{boyarski2015don}.
Just as with careful constraint selection, both aim to reduce the size of the CT and thereby reach an optimal solution in fewer iterations.
Incorporating disjunctive constraints would require a careful analysis of branching soundness to preserve the correctness guarantees, whereas conflict bypassing should require comparatively little effort.
Conflict bypassing means to quickly search for a conflict-free path with an equivalent cost instead of branching in the CT.
We suspect that such a strategy works best on gridmaps where many paths of equivalent cost exist, but not on roadmaps with non-unit length edges.
Furthermore, \oursTPSIPP is only one repair function among possibly many, and \ours supports several repair functions within the same portfolio.
Such repair functions could be paired with the selection policy best suited to them, or invoked at different stages of the search; some may perform better early on, for instance, while others may dominate later, once the underlying \OCCBS search has reached deeper CT nodes closer to a solution.

Further investigation into how \ours scales with the number of CPU cores would also be valuable.
We expect performance to improve with additional cores up to a certain point.
Consider the shallowest CT node representing an optimal solution, at depth $N$.
If every leaf node were expanded at each iteration, as in a breadth-first search, this node would be found within $N$ iterations, but doing so would require an impractically large $2^N$ cores.
The CT node-selection policies instead focus the search on the most promising nodes, so we expect an optimal solution to be reachable in close to $N$ steps without requiring as many as $2^N$ cores.
Where this saturation point lies remains an open empirical question.

Finally, deployment on physical hardware remains the ultimate test for a MAPF solver intended for robotics.
\ours allows for many types of agents, such as land-based robots, aerial vehicles, robotic arms, even non-physical agents such as jobs in a job-shop scheduling problem. 
As long as an agent has an associated state space which it moves through according to a discrete graph, \ours can plan its path to desired positions and trajectories, in the same system as other types of agents.  
Future work could evaluate how planning directly over feasible motions, which \ours enables at higher scalability, affects robustness and solution quality under real-world execution.

\section{Conclusion}
\label{sec:Conclusion}
This work generalised the MAPF problem to heterogeneous agents on individual graphs, continuous time, abstract pairwise conflicts, sequences of edge and vertex tasks, and agents free to move after completing them.
For this generalised problem we presented \ours, a solver that combines the guarantees of an exact solver with the responsiveness of a suboptimal one.
Like \OCCBS, \ours is exact and solution complete: given enough time, it is guaranteed to return an optimal solution.
Unlike \OCCBS, it does not require that time to elapse before returning anything: at every point during its search, \ours makes available an incumbent solution together with a certified upper bound on how far that incumbent can be from optimal, a bound that tightens as the search continues and reaches zero exactly when optimality is reached and verified.
This anytime behaviour comes from decoupling the search that establishes the lower bound from a configurable portfolio of repair functions that searches for incumbents, one instance of which, TP-SIPP, we introduced to handle agents that must vacate a completed task; the same decoupling also allows multiple policies to be expanded in parallel across processor cores.

Our preliminary experiments show that \ours matches \OCCBS at finding and verifying optimal solutions, while solving several times as many problem instances once a bounded optimality gap is accepted, extending scalability from the tens to the hundreds of agents.
The ablations further show that the portfolio's components trade off along different axes rather than uniformly.
Repair functions markedly extend the reachable agent count at loose optimality gaps, but the time they spend searching for incumbents is time not spent raising the lower bound, making them a net cost once the goal is a proven optimum rather than a fast solution.
Additional policies benefit harder instances more than easy ones, since managing a pool of processes costs more than it returns on instances a single core already solves quickly; their main effect is not raw speed but coverage, bringing instances that are otherwise out of reach within reach.
Careful conflict selection helps throughout the search rather than only near the end, and its advantage grows with instance difficulty, but that advantage is largest when the goal is a proven optimum and narrows, without disappearing, as the required gap loosens.
This motivates a portfolio that shifts its configuration as the optimality gap closes, which our gap-cadence configuration does.

Existing continuous-time MAPF solvers force a choice: an exact solver that returns nothing until it has proven optimality, or a suboptimal one that returns quickly but with no guarantee on solution quality, nor of eventual optimality.
\ours removes this trade-off.
By pairing OC-CBS's exactness with an anytime incumbent and a certified optimality gap, it is usable both as a fast heuristic solver and as an exact one, without having to commit to either role in advance.

%%%%%%%%%%%%%%%%%%%%%%%%%%%%%%%%%%%%%%%%%%%%%%%%%%%%%%%%%%%%%%%%%%%%%%%%%%%%%%%%%%%%

\appendix

\section{Implementation}
\label{sec:Implementation}

As discussed in section~\ref{sec:ProblemFormulation:generality}, 
much of the the problem formulation's generality stems from the abstract function $\inconflictfunc$ separating the \ours solver from the specifics of an application. 
This section introduces a practical framework surrounding this abstract $\inconflictfunc$ function, and additionally details our implementation of this framework under multiple different assumptions.
% \AC{All source code, examples, animations, and results are available on our repository\footnote{\ourRepo}.}

\begin{figure}
    \centering
    \includegraphics[width=1\linewidth]{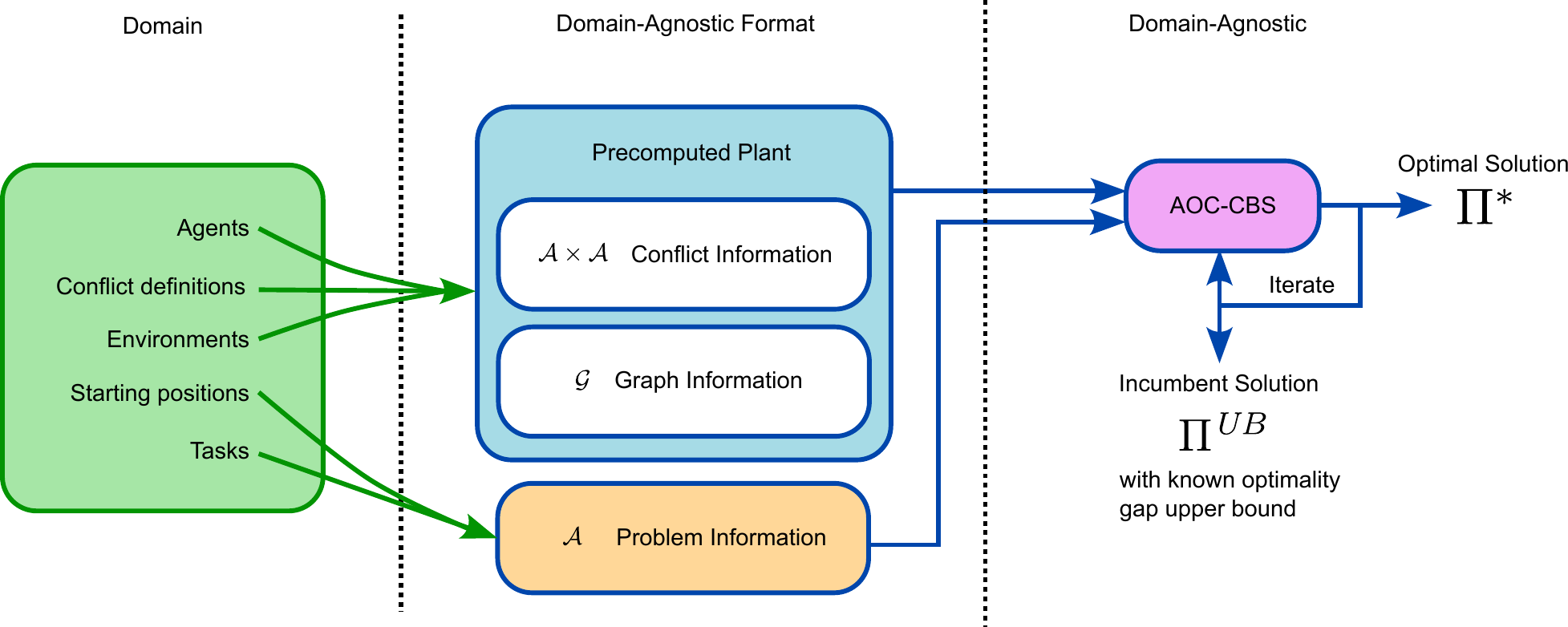}
    \caption{The architecture of the framework surrounding \ours.}
    \label{fig:ImplementationPipeline}
\end{figure}
Figure~\ref{fig:ImplementationPipeline} provides an overview of the architecture surrounding \ours.
On the left is the specific domain, containing the agents, what they are and how they move in their state spaces, the conflict definitions (e.g., geometric collisions, mutually exclusive resources, etc.), the environments in which the agents exist, the agents' starting positions, and their tasks.
This domain-specific information is converted into a domain-agnostic format (center of Figure~\ref{fig:ImplementationPipeline}), consisting of two parts: the \emph{plant} information and the \emph{problem} information.
The plant information is assumed to be known before-hand, including what each agent is and the graph it moves on. 
For instance, we assume that a warehouse knows how many agents they have and what the agents are before the start of operations, and that this information is static.
The precomputed information contains conflict information between agents as well as information about the graphs themselves that is useful during single-agent path planning.
The problem information is not known before-hand and contains each agent's starting position and its sequence of tasks; this information cannot be utilized outside of the runtime.
The plan and problem information are specific to the domain, encoding all that is required about the specific domain, but is saved in an domain-agnostic format.
That is, the format of this information is the same, regardless of if the domain is a logistical warehouse, and manufacturing cell of robotic arms, or a swarm of drones.
Finally, the plant information and problem information is fed to \ours, upon which the iterations begin.
At every iteration, an incumbent solution $\incumbentSolution$ (once one has been found) with a known optimality gap upper bound is made available.
Within a finite amount of time, \ours halts with an optimal solution $\JointPlan^*$.

\subsection{Precomputed Conflict Information}
\label{sec:Implementation:conflict_information}

% The plant contains the agents $\Agents$ which we assume are known before-hand, but not each agent's starting vertex or task sequence. 
% Thus, all information about the plant can be precomputed and stored, to be used during deployment on a specific problem.
The conflict information distils the domain-specific conflict definitions into a set of time intervals, an \emph{annotation} technique first introduced in~\cite{PSIPP}.
These time intervals can be precomputed and stored in lookup tables for use during runtime.
This conflict information is stored in a domain-agnostic format and is what \ours consumes.

For any agent pair $\tuple{i, j}\in\Agents\times\Agents: i\neq j$, the following conflict information can be computed and stored in lookup tables:
\begin{enumerate}
    \item \textbf{Vertex-Vertex}: $\forall \tuple{\vertex^i,\vertex^j}\in\Vertices^i\times\Vertices^j$, is $\inconflict{i, \vertex^i_\state, j, \vertex^j_\state}$?
    \item \textbf{Vertex-Edge}: 
    \begin{itemize}
        \item $\forall \tuple{\vertex^i,\edge^j}\in\Vertices^i\times\Edges^j$, for which times $t\in[0, \edge^j_\dur)$ is $\inconflict{i, \vertex^i_\state, j, \edge^j_\traj(t)}$?
        \item $\forall \tuple{\edge^i,\vertex^j}\in\Edges^i\times\Vertices^j$, for which times $t\in[0, \edge^i_\dur)$ is $\inconflict{i, \edge^i_\traj(t), j, \vertex^j_\state}$?
    \end{itemize}
    \item \textbf{Edge-Edge}: $\forall \tuple{\edge^i, \edge^j}\in\Edges^i\times\Edges^j$, for which start times $\delta\in[-\edge^i_\dur, \edge^j_\dur)$ for $i$ to traverse $\edge^i$ does there exists a $t\in[\max(0, \delta), \min(\edge^i_\dur, \edge^j_\dur))$ where $\inconflict{i, \edge^i_\traj(\delta+t), j, \edge^j_\traj(t)}$?
\end{enumerate}
Informally, the Vertex-Vertex case simply tell us which vertices can the two agents not be located at simultaneousness, 
the Vertex-Edge cases tell us during what times do the agents conflict if one waits at a vertex and the other begins traversing its edge at time $t=0$,
and the Edge-Edge case tells us at what times can one agent begin traversing an edge given that the other begins traversing its edge at time $t=0$.
This conflict information is what~\cite{OCCBS} refers to as \emph{intersection intervals} and what~\cite{CCBS_revisit} refers to as \emph{conflict intervals}.
The bounds on $t$ for the Vertex-Edge cases and $\delta$ and $t$ for the Edge-Edge case exist since we do not consider conflicts outside of the time during which the edge trajectories are defined.

The Vertex-Vertex case results in a countable and bounded list of mutually exclusive vertices, however, the remaining cases work over continuous times and therefore have certain requirements for these computations to be practically feasible. 
Under a finite-precision implementation with a minimum precision $\epsmach$, two real values differing by less than $\epsmach$ are indistinguishable and by only a few $\epsmach$ indistinguishable from rounding errors. 
Any conclusions at these scales unreliable.
For this reason, we assume that all conflicts have a duration of at least a few order of magnitude above $\epsmach$. 
We then know that over the searchable intervals ($t\in[0, \edge^j_\dur)$ and $t\in[0, \edge^i_\dur)$ for Vertex-Edge, $\delta\in[-\edge^i_\dur, \edge^j_\dur)$ for Edge-Edge) the times at which conflicts occur are grouped into maximally connected, non-zero-sized intervals.
Searching for these intervals can be done using established methods (e.g., using Lipschitz bounds as described in Section~\ref{sec:Implementation:LipschitzBounds} below) to a specified precision.
Furthermore, conservative estimations of these interval can be used, where the size of the conflict intervals are inflated, to include a safety margin to also allow for courser precision when searching for these intervals. 

This conflict information contains all that is needed to implement conflict avoidance. The Vertex-Vertex information provides mutually exclusive vertex-pairs that are used to avoid wait-wait conflicts. Although, we note that wait-wait conflicts are not reasoned over in the main \ours algorithm since, assuming the agents start in conflict-free positions, wait-wait conflicts must always be preceded by a move-wait or move-move conflict. However, this information is still useful for various optimizations.
The Vertex-Edge information provides times when an agent cannot wait at a vertex if another traverses and edge at time $t=0$; since the graphs are time-invariant, this information can be shifted in time.
Likewise for the Edge-Edge information that assumes one agent traverses its edge at $t=0$, this information can be shifted in time. 

Although it is stated in the above that these computations must be done for each ordered agent-pair $\tuple{i, j}\in\Agents\times\Agents: i\neq j$, in practice the above information computed for $\tuple{i,j}$ is easily converted for $\tuple{j, i}$: 
this is obvious for Vertex-Vertex, 
Vertex-Edge are already bi-directional, and for Edge-Edge the intervals can be negated and inverted since it is the relative time difference between the traversals that matters.
That is, if $i$ cannot begin its traversal in $[t_1, t_2)$ when $j$ begins its traversal at $t=0$, then if instead $i$ begins its traversal at $t=0$ then $j$ cannot begin its traversal in $[-t_2, -t_1)$. 
Thus, the worse-case number of agent pairs to compute for is reduced by roughly half.

Furthermore, coarse filtering can be applied to quickly reduce the number of graph element pairs to check for collisions. For instance, if two graph elements are too far from each other for two agents to possibly conflict while one each respective element, then computationally intensive conflict detection is not required. 
This further reduces the computational work for preprocessing.

Many applications may have multiple agents of the same \emph{agent model}, $\model^\agent$ for agent $\agent$. 
Although it is sufficient for our framework to let $\model^\agent$ remain abstract, $\model^\agent$ could for instance specify that $\agent$ is a robotic arm, as well as all the attributes required to determine how it interacts with other agents.
Importantly, if $\inconflict{i, \state_i, j, \state_j}$ then for another agent $k$ with $\model^k=\model^i$ we know that $\inconflict{k, \state_k, j, \state_j}$.
Given this, the conflict information does not need to be computed for every (unordered) agent pair, but instead every unordered pair of unique $\tuple{\model^\agent, \Graph^\agent}$ over $\agent\in\Agents$.
In the best case, all agents have the same model and operate on the same graph (a common assumption in the MAPF literature), requiring only a single instance of the conflict information.

Since this conflict information is the only connection between \ours and the specific application over which \ours plans, all correctness guarantees that we provide here and the validity of any returned solution is with respect to this conflict information.

\subsubsection{Lipschitz-Bounded Conflict Search}
\label{sec:Implementation:LipschitzBounds}

The Vertex-Edge and Edge-Edge conflict information consists of intervals along a continuum of time or offset, and reliably computing such an interval requires verifying that every point within it corresponds to a conflict, without evaluating each of the infinitely many points it contains. 
This section develops a general method for doing so using Lipschitz bounds, verifying whole intervals from a small number of pointwise evaluations.

Conflict is determined using two components: a \emph{conflict predicate} and a \emph{Lipschitz bound}. 
A conflict predicate $f(i, \state_i, j, \state_j)$ encodes whether agent $i$ in state $\state_i$ and agent $j$ in state $\state_j$ are in conflict, via $f(i, \state_i, j, \state_j) \leq 0 \Leftrightarrow \inconflict{i, \state_i, j, \state_j}$. 
A Lipschitz constant $L_f$ bounds how fast $f$ can change with respect to the relative position between the two agents' centres, i.e. $|f(i, \state_i, j, \state_j) - f(i, \state_i', j, \state_j')| \leq L_f \cdot | \Delta \position  -  \Delta \position'|$ for any two relative positions $\Delta \position$ and $\Delta \position'$.

These three components are used together to efficiently partition time into conflicting and non-conflicting intervals, without needing a closed-form solution for when conflicts occur. 
Consider agents $i$ and $j$ following trajectories $\edge_\traj^i$ and $\edge_\traj^j$, 
with durations $\edge^i_\dur$ and $\edge^j_\dur$, 
and known maximum speeds $|\velocity^i|_\mathrm{max}$ and $|\velocity^j|_\mathrm{max}$ respectively. 
At some time $t$, the predicate can be evaluated as $f(t) = f(i, \edge_\traj^i(t), j, \edge_\traj^j(t))$. 
Suppose $f(t) > 0$, i.e. at time $t$ the agents are not in conflict. 
Since time is continuous, ruling out only the single instant $t$ is not useful on its own. 
Instead, a window around $t$ can be ruled out directly using the Lipschitz bound. 
The relative position between the agents can change at a rate of at most $|\velocity^i|_\mathrm{max} + |\velocity^j|_\mathrm{max}$, 
and by the Lipschitz bound, $f$ can therefore change by at most $L_f \cdot (|\velocity^i|_\mathrm{max} + |\velocity^j|_\mathrm{max})$ per unit time. 
Consequently, over a duration of magnitude at most
\begin{equation*}
    \Delta t = \frac{f(t)}{L_f \cdot (|\velocity^i|_\mathrm{max} + |\velocity^j|_\mathrm{max} )}
\end{equation*}
$f$ cannot have decreased to $0$, so no conflict can occur within the interval $[t-\Delta t, t + \Delta t]$.
This interval can be skipped without further evaluation, and the process repeated from its boundary, allowing the conflict-free and conflicting regions of time to be identified using only pointwise evaluations of $f$.
The same procedure symmetrically holds if $f(t)\leq0$ (i.e., a conflict at time $t$) instead; $[t-\Delta t, t + \Delta t]$ is an interval of only conflict.

This method can be repeated at each boundary of verified conflict/no conflict interval until a desired precision is reached.
In practice, this method is not run to converge on a single point since that would require unboundedly many evaluations near any boundary where $f(t)$ approaches $0$.
Any interval that remains unresolved when the search terminates (that is, the small intervals between a verified conflict and verified non-conflict interval of size less than the desired precision) is conservatively treated as a conflict interval.
This yields an estimate where no genuinely conflicting interval is misclassified.
The cost of this guarantee is that a small margin of conflict-free interval is classified as conflicting.
However, for the intended use of this conflict information, a false conflict merely costs some avoidable caution, while the opposite would be unsound.

\subsection{Precomputed Graph Information}
\label{sec:Implementation:graph_information}

Besides precomputing conflict information from above, we also precompute information pertaining to a single graph $\Graph = \tuple{\Vertices, \Edges}$.
The low-level \oursPP (Section~\ref{sec:Method:low-level}) and \oursTPSIPP (Section~\ref{sec:Method:TPSIPP}) perform \Astar searches to construct a lowest-cost plan for a single agent.
The heuristic used in these \Astar searches take advantage of the minimum cost from any vertex to any other vertex, which is precomputed and stored. 
The runtime cost of computing the all-pairs shortest paths in a non-negative weighted graph is $\complexity{|\Vertices|^3}$ with the Floyd-Warshall algorithm or, on sparse graphs where $|\Edges| \ll |\Vertices|$, using Dijkstra's with a Fibonacci heap gives $\complexity{|\Vertices|^2\log|\Vertices| + |\Vertices||\Edges|}$~\cite{cormen2022introduction}.
The memory cost is $\complexity{|\Vertices|^2}$.

In our implementation, this graph information is computed and stored on disk. 
When a particular problem is given, we load into RAM only the shortest-path distances to the source vertex of every agent's tasks, as that is the only information needed during the \Astar searches. 
This reduces the amount of RAM storage needed (significant for large graphs) while avoiding computing this information at runtime.

\subsection{Agent and Environment Models}

The conflict information (Section~\ref{sec:Implementation:conflict_information}) and graph information (Section~\ref{sec:Implementation:graph_information}) build on an underlying model of the agents, the graphs, and the conflict definitions.
We currently provide implementations for the following MAPF-grounded models
\begin{itemize}
    \item[] \textbf{\nameref{sec:Implementation:models:1}} (Section \ref{sec:Implementation:models:1}): the common MAPF assumptions---agents are circular, omnidirectional, and without acceleration constraints. Edge trajectories are straight lines at constant speed, and all vertices are waitable. Conflicts are based on geometric collision. Unlike most MAPF problems, agents here may still have unique radii and move on unique graphs.
    \item[] \textbf{\nameref{sec:Implementation:models:2}} (Section \ref{sec:Implementation:models:2}): extends \nameref{sec:Implementation:models:1} with arbitrary connected and continuous trajectories. Vertices may be non-waitable.
    \item[] \textbf{\nameref{sec:Implementation:models:3}} (Section \ref{sec:Implementation:models:3}): extends \nameref{sec:Implementation:models:2} with arbitrary 2D polygonal-shaped agents.
    % \item[] \textbf{\nameref{sec:Implementation:models:4}} (Section \ref{sec:Implementation:models:4}): extends \nameref{sec:Implementation:models:2} to 3D, where agents are spherical.
    % \item[] \textbf{\nameref{sec:Implementation:models:5}} (Section \ref{sec:Implementation:models:5}): extends \nameref{sec:Implementation:models:4} with arbitrary mesh-defined agent shapes.
\end{itemize}
% The models progressively get more and more general, allowing for more problems to be modelled.
% However, the more general the model becomes, the fewer assumptions can be used for speeding up the precomputation of conflict and graph information.
% For instance, conflict information in \nameref{sec:Implementation:models:1} can be computed in closed-form, but not the preceding models. 
% We stress, however, that these implementations remain outside the core \ours solver, making it possible for additional models to be developed for other applications and then solved using \ours. 
\AC{Work is ongoing to extend these models into 3D space.}

\subsubsection{Model 1}
\label{sec:Implementation:models:1}

This model follows the common MAPF assumptions where 
agents and graphs are defined in 2D Euclidian space, 
the agents are circular (although here they have different radii and exist on separate graphs),
and edge traversals are done in straight lines at constant speed from the source to target vertex.
An agent's state space consists of 2D position only; orientation does not matter for circular agents, and velocity is not necessary since all vertices are waitable.
Conflicts include only geometric collision: two agents conflict when their circles overlap.

This setting allows for the closed-form computation of conflict information~\cite{walker2019collision, PSIPP, combrink2025advances} (see~\cite{combrink2025advances} for a visual description), making this model the fastest to precompute out of those provided in this work.
Vertex-Vertex information reduces to simply checking if the distance $d$ between two agents at their respective vertices is less than the sum of their radii, $d < r_1 + r_2$,
and Vertex-Edge requires the same check for the distance $d$ from the stationary agent's vertex to the nearest point on the line segment of the moving agent's traversal.
Edge-Edge collisions instead form a conic section $d^2(t,\delta)$ representing the squared distance between the two agents, as a function of time $t$ and the one agent's starting time $\delta$ (the other agent starting at $t=0$).
The interior of the ellipse $d^2(t, \delta)=0$ in the $(t,\delta)$-plane represents a guaranteed collision only if the agents travel along the infinite line between their respective source and target vertices. 
The time when both agents are actually traversing their respective line segments (instead of the infinite lines) form a parallelogram in the $(t,\delta)$-plane.
For every point $(t,\delta)$ within both the parallelogram and the ellipse, the agents will be colliding with each other at time $t$ if the one agent traverses its edge starting at time $\delta$.
The points minimizing and maximizing $\delta$ within this intersection form a single interval, which is what we store as the Edge-Edge conflict information.

\subsubsection{Model 2}
\label{sec:Implementation:models:2}

This model extends \nameref{sec:Implementation:models:1} to include arbitrary connected and continuous edge trajectories.
We describe how trajectories are represented and how the conflict information can be computed efficiently.

A state in this model is a pair $(\position, \velocity) \in \Real^2 \times \Real^2$ consisting of a position $\position$ and velocity $\velocity$;
orientation is omitted since agents are circular.
Each vertex $\vertex$ corresponds to a specific state $\vertex_\state = \tuple{\position_\vertex, \velocity_\vertex}$, and is waitable if $\velocity_\vertex=\mathbf{0}$.
The trajectory $\edge_\traj(t)$ of an edge $\edge\in\Edges$ from $v\in\Vertices$ to $u\in\Vertices$ is defined by a pair of cubic splines, one for each of the two dimensions, interpolated through a sequence of position-time waypoints with clamped boundary conditions.
At the first waypoint, $\edge_\traj(0)=\position_v$ with $\edge'_\traj(0) = \velocity_v$, and at the last, $\edge_\traj(\edge_\dur) = \position_u$  with $\edge'_\traj(\edge_\dur)) = \velocity_u$.
The trajectory passes through each intermediate waypoint at the position and time it specifies, with the velocity and acceleration of the incoming segment matching those of the outgoing segment.

Unlike in the previous model where straight-line trajectories allow for the closed-form computation of conflict information, 
here the arbitrarily shaped trajectories warrant the use of the Lipschitz bounded method (Section~\ref{sec:Implementation:LipschitzBounds}).
Additionally, since conflict is collision-based and therefore depends on the spatial separation of agents, 
we additionally make use of an \emph{interaction radius}
which is the largest possible center-to-center separation between two agents at which a conflict remains possible.
That is, if the agents' centres are further apart than the sum of their interaction radii, no conflict can occur regardless of their orientation or state.

In this model, agents are circular, and conflict is defined by circle separation: two agents with radii $r_i$ and $r_j$ are in conflict when the distance between their centres is at most $r_i + r_j$. 
This gives the straight-forward conflict predicate
\begin{equation*}
    f(i, \state_i, j, \state_j) = | \position_j - \position_i| - (r_i + r_j),
\end{equation*}
where $\position_i$ and $\position_j$ are the positions extracted from $\state_i$ and $\state_j$. 
Under this predicate, the interaction radius is exactly $r_i + r_j$, since $f\leq0$ is possible only when the centres are within this distance. 
The Lipschitz constant is $L_f = 1$, since $f$ is a norm of the relative position minus a constant, and the norm is 1-Lipschitz.

The search for finding all necessary conflict information consists of two levels:
the \emph{broadphase} and the \emph{narrowphase}.
The broadphase quickly filters out graph element pairs that cannot be involved in a conflict; 
two edges on separate sides of the environment do not require expensive conflict detection if they are too far away for agents on each edge to possibly conflict with each other.
Each vertex and edge gets a swept bounding box, inflated by the interaction radius;
non-overlapping pairs are discarded without evaluating $f$ since the agents never get within the interaction radius.
For edges, this box is computed exactly via root-finding on the spline's derivative polynomials. 
Non-waitable vertices (those with non-zero velocity components) are excluded from Vertex-Vertex and as Vertex-Edge sources, since they cannot produce a waiting conflict.

The narrowphase evaluates every graph element pair that passed the broadphase.
Vertex-Vertex needs one evaluation on $f$, just as in \nameref{sec:Implementation:models:1}.
For Vertex-Edge,
the waiting agent $i$ is kept stationary at $\vertex^i_\state$ and we use the Lipschitz-bounded conflict search to find all
$T\in[0, \edge^j_\dur): \varphi(T) \leq 0$ where $\varphi(T) = f(i, \vertex^i_\state, j, \edge^j_\traj(T))$
For Edge-Edge, a nested search is required instead. 
We fix $j$ to begin traversing $\edge^j$ at time $t=0$, 
and seek all start times $\delta$ for $i$ to begin traversing $\edge^i$ such that the agents conflict while traversing their respective edges.
The outer search therefore is done over $i$'s start times $\delta\in[-\edge^i_\dur, \edge^j_\dur)$ using the predicate
\begin{equation*}
    h(\delta) = \min_{t} g(t,\delta),
\end{equation*}
with the inner search over $t\in[\max(0, \delta), \min(\edge^i_\dur, \delta + \edge^j_\dur)]$ using the predicate
\begin{equation*}
    g(t,\delta) = f(i, \edge^i_\traj(t), j, \edge^j_\traj(t-\delta)).
\end{equation*}

Each evaluation of $h(\delta)$ requires an inner minimization over $t$, since two agents traversing with offset $\delta$ are conflict-free only if $g(t,\delta)>0$ for every $t$ in the valid window, 
which is exactly the condition $\min_{t} g(t,\delta)$.
However, the inner search does not need the true minimum, only its sign:
as soon as any evaluated $t$ yield $g(t,\delta)\leq0$, a conflict at offset $\delta$ is already certified, and the search terminates immediately.
Only when no such $t$ is found does the search need to continue until every remaining cell has been certified $g>0$, at which point $\delta$ is certified conflict-free. 
This asymmetry means a conflicting $\delta$ is typically resolved after a few evaluations, while a conflict-free $\delta$ requires exhausting the entire inner domain. 

Both searches rely only on a Lipschitz bound on their respective objective, obtained from, obtained from $L_f$ by the chain rule.
For the inner search, both agents move, giving $|\partial g/\partial t| \leq L_f \cdot (|\velocity^i|_\mathrm{max} + |\velocity^j|_\mathrm{max})$; 
for the outer search, only $\edge^j$'s traversal shifts with $\delta$, 
giving $|\partial h/\partial \delta| \leq L_f \cdot |\velocity^j|_\mathrm{max}$. 
The smaller outer constant lets each evaluation of $h$ certify a wider band of $\delta$ at once than a comparable evaluation of $g$ could certify for $t$,
and since each evaluation of $h$ triggers an entire inner search, fewer outer evaluations directly reduce the total number of inner minimisations needed.

Given a Lipschitz constant $L$ for either search, a cell $[a,b]$ with midpoint $m$ and width $w$ can be certified through only an evaluation at $m$:
if $\varphi(n) - Lw/2 >0$ the cell is conflict-free and discarded.
If instead $\varphi(m) + Lw/2 \leq 0$ the cell is entirely conflicting and kept whole, otherwise it is bisected further.
The Vertex-Edge search, and the outer Edge-Edge search over $\delta$,
apply this test directly to find the sublevel set $\{\delta: h(\delta) \leq 0\}$.
The inner Edge-Edge minimisation over $t$ uses the same bound within a branch-and-bound search that always expands the cell with the lowest lower bound $\varphi(m) - Lw/2$.

The resulting conflict information is a one-sided over-approximation: 
a cell is discarded only when the Lipschitz certificate proves $f>0$ throughout it,
so a conflicting offset $\delta$ is never certified conflict-free, though a marginally conflict-free offset may conservatively be treated as conflicting.

\subsubsection{Model 3}
\label{sec:Implementation:models:3}

This model extends \nameref{sec:Implementation:models:2} by replacing the circular agent shape with a 2D polygon, which may be non-convex.
Everything else about the framework is unchanged: the same broadphase/narrowphase search and the same Lipschitz-bounded certified sampling produce the conflict information.
What changes is the state space, the set of edge trajectories, and the conflict predicate.

A polygon is not rotation-invariant, so orientation becomes part of the state: a state is a triple $\tuple{\position, \orientation, \velocity}$ consisting of a position $\position\in\Real^2$, an orientation $\orientation\in\Real$, and a velocity $\velocity\in\Real^2$.
An agent's shape is given as an ordered list of polygon vertices in the agent's body frame, relative to its reference point; placing the agent in a state rotates those vertices by $\orientation$ about the reference point and translates them by $\position$.
As in \nameref{sec:Implementation:models:2}, each vertex of the graph corresponds to a specific state, and is waitable if its velocity is zero.

Adding orientation admits edge trajectories that \nameref{sec:Implementation:models:2} cannot express, since a trajectory must now prescribe a heading in addition to a position at every $t$.
Nothing in the method depends on the specific choice of trajectory, only on the properties it exposes: a trajectory need only be continuous, connect its endpoint states, and report its duration, a bounding box, and a bound on how fast the agent's body can move.
This abstraction accommodates trajectory types ranging from a simple turn in place or straight-line traversal at a constant held heading, to curvature-constrained curves (e.g., Dubins) or a position spline paired with an independently prescribed orientation schedule.

Conflict is again geometric collision, but between two placed polygons rather than two circles.
We use the signed separation
\begin{equation*}
    f(i, \state_i, j, \state_j) =
    \begin{cases}
        -d & \text{if the placed polygons overlap,} \\
        +d & \text{otherwise,}
    \end{cases}
\end{equation*}
where $d$ is the minimum Euclidean distance between the two placed polygons' boundaries, taken as the minimum over all pairs of boundary segments, one from each agent.
Overlap is detected by two tests: whether any boundary segment of one agent transversally crosses a boundary segment of the other, and whether either polygon contains a vertex of the other.
Containment uses even-odd ray casting, which is correct for any simple polygon, so no convexity is assumed anywhere and no convex decomposition is performed.
Disjoint agent shapes therefore give $f=d>0$, one shape fully inside the other gives $f=-d<0$, and a partial overlap, where the boundaries touch, gives $d=0$ and hence $f=0$.
The set $\{f\leq0\}$ is thus exactly the overlap set, which is what soundness requires.
Note that $f$ is a boundary distance and not a penetration depth, so it does not necessarily continue to decrease as a partial overlap deepens; the only consequence is that the certified search bisects somewhat more near deep overlaps than a true penetration depth would require.

The interaction radius is the sum of the two agents' circumradii, where an agent's circumradius is the largest distance from its reference point to any of its polygon vertices.
Beyond that separation the agent shapes cannot touch, whatever their orientations.
The Lipschitz constant is again $L_f=1$, since a boundary distance changes by at most the distance any body point moves.

This last point is where \nameref{sec:Implementation:models:3} differs from \nameref{sec:Implementation:models:2} in a way that matters for the certified search.
In \nameref{sec:Implementation:models:2} the chain rule converts $L_f$ into a bound in time using each agent's maximum translational speed, because a circular agent shape is completely described by its centre.
A rotating polygon, by contrast, sweeps area even while its reference point is stationary, so the relevant quantity is the speed of the fastest point of the body: the translational speed plus the angular speed multiplied by the circumradius.
Each trajectory reports this quantity.
A turn in place has zero translational speed and reports (angular speed $\times$ circumradius); a Dubins curve travelling at speed $v$ with turning radius $R$ reports $v(1 + r_\mathrm{body}/R)$; a spline with an orientation schedule reports the spline's maximum speed plus the maximum of the schedule's angular rate times the circumradius.
With this bound in place, the Vertex-Edge and Edge-Edge searches of \nameref{sec:Implementation:models:2}, including the nested inner minimisation over $t$ for Edge-Edge, are used verbatim, and inherit the same guarantee: the resulting conflict information is a one-sided over-approximation, never certifying a conflicting offset as conflict-free.

The broadphase likewise carries over: every vertex and edge gets a bounding box inflated by the interaction radius, and non-overlapping pairs are discarded.
A turn-in-place edge has a point-sized box, since its reference point does not move, and the inflation by the two circumradii covers the whole area the rotating body sweeps.
A Dubins edge's box is obtained from its straight segments' endpoints together with the bounding boxes of its arcs' supporting circles.

Polygon separation is by far the most expensive of the predicates we implement: a single evaluation performs a segment-to-segment distance computation for every pair of boundary segments of an agent's shape, one from each agent, and is therefore quadratic in the agent shapes' vertex counts, where the circle predicate is a single norm.
Since the certified search calls the predicate many times per graph element pair, this makes \nameref{sec:Implementation:models:3} the most expensive model to precompute, and is the reason \nameref{sec:Implementation:models:1} and \nameref{sec:Implementation:models:2} retain the circle predicate rather than treating a circle as a many-sided polygon.

Finally, a polygonal agent shape does not require a graph with orientation in its states.
A polygon held at a fixed orientation can be used directly on a \nameref{sec:Implementation:models:1} graph, and polygon-with-circle agent pairs are supported, so a polygonal ground robot can share an environment with a circular agent.
Whichever route is taken, the resulting conflict information is written in the same format, and \ours cannot distinguish the two.

\section{Roadmap Generation}
\label{sec:BenchmarkGeneration}

The MovingAI gridmaps~\cite{SturtevantBenchmarks} are the standard benchmark for discrete-time MAPF. 
These gridmaps are represented by a connected, undirected graphs with unit-length edges. 
Roadmap instances for continuous-time MAPF exist in prior work, e.g., the \MovingAIMap{Den520d} roadmaps used by CCBS~\cite{CCBS} and OC-CBS~\cite{OCCBS}, and PRM-sampled instanced from~\cite{PSIPP},
but neither offers a systematic scheme for deriving roadmaps from the standard MovingAI benchmark suite with guaranteed connectivity, obstacle clearance, and bounded stretch. 
We present such a scheme, apply it across the full suite, and release the resulting roadmaps together with matching scenario files. 

% \begin{figure}[htbp]
%     \begin{subfigure}[t]{0.24\textwidth}
%         \centering
%         \includegraphics[width=\linewidth]{figures/BenchmarkMaps/empty-16-16.pdf}
%         \caption{Empty-16-16}
%         \label{fig:MovingAI_maps:empty-16-16}
%     \end{subfigure}
%     \hfill
%     \begin{subfigure}[t]{0.24\textwidth}
%         \centering
%         \includegraphics[width=\linewidth]{figures/BenchmarkMaps/maze-32-32-2.pdf}
%         \caption{Maze-32-32-2}
%         \label{fig:MovingAI_maps:maze-32-32-2}
%     \end{subfigure}
%     \hfill
%     \begin{subfigure}[t]{0.24\textwidth}
%         \centering
%         \includegraphics[width=\linewidth]{figures/BenchmarkMaps/room-32-32-4.pdf}
%         \caption{Room-32-32-4}
%         \label{fig:MovingAI_maps:room-32-32-4}
%     \end{subfigure}
%     \hfill
%     \begin{subfigure}[t]{0.24\textwidth}
%         \centering
%         \includegraphics[width=\linewidth]{figures/BenchmarkMaps/den312d.pdf}
%         \caption{Den312d}
%         \label{fig:MovingAI_maps:den312d}
%     \end{subfigure}
%     \caption{MovingAI maps}
%     \label{fig:MovingAI_maps}
% \end{figure}

\subsection{Sampling Roadmaps From Gridmaps}

Given a MovingAI map, its free cells define a free space $\freeCells \subset \Real^2$, with every unit cell centred on an integer coordinate and everything outside the map counted as blocked.
To turn this into a roadmap $\Graph = \tuple{\Vertices, \Edges}$, we require four parameters: a target vertex density $\density>0$, a clearance radius $\clearance>0$, a stretch factor $\stretchfactor \geq 1$, and a maximum edge length $\maxlength>0$.
Figure~\ref{fig:Roadmap_construction} provides an illustration of the sampling steps. 
Informally, 
the free cells are divided into finer grid cells; 
these fine cells within $\clearance$ from blocked cells are removed.
Vertices are scattered evenly across the free space at roughly the requested density $\density$.
Each pair of vertices whose free-space Voronoi regions are adjacent (Figure~\ref{fig:Roadmap_construction:regions}) is joined along a route through the boundary between them, reduced to as few straight segments as possible (Figure~\ref{fig:Roadmap_construction:result});
and shortcut edges shorter than $\maxlength$ are added wherever this graph would otherwise force a detour between two vertices that is more than $\stretchfactor$ times the straight-line distance.
The resulting roadmaps using three different densities on the \MovingAIMap{Den312d} gridmap are shown in Figure~\ref{fig:Roadmap_construction:densities}.
We detail each of the steps in turn.

\begin{figure}[htbp]
    \begin{subfigure}[t]{0.47\textwidth}
        \centering
        \includegraphics[width=\linewidth]{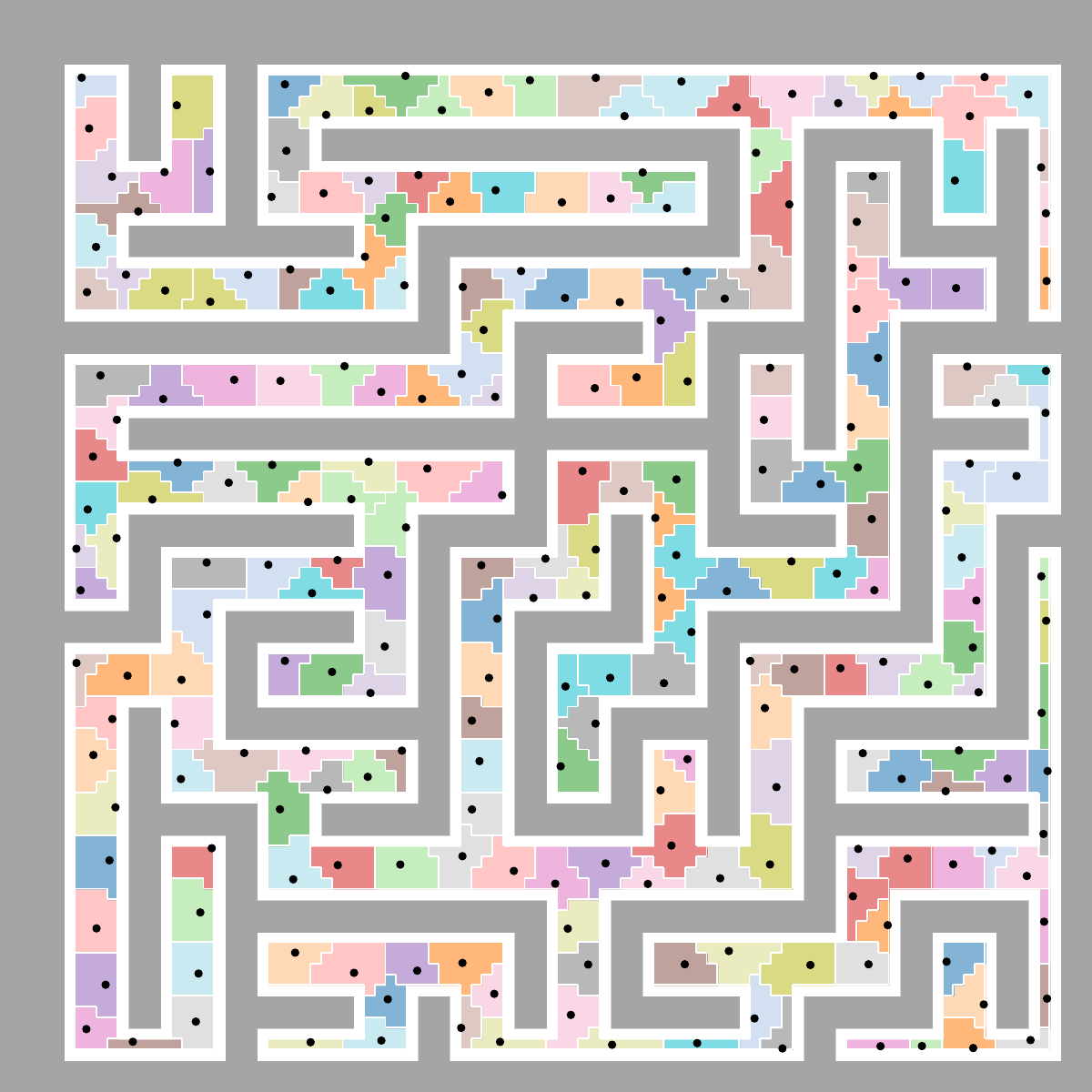}
        \caption{The free space is divided into finer cells (individual coloured squares); fine cells are removed if too close to blocked cells; vertices (dots) are scattered to roughly match a given density; and each vertex claims the fine cells around it (coloured regions).}
        \label{fig:Roadmap_construction:regions}
    \end{subfigure}
    \hfill
    \begin{subfigure}[t]{0.47\textwidth}
        \centering
        \includegraphics[width=\linewidth]{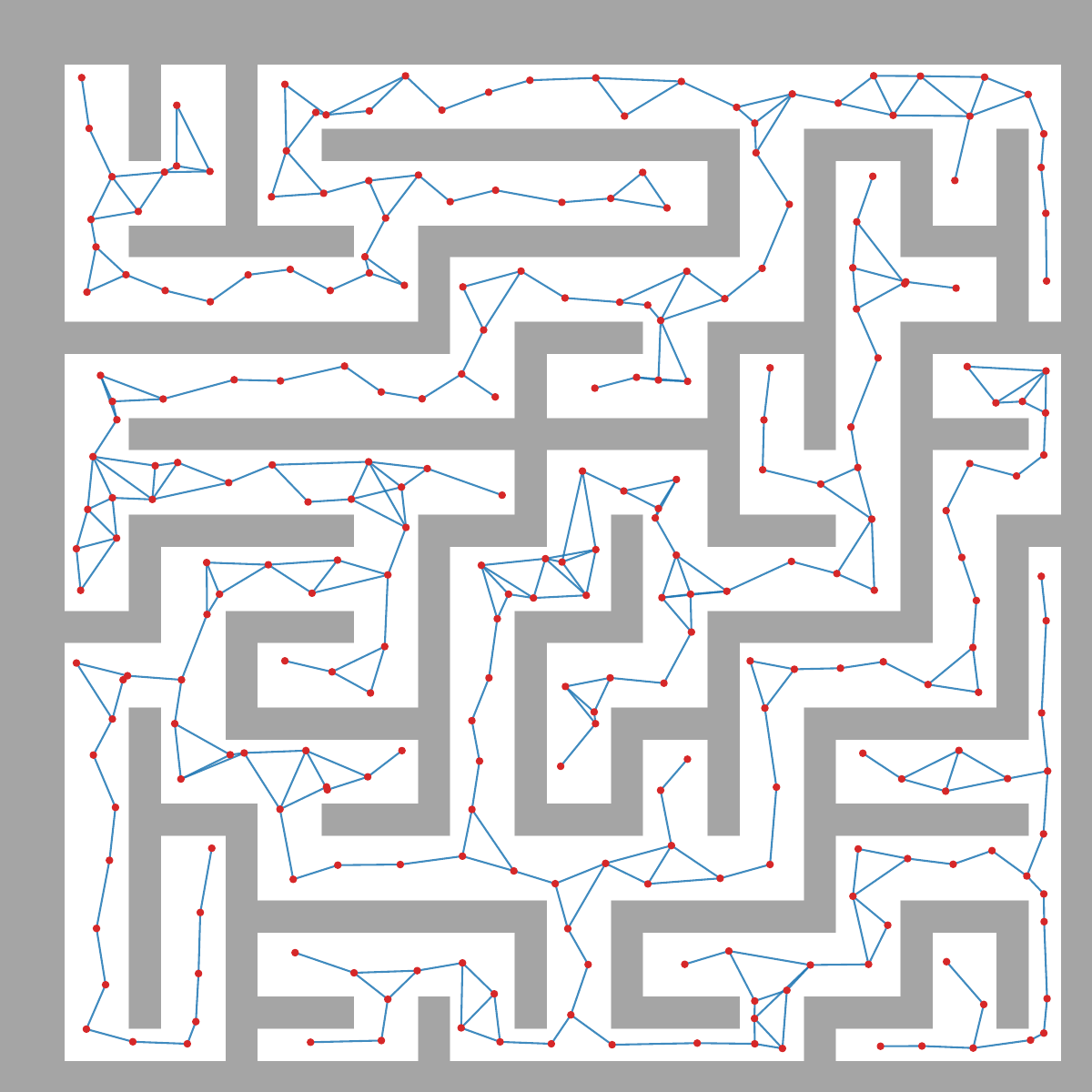}
        \caption{Any two vertices that share a border are connected by as few straight edges as possible while ensuring no edge gets too close to a blocked cell.}
        \label{fig:Roadmap_construction:result}
    \end{subfigure}
    \caption{A roadmap sampled from the MovingAI gridmap \MovingAIMap{Maze-32-32-2}.}
    \label{fig:Roadmap_construction}
\end{figure}
\begin{figure}[htbp]
    \begin{subfigure}[t]{0.31\textwidth}
        \centering
        \includegraphics[width=\linewidth]{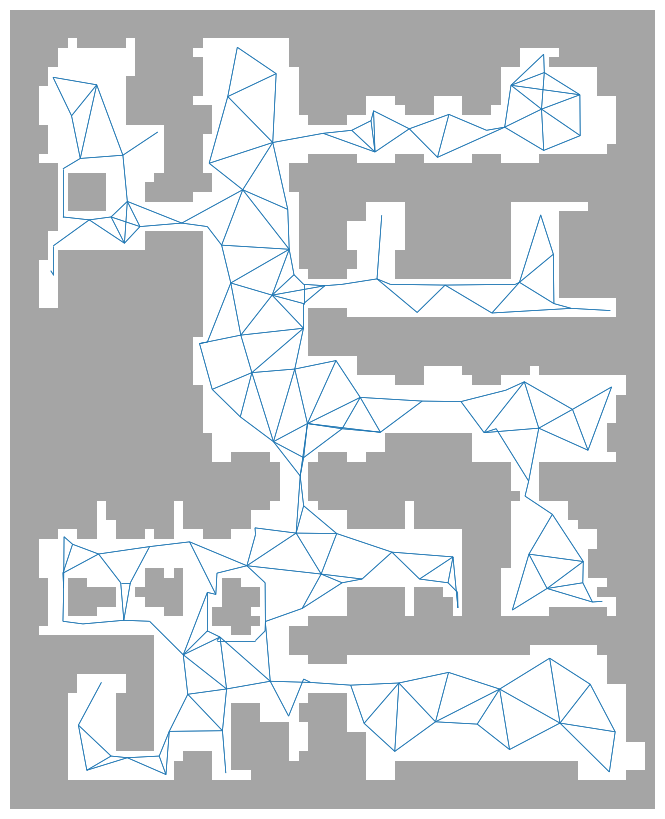}
        \caption{Density $\density=0.05$, with $|\Vertices|=223$ and $|\Edges|=824$.}
        \label{fig:Roadmap_construction:densities:1}
    \end{subfigure}
    \hfill
    \begin{subfigure}[t]{0.31\textwidth}
        \centering
        \includegraphics[width=\linewidth]{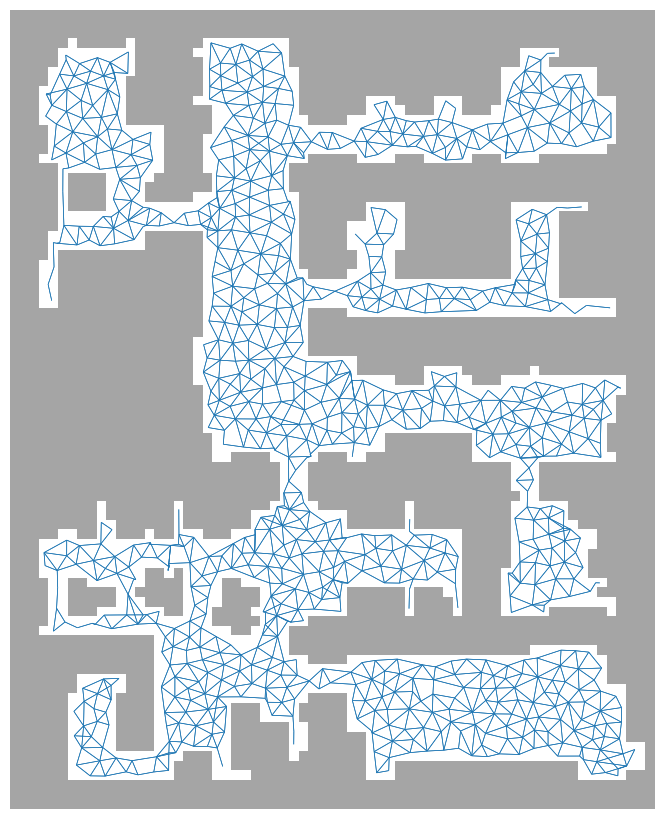}
        \caption{Density $\density=0.5$, with $|\Vertices|=1\;607$ and $|\Edges|=7\;894$.}
        \label{fig:Roadmap_construction:densities:2}
    \end{subfigure}
    \hfill
    \begin{subfigure}[t]{0.31\textwidth}
        \centering
        \includegraphics[width=\linewidth]{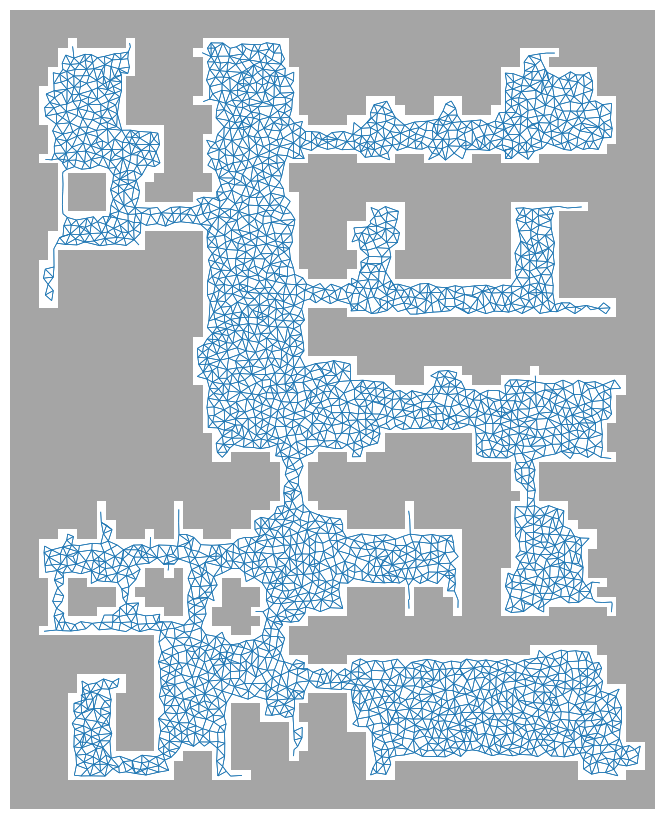}
        \caption{Density $\density=2.0$, with $|\Vertices|=4\;582$ and $|\Edges|=24\;442$.}
        \label{fig:Roadmap_construction:densities:3}
    \end{subfigure}
    \caption{The MovingAI gridmap \MovingAIMap{Den312d}, with $|\Vertices|=2\;445$ and $|\Edges|=8\;782$, sampled with various densities. The remaining parameters are $\clearance=0.5$, $\stretchfactor=1.5$, and $\maxlength=4\separation$.}
    \label{fig:Roadmap_construction:densities}
\end{figure}
From the target density $\density$ we derive a minimum separation $\separation = \sqrt{0.8/\density}$ between sampled vertices.
Placing vertices on a regular grid with spacing $\sqrt{1/\density}$ would achieve exactly the target density $\density$;
however, rejection sampling under a minimum-separation constraint yields a lower density than a regular grid at the same separation, so we shrink the separation below $\sqrt{1/\density}$ to partially compensate, using the empirically tuned factor $0.8$.
The region construction and the routes between vertices (below) are computed on a refined grid obtained by subdividing each map cell;
every clearance test itself remains an exact geometric test of a point or segment against the original map cells.
The subdivision factor is found by rounding $\sqrt{2}/\separation$ up to the nearest odd integer value, so that the diagonal of a refined cell is shorter than $\separation$ and no two sampled vertices can therefore land in the same refined cell, letting each vertex seed its own region.
The refinement is additionally chosen odd so that every map cell's centre remains at the centre of a refined cell;
with an even factor, refinement can shift the fine grid off the centreline of a corridor which may hold the only traversable line through it.
Refining relative to the density $\density$ also removes any upper bound on how dense a roadmap can be requested, independent of the original map resolution.

A fine cell is admissible if its centre lies at least $\clearance$ from every blocked cell.
We compute the set of admissible fine cells once and restrict it to its largest 4-connected component, since this is the largest region a disc of radius $\clearance$ can move through without leaving $\freeCells$;
sampling outside this component would place vertices an agent of radius $\clearance$ could never reach from the rest of the map.

Vertices are placed by rejection sampling.
A candidate refined cell is drawn at random from the admissible region, and a point is drawn uniformly at random within it.
The point is accepted only if it lies at least $\separation$ from every previously accepted vertex;
since only the cell's centre was certified admissible, the point must additionally, and independently, lie at least $\clearance$ from every blocked cell.
If both conditions hold, the point is added to $\Vertices$.
Sampling stops once $500$ consecutive candidates in a row have been rejected, rather than continuing until the admissible region is fully saturated;
the achieved density therefore falls somewhat short of $\density$ even after the $0.8$ correction in the separation, which only partially compensates for this.
The result is an even, non-clustered covering of $\freeCells$.

To connect the vertices, a breadth-first search over 4-adjacent admissible cells is run simultaneously from every vertex's cell, with each step claiming an unclaimed admissible cell;
every admissible cell thereby ends up owned by the vertex reachable in fewest steps, a discretized Voronoi diagram of $\Vertices$ over the admissible region under this graph distance.
Two vertices are neighbours when their regions contain adjacent cells.
For each neighbouring pair $\tuple{u, v}$, we take the shortest route, by number of cells, among those passing through a boundary cell between their two regions, obtained by following the search trees back from that boundary cell to $u$ and to $v$;
this route is generally longer than the true shortest clearance-respecting path between $u$ and $v$.
The route, a dense sequence of fine-grid positions, is then reduced to a small number of straight segments by repeatedly jumping to the furthest later position still reachable by a segment that keeps clearance $\clearance$ throughout;
this greedy reduction is not guaranteed to be minimal, since visibility along the route need not be monotone, for instance when rounding a corner into open space lets a later position be clear while an intermediate one is not.
Where $u$ and $v$ have such a clear segment directly between them, the reduction collapses to a single edge $\edge = \tuple{u, v} \in \Edges$;
otherwise, the retained turning positions, each lying at the centre of a fine cell, become additional vertices of $\Vertices$, and the connection becomes a short chain of edges bending around the obstacle, with a turning position shared by several routes becoming a single shared vertex.
The reduction is guaranteed to succeed for every interior segment of the route, since consecutive fine-grid positions along it are adjacent and the segment between two admissible adjacent positions is itself clear;
this argument does not cover the first and last segments, which run from a sampled vertex, at an arbitrary position within its cell, to that cell's centre.
If the reduction fails on one of these two endpoint segments, the corresponding connection is dropped, and the roadmap's connectivity is instead checked after construction.

The graph produced so far is purely local and can force long detours between vertices that are close in a straight line but far apart along $\Graph$.
We consider candidate pairs $u, v \in \Vertices$ whose straight-line distance is at most $\maxlength$, processed in increasing order of that distance.
For each candidate, if its distance in the current graph, that is $\Graph$ together with all shortcuts accepted so far in the sweep, exceeds $\stretchfactor$ times the straight-line distance between $u$ and $v$, and the straight segment between them is clear, we add a shortcut edge $\edge = \tuple{u, v}$.
Processing pairs in increasing order of straight-line distance means an accepted shortcut removes the need for longer shortcuts covering the same detour, while restricting candidates to those within $\maxlength$ from the outset, rather than testing and discarding longer pairs, ensures every shortcut remains local regardless of how large a detour it removes.

Every edge $\edge \in \Edges$ is traversable in both directions.
Clearance is guaranteed by construction: every vertex and edge keeps clearance $\clearance$ from every blocked cell, independent of the parameters.
Connectivity of $\Graph$ is not guaranteed in the same sense;
it is instead verified after construction, and a parameter setting that fails to yield a connected roadmap, for instance too low a density or too large a clearance relative to the map, is rejected and reported rather than returned.
The stretch factor $\stretchfactor$ and the target density $\density$ jointly control the size and quality trade-off: values of $\stretchfactor$ near $1$ and higher $\density$ yield roadmaps whose distances closely approximate true free-space distances at the cost of more vertices and edges, while looser values yield sparser graphs on which paths detour more.

% Density $0.1$, clearance $0.5$, stretch $1.5$, speed $1$
% \begin{figure}[htbp]
%     \begin{subfigure}[t]{0.24\textwidth}
%         \centering
%         \includegraphics[width=\linewidth]{figures/BenchmarkMaps/empty-16-16_roadmap.pdf}
%         \caption{Empty-16-16\\24 vertices, 104 edges}
%         \label{fig:MovingAI_roadmaps:empty-16-16_roadmap}
%     \end{subfigure}
%     \hfill
%     \begin{subfigure}[t]{0.24\textwidth}
%         \centering
%         \includegraphics[width=\linewidth]{figures/BenchmarkMaps/maze-32-32-2_roadmap.pdf}
%         \caption{Maze-32-32-2\\133 vertices, 236 edges}
%         \label{fig:MovingAI_roadmaps:maze-128-128-2_roadmap}
%     \end{subfigure}
%     \hfill
%     \begin{subfigure}[t]{0.24\textwidth}
%         \centering
%         \includegraphics[width=\linewidth]{figures/BenchmarkMaps/room-32-32-4_roadmap.pdf}
%         \caption{Room-32-32-4\\223 vertices, 602 edges}
%         \label{fig:MovingAI_roadmaps:room-64-64-8_roadmap}
%     \end{subfigure}
%     \hfill
%     \begin{subfigure}[t]{0.24\textwidth}
%         \centering
%         \includegraphics[width=\linewidth]{figures/BenchmarkMaps/den312d_roadmap.pdf}
%         \caption{Den312d\\223 vertices, 824 edges}
%         \label{fig:MovingAI_roadmaps:den312d_roadmap}
%     \end{subfigure}
    
%     \caption{Roadmaps sampled from MovingAI maps}
%     \label{fig:MovingAI_roadmaps}
% \end{figure}

\subsection{Scenario Generation}

Given a graph $\Graph = \tuple{\Vertices, \Edges}$ and an agent model, we construct a scenario as follows.
Because roadmap vertices are samples in continuous space, distinct vertices do not guarantee non-conflicting agents: two agents at different vertices may still occupy overlapping space, and two agents may be placed so that neither can move without conflicting with the other.
Scenario generation must exclude both cases, and other similarly unsolvable configurations.

Two vertex sets are derived from $\Graph$ and the agent model prior to sampling, and reused across all scenarios on that graph.
A vertex qualifies as a start if it is waitable (an agent may remain there indefinitely) and lies in the graph's largest strongly connected component (SCC).
We require the SCC condition rather than only a path from each start to its own goal, because an agent must reach its goal not only from its start but from any vertex it may be displaced to while yielding to another agent;
the SCC condition guarantees this, a per-pair path check does not.
Although this excludes some vertices that admit a feasible start-goal path on their own, such a path offers no guarantee once an agent must yield to another: none of the checks introduced below examine intermediate positions, so an unrecoverable vertex would go undetected until execution.
 
A vertex qualifies as a goal only if removing it, together with everything its permanent occupation would put out of reach, leaves the remaining admissible vertices mutually reachable.
This follows from treating a reached goal as held for the remainder of the scenario: an agent that lets others pass before occupying a cut vertex still severs the graph, only later.
The condition is stricter than necessary, since it excludes a vertex regardless of whether any agent in a given instance would cross it, but it is the only safeguard against this failure mode.
 
For each admissible vertex, we also record the vertices and edges unavailable to a second agent while it is occupied, determined again by spatial overlap.

Goals are drawn first, by shuffling the admissible goal vertices and accepting each in turn unless it conflicts with an accepted goal, until the requested count is reached or the pool is exhausted (failure).
Starts are drawn second, one per goal in draw order: the admissible start vertices (the full set, not only those also valid as goals) are shuffled once and scanned per agent, accepting the first vertex that does not conflict with an accepted start, and is not that agent's goal.
A vertex rejected for one agent may be accepted for another, so start-goal pairing need not follow the shuffle order;
exhausting the pool for any agent fails the attempt.

A scenario generation attempt failing any of the following conditions is discarded and redrawn, up to a redraw budget.
Starts must be pairwise non-conflicting, as must goals.
However, a start and goal of different agents may coincide, since the former agent departs before the latter arrives.
If there exists some group of agents where every agent's outgoing edges are blocked by other agents in the group, then that scenario is discarded.
The mirror condition applies to incoming edges: a scenario is also discarded if some group of agents cannot make their final approach, with every incoming edge to each member's goal blocked by another member of the group already there.

All scenarios within a set (on the same graph) are kept at a common agent count.
Scenario seeds within a set are drawn independently and distinctly, so scenarios within and across sets do not repeat.

\printbibliography
% \balance

\end{document}